\documentclass[10pt,journal,compsoc]{IEEEtran} 				
\pdfoutput=1

\ifCLASSOPTIONcompsoc

\else

\fi
\ifCLASSINFOpdf
\else
\fi

\usepackage{url}
\usepackage{color}
\usepackage{amsmath}
\usepackage{amsthm}
\usepackage{algorithmicx}
\usepackage{mathtools}
\usepackage{algorithm}
\usepackage[noend]{algpseudocode}
\usepackage{wasysym}
\usepackage{graphicx}
\usepackage{booktabs}

\usepackage{grffile}

\usepackage{multirow}

\usepackage[bottom,ragged,splitrule]{footmisc}

\usepackage{fancyhdr}

\usepackage{relsize}

\usepackage{subfig}

\usepackage{varioref}

\usepackage{flushend}

\newcommand{\numberOfTasks}{v}
\newcommand{\numberOfResources}{p}
\newcommand{\numberOfTypesOfResources}{\mathcal{P}}

\newcommand{\currentTask}{t_i}
\newcommand{\currentProcessor}{p_j}
\newcommand{\currentParent}{t_k}
\newcommand{\parentsProcessor}{p_l}

\newcommand{\numberOfEdges}{e}

\newcommand{\DParray}{CEFT}
\newcommand{\costTable}{C_{comp}}
\newcommand{\comm}{C_{comm}}
\newcommand{\communicationSetupTime}{\mathcal{L}}

\newtheorem{mydef}{Definition}
\newtheorem{lemma}{Lemma}

\newcommand{\ul}{\underline}

\newcommand*\Let[2]{\State #1 = #2}
\algrenewcommand{\alglinenumber}[1]{\color{red}\footnotesize#1:}

\newcommand{\csubfloat}[2][]{%
	\makebox[0pt]{\subfloat[#1]{#2}}%
}
\newcommand{\centerhfill}[1][\quad]{\hspace{\stretch{0.9}}#1\hspace{\stretch{0.9}}}

\newcommand{\longversion}{}

\begin{document}
\title{Mutual Inclusivity of the Critical Path and its Partial Schedule on Heterogeneous Systems}

\author{Aravind Vasudevan and David Gregg%
\IEEEcompsocitemizethanks{\IEEEcompsocthanksitem The authors are with the School of Computer Science and Statistics, Trinity College Dublin, Dublin 2\protect\\
E-mail: \{vasudeva, dgregg\}@scss.tcd.ie}
\thanks{}}

\markboth{}%
{Vasudevan and Gregg: Mutual Inclusivity of the Critical Path and its Partial Schedule on Heterogeneous Multi Processors}

\IEEEcompsoctitleabstractindextext{%
\begin{abstract}
	
The critical path of a group of tasks is an important measure that is commonly used to guide task allocation and scheduling on parallel computers. The critical path is the longest chain of dependencies in an acyclic task dependence graph. A problem arises on heterogeneous parallel machines where computation and communication costs can vary between different types of processor. Existing solutions for heterogeneous machines attempt to estimate the critical path using average values of computation and communication costs. However, this ignores opportunities to match specific tasks to specific classes of processor and communication links, and can result in quite misleading paths being identified as critical. We argue that an accurate critical path must consider the mapping of tasks to classes of processor and communication links. We formulate a polynomial time algorithm to find such a critical path. Our Critical Earliest Finish Time (CEFT) algorithm finds both the length of the critical path and an allocation of tasks to processors on that path. We compared CEFT experimentally to existing approaches such as averaging execution times across processors. The latter approach fails to accurately model the execution cost of tasks, and as a result fails to identify a correct critical path in 83.99\% of cases in our experiments. We also adapted a critical path-oriented scheduling algorithm (CPOP) to use our critical path algorithm and found that the resulting schedules are faster.

\end{abstract}

}

\maketitle

\IEEEdisplaynotcompsoctitleabstractindextext

\IEEEpeerreviewmaketitle

\section{Introduction}
\label{sec:paper-crit-path-intr}

\IEEEPARstart{S}{cheduling} of tasks onto resources is one of the most fundamental problems in parallel computing. A \textit{critical path} is the longest chain of dependent tasks in a graph. It is impossible to execute the graph in less time than the length of the critical path, even with infinite resources. Many successful scheduling algorithms~\cite{kohler1975preliminary, topcuoglu2002performance} prioritise tasks on the critical path.

Although the critical path is well defined when the computing resources are all of the same type, a problem arises on heterogeneous parallel computers. A given task usually has a different execution time depending on the type of computing resource upon which it is executed. For example, a heterogeneous parallel machine may consist of a small number of powerful processors that execute tasks quickly, and a larger number of low-power processors that are more energy efficient. Without a fixed execution time for each task, the critical path is poorly defined\footnote{However, if there is no communication cost, the problem is simple to solve. One can simply form a new graph by placing every task on a processor that minimizes its finish time and calculate the longest path in this resultant graph.}. In addition, communication time between processors typically varies in real machines. For example, communicating between cores on a single chip usually has a low cost, whereas communicating over the network in a cluster is much slower and hence has a higher cost.

\ifdefined\longversion
Existing algorithms to compute the critical path of a graph for heterogeneous machines make simplifying assumptions. One simple strategy is to take each of the execution times of a given task on various processors and average them~\cite{kwok1996dynamic}. Another, is to assume that all tasks on the critical path are executed on a single processor, and to simply choose the processor that minimizes the critical path length \cite{topcuoglu2002performance}.
\fi

In this paper we propose a definition of the critical path for heterogeneous processors that is much closer to the intuitive idea of the shortest possible execution time, based purely on dependencies.  A practical problem with our definition is that for each task we need to consider all of the different processors that it can be executed on. This could potentially lead to an exponential search space. However, we provide a dynamic programming algorithm that can consider all possible allocations in polynomial time.

Our main contributions are:
\begin{itemize}
	\item We propose a new definition of the critical path for tasks graphs on parallel computers with heterogeneous execution and communication times.
	\item We provide a novel dynamic programming algorithm for finding our critical path (Critical Earliest Finish Time (CEFT)) in $\numberOfTypesOfResources^2\numberOfEdges$ time, where $\numberOfTypesOfResources$ is the number of classes of processors and $\numberOfEdges$ is the number of edges in the task graph. We evaluate our new approach and find that quality of the critical paths are shown to be better than those produced by the state of the art CPOP with our critical paths being shorter in most cases.
	\item We also extend our critical path finding algorithm into a DAG scheduling algorithm (\textit{CEFT-CPOP}) by replacing the path found by our algorithm (with its corresponding partial assignment) into the CPOP algorithm. Our experiments suggest that our algorithm (CEFT-CPOP) produces smaller makespans in 15.9\%, 75.94\%, 90.29\% and 89.69\% of the experiments in four different models of parallel workloads.
\end{itemize}

\section{Background}
\label{sec:paper-crit-path-pape-back}
Critical paths have long been used to guide heuristic algorithms for scheduling directed acyclic graphs (DAGs) of tasks~\cite{hu1961parallel,lockyer1969introduction,kohler1975preliminary}. A task graph is a DAG where the vertices represent computational tasks and the edges represent data dependencies between the tasks. When scheduling for homogeneous parallel architectures, weights are assigned to vertices of the graph to represent the execution time of the task, and to the edges to represent the communication cost of data flow between tasks. We consider the problem of static scheduling of task graphs, where the execution time and communication costs can be estimated with some accuracy prior to scheduling. In common with existing literature on static scheduling, we do not consider the case of computation or communication costs that are strongly dependent on unpredictable data inputs or run-time artifacts.

The conventional definition of the critical path of a DAG in static task scheduling is \textit{the longest path of from the entry node to the exit node in the task graph}. On homogeneous parallel computers the standard algorithm for finding the longest path in a DAG can be used to find the critical path. However for heterogeneous parallel architectures, finding the critical path is more difficult. The execution time of each task can be quite different, which means that the dependence length of a path through the graph depends not just on the tasks in the path, but also on which of the heterogeneous processors each task is allocated.

Two approaches are commonly used to estimate the critical path in scheduling algorithms for heterogeneous architectures. Popular scheduling approaches such as HEFT and CPOP~\cite{topcuoglu2002performance} assign an execution cost to each task that is the average execution time of the task on each processor. They also assume that the communication cost between pairs of processors depends purely on the quantity of transferred data, so that costs between different pairs of processors are all the same. These two assumptions provide a single execution cost for each task and a single communication cost for each edge, so that the standard critical path algorithm for homogeneous architectures can be used. However, in cases where task execution times differ widely on different processors, this can be inaccurate. For example, a GPU might be an order of magnitude faster than a CPU on an array processing task, but absolutely hopeless for single-threaded code. The average of CPU and GPU execution on each task may be a multiple of the execution cost on the best architecture for that task.

A second common approach is to assume that the \textit{entire} task graph will be executed on a single type of processor, and use the execution costs for that type of processor to compute a critical path~\cite{daoud2008high}. A heterogeneous machine may contain several different types of processor, and a critical path will need to be computed for each type. However, selecting the processor that results in the shortest critical path gives some estimate of the minimum possible time needed to execute the code. This heuristic may work well on heterogeneous machines where some processors are simply more powerful versions of others. However, where different types of processors --- such as CPUs and GPUs --- are suited to different types of tasks, this approach may result in an estimate of the critical path length that is much longer than can be achieved by choosing the best processor for each task.

Neither of these approaches to estimating the critical path is entirely satisfactory. Both can result in very misleading paths being identified as critical, as we show in section~\ref{sec:paper-crit-path-expe-resu}. A better heuristic for identifying the critical path should take into account the different execution times of tasks on different processors and the cost of data communication between processors.


\section{Defining a critical path for heterogeneous processors}
\label{sec:paper-crit-path-form}


As we saw in the previous section, identifying the critical path of a task graph on a heterogeneous parallel computer is complicated because the execution and communication times vary depending on the allocation of tasks to processors. Existing approaches use simplifying assumptions, such as taking the average execution time of a task across all processors.

For a restricted version of the problem, we can find a more accurate estimate than previous approaches. If we assume that communication costs between processors are zero, then for each task we simply choose the processor allocation that minimizes the execution time of the task. This provides us with a single minimal execution time for each task, allowing us to find the critical path using the standard algorithm for homogeneous architectures. This same approach can also be used if we make the same assumption about communication costs as Topcuoglu et al.~\cite{topcuoglu2002performance}: the communication cost is the same irrespective of the source and destination processor, even if the source and destination are the same processor. Where communication costs are entirely independent of the allocation of tasks, we can simply allocate each task to the processor that minimizes its execution time. Although this approach is no more complex than Topcuoglu et al.'s and is likely more accurate, to our knowledge we are the first to propose it.

However, this assumption that communication costs are independent of processor allocation of tasks is unsatisfactory for two main reasons. First, on large scale parallel computers communicating with a nearby processing element is typically faster than sending the same data over a long distance. Second, there is an important case where two tasks are allocated to the same processor, which can eliminate communication costs entirely. A critical path for heterogeneous parallel computers should ideally take account of heterogeneous communication costs as well as computation costs. However, when we consider both the costs together, there is no longer a simple strategy for choosing the best allocation of tasks to processors that will minimize the critical path. There is an exponential number of possible allocations of processors to tasks. In section~\ref{sec:paper-crit-path-crit-path-pape-dyna-prog-solu} we present a polynomial time dynamic programming algorithm that finds a critical path considering both heterogeneous computation and communication costs. First, we present several definitions to formalize the problem.

\subsection{Definitions}
\label{sec:paper-crit-path-pape-defi}
Table~\ref{tab:nota} gives a list of notations and their descriptions that we will be using for the rest of this paper. A task graph is a weighted directed acyclic graph $G_t(V_t,E_t)$, such that each vertex $V_t$ is a program statement(s) or kernel in the application, and \mbox{$E_t \subseteq (V_t \times V_t)$}, represents the communication edges between the vertices. The system resources (processor graph) are represented by a weighted undirected graph $G_r(V_r,C_r)$. Where $V_r$ represents a processing element in the underlying processor graph, while the edge $C_r \subseteq (V_r \times V_r)$, represents the communication links. For the sake of simplicity in the rest of this paper we will refer to a task $v_t \in V_t$ as $t_i$ and a processor $v_r \in V_r$ as $p_j$.

\ifdefined\longversion

\begin{mydef}
	We define an \textbf{assignment} of a task as it being mapped on to a processor for execution. It differs from a schedule in the sense that we do not have to specify order of execution as we deal with one task and one processor. The \textbf{assignment} of a path is set of individual assignments for all the tasks in the given path\footnote{In the rest of this paper, we use \textbf{assignment} and \textbf{mapping} interchangeably.}.
\end{mydef}

\begin{table}[t!]
	\begin{center}
\resizebox{\linewidth}{!}{%
	\begin{tabular}{|c|c|}
		
		\hline \textbf{Notation} & \textbf{Description} \\ 
		\hline \DParray & The dynamic programming array \\ 
		\hline $\numberOfResources$ & Number of processors \\ 
		\hline $\numberOfTasks$ & Number of tasks \\ 
		\hline $\numberOfEdges$ & Number of edges \\ 
		\hline $\mathcal{P}(t_i)$ & Parents of the task $t_i$ \\ 
		\hline $\costTable(t_i, p_j)$ & Execution time of task $t_i$ on processor $p_j$ \\
		\hline $\communicationSetupTime(p_i)$ & Communication startup time of processor $p_i$\\
		\hline 
	\end{tabular} 
}
	\caption{List of notations}
	\label{tab:nota}
	\end{center}
\end{table}

\begin{mydef}
	$\mathcal{P}(t_i)$ denotes the set of parents (commonly referred to as the set of immediate predecessors in the literature) of task $t_i$ in the given DAG. Any task that has no parent is called a \textit{source task} or an \textit{entry task}. Similarly, any task that doesn't have any children are identified as a \textit{leaf node} or \textit{exit task}.
\end{mydef}
\else

\begin{table}[t!]
	\begin{center}
		\resizebox{\linewidth}{!}{%
			\begin{tabular}{|c|c|}
				
				\hline \textbf{Notation} & \textbf{Description} \\ 
				\hline $\costTable(t_i, p_j)$ & Execution time of task $t_i$ on processor $p_j$ \\
				\hline $\communicationSetupTime(p_i)$ & Communication startup time of processor $p_i$\\
				\hline \DParray & The dynamic programming array \\ 
				\hline $\numberOfResources$ & Number of processors \\ 
				\hline $\numberOfTasks$ & Number of tasks \\ 
				\hline $\numberOfEdges$ & Number of edges \\ 
				\hline $\mathcal{P}(t_i)$ & Parents of the task $t_i$ \\ 
				\hline $EST(t_i, p_j)$ & Earliest start time of task $t_i$ on processor $p_j$\\ 
				\hline $EFT(t_i, p_j)$ & Earliest finish time of task $t_i$ on processor $p_j$\\ 
				\hline 
			\end{tabular} 
		}
		\caption{List of notations}
		\label{tab:nota}
	\end{center}
\end{table}

\fi

\begin{mydef}
	We define \textbf{communication cost} between task $t_k$ on processor $p_l$ and task $t_i$ on processor $p_j$ as,
	\begin{equation}
	\label{eq:crit-path-comm-cost}
	\comm(\{t_k, p_l\}, \{t_i, p_j\}) \\ = \begin{cases}
	\communicationSetupTime(p_l) + \dfrac{data_{t_k, t_i}}{c_{p_l, p_j}} & \text{, $p_j \neq p_l$} \\
	0 & \text{, $p_j = p_l$}
	\end{cases}
	\end{equation}
	\noindent where $\communicationSetupTime(p_l)$ is the setup time associated with a processor every time it has to send data; $data_{t_k, t_i}$ represents the data size that has to be sent from task $t_k$ to task $t_i$. Similarly, $c_{p_l, p_j}$ is the bandwidth between processor $p_l$ and $p_j$.
\end{mydef}
%
%

\noindent The critical-path of a DAG is conventionally~\cite{arabnejad2014list} defined as following,
\begin{mydef}
	\label{def:crit-path-defi}
	\textbf{Critical-Path} (CP) of a DAG is the longest path from the entry node to the exit node in the application graph. The minimum critical-path length ($CP_{MIN}$) is computed by considering the minimum computational costs of each task in the critical path.
\end{mydef}

This definition of a critical-path suggests that the critical-path is a property \textit{only} of the application DAG. Although this holds true in the homogeneous setting, it breaks down in heterogeneous setting. We examine this through the following lemmas and their proofs.

\begin{lemma}
	\label{lemm:crit-path-inde}
	The critical path cannot be independent of its mapping to the processors in a heterogeneous parallel machine
\end{lemma}
\begin{proof}
	In a heterogeneous setting, the execution time of tasks in the DAG is given by $\costTable(t_i, p_j)$ which is a cell in a two-dimensional matrix of order $\numberOfTasks\times\numberOfResources$. This implies that we cannot reduce the vector $\costTable(t_i)$ into a scalar value and hence the vertex weights of the DAG cannot be known a priori as the weights do not exist independent of a mapping for the tasks onto processors. Weights cannot be given a single value independent of a mapping, and hence the critical path can not exist independent of its mapping.
\end{proof}

%
%
%

\subsection{Defining the critical path}
In section~\ref{sec:paper-crit-path-pape-back}, we saw that existing definitions of the critical path are inadequate for the heterogeneous execution setting. The definition of critical-path, which was valid in the homogeneous setting, is being used to estimate the critical path in the heterogeneous setting. In this section, we will attempt to redefine the critical-path for the newer setting.

\ifdefined\longversion

Traditionally, the \textit{earliest finish time} of a task is the earliest time at which it can finish under a legal partial schedule and is defined as follows.

\begin{mydef}
	Earliest Start Time (EST) is defined as the earliest time in the schedule at which a given task $t_i$ can start
	\begin{equation}
	\label{eq:crit-path-est}
	EST(t_i, p_j) = max(avail[j], \max_{t_m \in pred(t_i)}(AFT(t_m) + c_{m,i}))
	\end{equation}
\end{mydef}

\begin{mydef}
	Consequently, Earliest Finish Time (EFT) is defined as the earliest time at which the task $t_i$ can finish,
	\begin{equation}
	\label{eq:crit-path-eft}
	EFT(t_i, p_j) = EST(t_i, p_j) + w(t_i, p_j)
	\end{equation}
	\noindent where $w(t_i, p_j)$ is the execution time of task $t_i$ on processor $p_j$.
\end{mydef}

Although this definition conveys the meaning of the earliest start time and end times, they are not adequate when it comes to defining the earliest start and finish times of tasks when calculating the critical-path. Several pieces of literature~\cite{kwok1996dynamic,shi2006scheduling,arabnejad2014list} have redefined the earliest finish time to suit their needs. In order to define the critical-path in a heterogeneous setting, we also redefine the start and finish times of tasks.

\fi

\begin{mydef}
	A \textbf{Critical-Path} (CP) is the longest path in the DAG when it has a corresponding optimal partial assignment
\end{mydef}

This definition of the critical-path stems from Lemma~\ref{lemm:crit-path-inde} which states that the CP cannot be independent of its partial assignment. Hence, we define our CP as the path that has the longest path length when all the tasks in that path are mapped in the most effective way possible.

\begin{mydef}
	From these two restrictions imposed by the new definition of the CP, we redefine our earliest finish time as Critical Earliest Finish Time (CEFT) :
	\begin{multline}
	\label{eq:crit-path-ceft}
	\DParray(t_i, p_j) = \\\textbf{max}_{T_k \in \mathcal{P}(t_i)} \{ \textbf{min}_{p_l \in \{0, \cdots, \numberOfResources-1\}} \{\costTable(t_i, p_j)\\ + \DParray(T_k, p_l)+ comm(\{t_i, p_j\}, \{T_k, p_l\})\} \}
	\end{multline}
\end{mydef}

Note that the above definition~\ref{eq:crit-path-ceft} is entirely satisfactory in two circumstances: when we consider the length of a single path within the task graph and when tasks can be duplicated onto multiple processors. When we consider multiple paths through the graph simultaneously and task duplication is not possible, the problem becomes more difficult. In the next section we deal with the case where paths are considered in isolation or tasks may be duplicated.
\section{A dynamic programming solution}
\label{sec:paper-crit-path-crit-path-pape-dyna-prog-solu}

\begin{algorithm}[t!]
	\scriptsize
	\caption{Identify \& map the critical-path of a given DAG onto a set of processors}
	\label{algo:map-path}
	\begin{algorithmic}[1]
		\Require Given application graph is a DAG and vertex IDs are in topological order
		\Function{FindAndMapCriticalPath}{$G_t, G_r$}
		\For {$\currentTask=0 \cdots \numberOfTasks$} 
		\If {current task is a source task}
		\State Set \DParray($\currentTask$, $\currentProcessor$) as the execution time of $\currentTask$ on each processor
		\Else
		\For {$\currentProcessor=0 \cdots \numberOfResources$}
		\For {$\currentParent \in pred(\currentTask)$}
		\For {$\parentsProcessor=0 \cdots \numberOfResources$}
		\State $\currentTask$ is the current task under investigation
		\State $\currentProcessor$ is the current processor to which $\currentTask$ is mapped currently
		\State $\currentParent$ is a parent of $t_i$
		\State $\parentsProcessor$ is the processor to which $\currentParent$ is mapped currently
		\Let{$commCost$}{$comm(\{\currentParent, \parentsProcessor\},\{t_i,p_j\})$}
		\Let{$compCost$}{$\costTable(\currentTask,\currentProcessor)+\DParray(\currentParent,\parentsProcessor)$}
		\Let{$totalCost$}{$compCost+commCost$}
		\EndFor
		\State Choose the processor $\parentsProcessor$ that minimizes the EFT of $\currentParent$
		\EndFor
		\State From among these minimized choices of $(\currentParent, \parentsProcessor)$ pairs, choose the one that maximizes eq~\ref{eq:crit-path-ceft} and call it $(\currentParent^{max},\currentProcessor^{min})$
		\Let{$\DParray(\currentTask,\currentProcessor)$}{$totalCost$ belonging to the $(\currentParent^{max},\currentProcessor^{min})$ pair}
		\Let{$\DParray(\currentTask,\currentProcessor).path$}{$\DParray(\currentParent^{max},\currentProcessor^{min}).path$}
		\State $\DParray(\currentTask,\currentProcessor).path.push\_back((\currentTask,\currentProcessor))$
		\EndFor
		\EndIf
		\EndFor
		\For {$t_s \in listOfSinks$}
		\For {$p_s=0 \cdots \numberOfResources$}
		\Let{$cost$}{$\DParray(t_s,p_s)$}
		\EndFor
		\Let{$p_s^{min}$}{$p_s$ that minimizes $cost$}
		\EndFor
		\State From among these minimized costs, choose the task $t_s^{max}$ that maximizes the minimized cost
		\State The critical-path is the path represented by $\DParray(t_s^{max},p_s^{min}).path$
		\EndFunction
	\end{algorithmic}
\end{algorithm}

Our definition~\ref{eq:crit-path-ceft} of the critical earliest finish time (CEFT) includes an optimal mapping of tasks to processors, and allows us to define a more accurate critical path for heterogeneous architectures than previous approaches. However, there are an exponential number of allocations of tasks to processors, so any algorithm that considers all mappings individually will require exponential time. In this section, we present a dynamic programming approach that computes the length of a path in the dependence graph using our CEFT-based definition of dependence length. Using this approach we formulate a polynomial time algorithm that finds the CEFT-longest dependence path in the task graph. In the case where tasks can be duplicated to reduce communication costs, a longest path is also a critical path in the task graph.

Algorithm~\ref{algo:map-path} traverses the vertices of the task graph in topological order. The algorithm computes the critical earliest finish time ($CEFT(t_i,p_j)$) of each task, $t_i$, on each of the processors, $p_j$, in the machine. Note that where the machine contains groups of multiple identical processors (with the same computation and communication times) the entire group can be considered a single processor for the purposes of computing a critical path. A critical path computes a lower bound on the execution time of the task graph based solely on the length of dependence chains rather than on resource constraints. Therefore, having multiple instances of the same class of processor does not affect the critical path.

Where a task, $t_i$ has no predecessors (i.e. a source task) its $CEFT(t_i,p_j)$ on processor $p_j$ is simply the execution time of that task on the processor. Where a task has one or more parents in the task graph, the task cannot start executing until the predecessor task, $t_k$, has completed and its results have been communicated. However, we do not have a single critical earliest finish time for $t_k$. Instead we have a separate $CEFT(t_i,p_j)$ for each processor $p_j \in p$. To compute $CEFT(t_i, p_j)$, that is the earliest finish time of task $t_i$ on processor $p_j$ we consider each of the possible processor allocations, $p_l$, of the predecessor task, $t_k$. We select the allocation of $t_k$ to $p_l$ that results in the lowest value for $CEFT(t_i, p_j)$, taking into account the execution time of task $t_i$ on processor $p_j$ and the communication time between processor $p_l$ and $p_j$. Where tasks $t_k$ and $t_i$ are allocated to the same processor, we assume that communication costs are zero.

The resulting algorithm is significantly more complicated than the critical path algorithm for homogeneous architectures while having a higher time complexity. However, the time complexity is polynomial and we find the critical earliest finish time for each task without having to separately enumerate all possible allocations of all tasks.

The objective of our algorithm is to redefine the earliest finish time to more closely represent the intuitive idea behind the shortest possible execution time based on dependencies. To this end, we do not fix the allocation/mapping of a task to a processor as we progress through the algorithm by iterating over the dynamic programming array. We instead, compute where the parents of the current task should be mapped to minimize the $CEFT$ of the current task on the current processor. This is done in lines 16--18 in the algorithm. We set up our loops to be able to iterate over four variables : $t_i$ (current task), $p_j$ (current task's current processor), $t_k$ (current task's current parent) and $p_l$ (current parent's current processor). From lines 6--12, it is evident that for every combination of a current task and a specific processor for the current task ($t_i, p_j$) we examine all possible assignments of all possible parents and choose the set of assignments of the parents which minimizes the earliest finish time of task $t_i$ on processor $p_j$. 

We reiterate by stressing on the fact that this \textit{does not} fix the assignment for the current task on the current processor. It simply examines it. This also does not fix the assignment of any of the parents of task $t_i$. The algorithm only fixes the assignment of the parent that has led to the minimization of the earliest finish time of the task $t_i$, locally to the pair of $t_i, p_j$ under consideration\ifdefined\longversion\footnote{This gives us a major added benefit of incorporating lookahead features for much lesser complexity.~\cite{arabnejad2014list} claim to incorporate lookahead features into their algorithm by calculating an \ul{o}ptimistic \ul{c}ost \ul{t}able (OCT) for $O(\numberOfResources^2\numberOfEdges)$ complexity. The original lookahead based scheduling algorithm~\cite{bittencourt2010dag}, incorporates lookahead features at a much higher complexity of $O(\numberOfTasks^4\numberOfResources^3)$. We incorporate lookahead features into our algorithm without compromising as much complexity. We ensure tasks that are being inspected are not immediately assigned to a processor. Hence any a given task has a set of finish times that it could take potentially, based on the assignment of its parent and itself. Only when an all exit tasks have been visited and one is chosen, we fix the assignment of all the tasks in the CP. In doing so, we gain greater flexibility and the ability to find the right critical path.}\fi. 

This is reflected in the lines 19--20. Once the latest completing parent ($t_k^{max}$) has been identified, we copy the path information from this parent and append the pair ($t_i, p_j$) onto this path information. Lines 21--26 show how the critical-path can be fixed by examining the dynamic programming array ($\DParray$) entries of all the sink/exit nodes and identify $\DParray(t_s^{max},p_s^{min})$. This gives an idea of how the critical-path is always in a state of flux and is not fixed until the algorithm finishes. 

\begin{figure}[t!]
	\centering
	\includegraphics[scale=0.15]{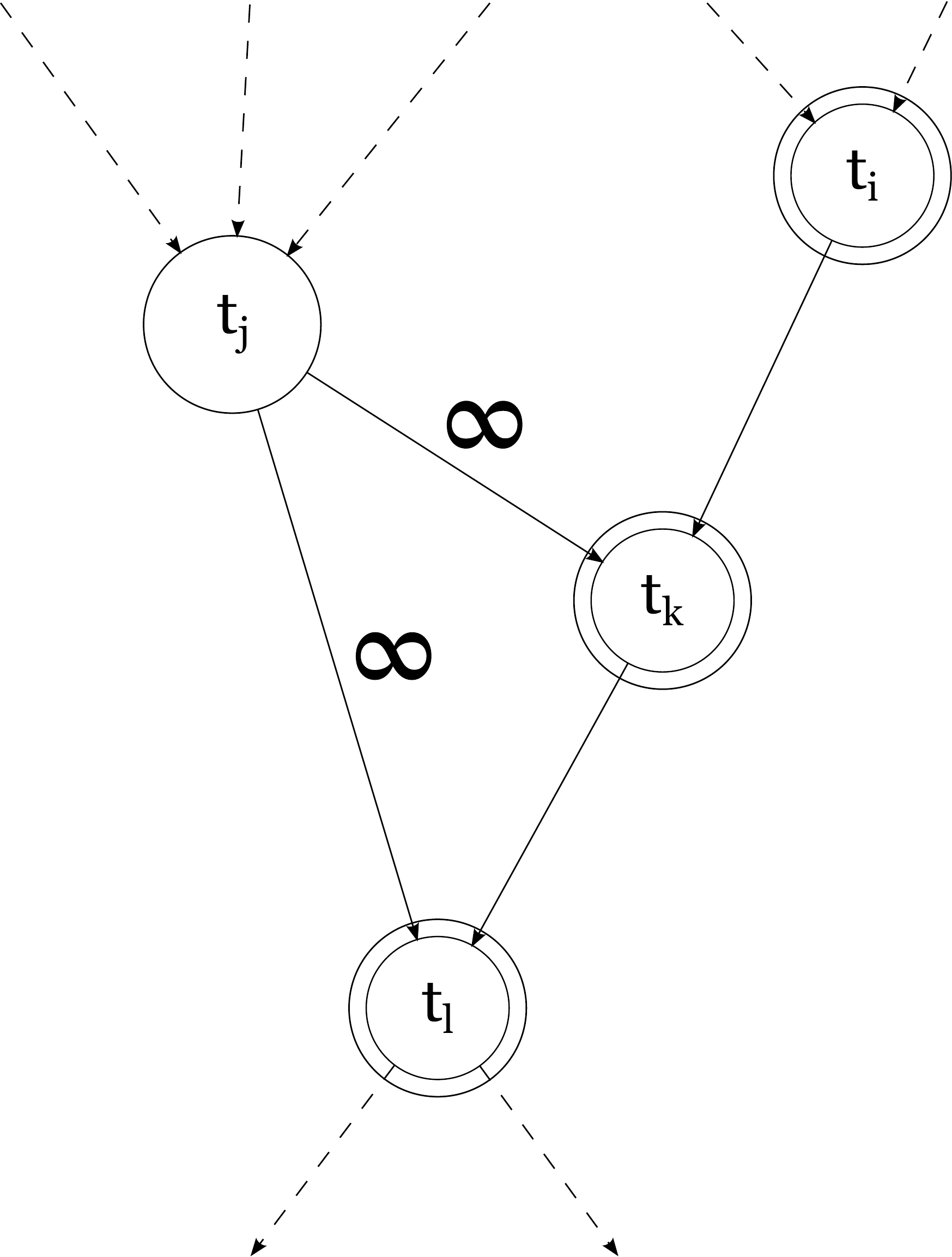}
	\caption{Section of a sample application graph}
	\label{fig:paper-crit-path-task-dupl-grap}
\end{figure}

\subsection{Task duplication}
\label{sec:paper-crit-path-task-dupl}

%

Our algorithm finds the critical path by identifying the longest dependence path through the task graph. We use our CEFT definition of path length, meaning the time required for computation and inter-processor communication, assuming an optimal allocation of the tasks on the path to various classes of processors. However, this definition of the longest path contains an important assumption. We assume that we can compute the length of two different paths independently. However, the length of a path depends on the allocation of the tasks in the path to processors. If a single task appears on multiple different paths, the different paths may require a different allocation of that task to minimize path length.


For instance, let us consider a section of an application graph shown in Figure~\ref{fig:paper-crit-path-task-dupl-grap}. Let us assume that all the tasks that are denoted using concentric circles are in the critical path of the application. For the sake of simplicity, let us assume that the amount of data to be transferred from task $t_j$ to tasks $t_k$ and $t_l$ is $\infty$. In this scenario when $t_k$ is being evaluated, its parent $t_j$ will be assigned to the same processor that $t_k$ gets assigned to and the same holds true in the case of $t_l$. When the final critical path is decided as we reach exit tasks in this application graph, there might be a situation where $t_k$ and $t_l$ might be assigned to different processors. If this situation arises, the task $t_j$ (even though it is not on the critical path) has to be assigned to two different processors to make sure that $t_k$ and $t_l$ stay on the critical path.


Many existing scheduling algorithms use task duplication~\cite{kruatrachue1988grain,ahmad1998exploiting} to reduce communication costs on heterogeneous parallel architectures. Where a parent task has more than one successor, it can sometimes improve the schedule to duplicate the parent so that identical copies of the task execute in parallel on two different processors. This can reduce communication time between processors, particularly on heterogeneous architectures where different tasks are often suited to different types of processor. Where task duplication is used, our Algorithm~\ref{algo:map-path} will compute a correct critical path in all cases. However, where the subsequent scheduling approach does not allow task duplication, our algorithm may result in an overly-optimistic critical path length. 

In the case where task duplication is not allowed, each task must be allocated to exactly one processor. As with the case where tasks may be duplicated, we must deal with two sets of costs when computing path lengths. We must consider the cost of executing a given task on a given processor, and the cost of the communication between tasks on each possible pair of processors. Where the communication costs between processors can vary arbitrarily, this is equivalent to the Partitioned Boolean Quadratic Problem (PBQP) which is known to be NP-complete \cite{Scholz:2002}. Thus, in the absence of task duplication, finding a critical path in a task graph on parallel architectures with heterogeneous execution times and communication costs is NP-complete.

\section{Complexity analysis}
\label{sec:paper-crit-path-crit-path-pape-comp-anal}
In this section we analyse the space and time complexity of the dynamic programming method from algorithm~\ref{algo:map-path} proposed in Section~\ref{sec:paper-crit-path-crit-path-pape-dyna-prog-solu}. The outermost loop in the algorithm runs from $t_i=0 \cdots \numberOfTasks$. This loop inspects all the tasks in the DAG. The second level loop inspects all possible processors for the current task $t_i$. This implies that, $p_j$ runs from $0 \cdots \numberOfResources$. For each \mbox{(task, processor)} pair, we need to inspect every parent of $t_i$ as the algorithm tends to fix the parents processors based on its child's requirements. This necessitates that $t_k$ runs from $0 \cdots pred(\currentTask)$. To fix the parents processor, we need to inspect all the processors again to see which processor for the parent gives the earliest start for the current child. Hence, $p_l$ runs from $0 \cdots \numberOfResources$. 

Then the complexity of the entire algorithm in the worst case of all the upper limits of these nested loops is: $\numberOfTasks\times\numberOfResources\times n_{pred(\currentTask)}\times\numberOfResources$. $n_{pred(\currentTask)}$ is the number of parents for any given task. In the general case, this can be assumed to be the average in-degree of the application DAG. The average in-degree of a DAG can be further written as $\numberOfEdges/\numberOfTasks$. Hence the complexity of the algorithm can be simplified as $O(\numberOfResources^2\numberOfEdges)$. In the worst case, where the DAG is a fully connected graph, the number of edges in the graph is equal to $\numberOfTasks^2$. In this case, the complexity of the algorithm is increased to $O(\numberOfTasks^2\numberOfResources^2)$ which is higher than the complexity of other list scheduling heuristics like HEFT and CPOP which is $O(\numberOfTasks^2\numberOfResources)$.

However, if processors can be divided into $\numberOfTypesOfResources$ processors (where processors in each class have identical computation and communication costs), then the algorithm only needs to deal with the number of such classes of processors rather than $\numberOfResources$. This is feasible as our algorithm is a critical path finding algorithm and hence doesn't need to keep track of availability of processors. When trying to map a task from the critical path all processors will be free. This greatly reduces the computational complexity of our algorithm from $O(\numberOfTasks^2\numberOfEdges)$ to $O(\numberOfTypesOfResources^2\numberOfEdges)$, where $\numberOfTypesOfResources$ is the number of types of processors.

The space complexity of the algorithm at first glance is $O(\numberOfResources\numberOfTasks)$ as the $\DParray$ is a two dimensional array of size $\numberOfTasks\times\numberOfResources$. But this can be further reduced by storing the path information of only a frontier that is moving down along the DAG. Since we incorporate the path information from the previous states into the current state, we can ignore the state information of all the $\DParray$ elements that have been absorbed into other $\DParray$ elements. This would in turn be a frontier that is moving down the DAG. Hence, the space complexity can be reduced down to $O(\beta\numberOfResources)$ where $\beta$ is the width parameter of the graph.

\section{From critical path to makespan: CEFT-CPOP}
\label{sec:paper-crit-path-from-crit-path-to-make}

\begin{algorithm}[t!]
	\scriptsize
	\caption{The critical path on a processor (CPOP) algorithm}
	\label{algo:cpop}
	\begin{algorithmic}[1]
		\Function{CPOP}{}
		
		\State Set the comp costs of tasks and comm costs of edges with mean values
		\State Compute $rank_u$ of tasks by traversing graph upward, starting from exit task
		\State Compute $rank_d$ of tasks by traversing graph downward, from entry task
		\State Compute $priority(t_i)=rank_d(t_i)+rank_u(t_i)$ for each task $t_i$ in the graph
		\State $|CP|=priority(t_{entry})$, where $t_{entry}$ is the entry task
		\State $SET_{CP}=\{t_{entry}\}$, where $SET_{CP}$ is the sect of tasks on the critical path
		\State $t_k \leftarrow t_entry$
		\While{$t_k$ is not the exit task}
		\State Select $t_j$ where (($t_j \in succ(t_k)$) and ($priority(t_i)==|CP|$))
		\State $SET_{CP}=SET_{CP} \bigcup \{t_j\}$
		\State $t_k \leftarrow t_j$
		\EndWhile
		\State Select the critical-path processor ($p_{cp}$), $p_{cp}$ minimizes $\sum_{t_i \in SET_{CP}} w_{i,j}$
		\State Initialize the priority queue with the entry task
		\While{there is an unscheduled task in the priority queue}
		\State Select the highest priority task $t_i$ from the priority queue
		\If{$t_i \in SET_{CP}$}
		\State Assign the task $t_i$ on $p_{cp}$
		\Else
		\State Assign the task $t_i$ to the processor $p_j$ which minimizes the $EFT(t_i, p_j)$
		\EndIf
		\State Update the priority-queue with the successors of $t_i$, if they become ready tasks
		\EndWhile
		\EndFunction
	\end{algorithmic}
\end{algorithm}

CEFT is a critical path finding algorithm for application DAGs on heterogeneous processors. We extend this critical path finding algorithm into a DAG scheduling algorithm for heterogeneous processors by incorporating the critical path obtained from CEFT into CPOP. We have cleverly named this \textit{CEFT-CPOP}.

Let us recall the CPOP algorithm from the brief discussions in section~\ref{sec:paper-crit-path-rela-work}. It is a critical path based list scheduling algorithm that calculates its critical path based on mean values of computation costs and communication costs as shown in line 2 of algorithm~\ref{algo:cpop}. In lines 3--5 the authors of CPOP calculate the priority function which orders the tasks according to their relative importance. The entry task is added to the CP and the graph is then traversed downward from the entry task. Then, a child $t_j$ of the entry task that has the same priority value as itself is added to the critical path. Consequently, $t_j$'s children are examined and the one that has the same priority value is added to the path and the algorithm continues until it reaches the exit task. This path is then assigned to a single processor $p_{cp}$ in line 13 of the algorithm, in an attempt to produce the smallest possible critical path length for the tasks in the critical path\ifdefined\longversion\footnote{As we have discussed before, we believe that the tasks in this path are faulty as they have been calculated based on average values.}\fi. Once the path has been assigned to the processor that minimizes the path length, a priority queue with the entry task in it is examined. The task that has the highest priority function value is popped out of this queue. If this path is part of the critical path calculated earlier, it is scheduled on $p_{cp}$, otherwise it is scheduled on the processor $p_j$ which minimizes the $EFT(t_i,p_j)$. If any of the successors of the task that was just scheduled are now ready to be scheduled, they are added onto the priority queue and the algorithm continues until all the tasks in the priority queue have been scheduled.

In order to extend our critical path finding algorithm into a scheduling algorithm, the only modification we make to the CPOP algorithm is one of finding the critical path. Hence, we remove lines 2 -- 13 of the CPOP algorithm and assign $SET_{CP}$ to the critical path found by our algorithm. The rest of the algorithm remains the same. Our main comparison in terms of makespan and related metrics is between CEFT-CPOP and CPOP. This provides us with a basis of a real comparison of the effectiveness of the critical path as the only difference between the two algorithms is the way the critical paths are calculated. We also provide a comparison against HEFT, to show far away our results are from the state of the art scheduling algorithm.
\section{Experimental Setup}
\label{sec:paper-crit-path-expe-resu}
In this section we present a statistical comparison of our algorithm with the current state of the art algorithms (CPOP and HEFT) in the context of a critical path finding algorithm and its ability to be adapted into a scheduling algorithm. In section~\ref{sec:rand-gene-grap} we outline the workloads on which the experiments are based upon and outline the experimental setup. Section~\ref{sec:paper-crit-path-comp-metr} defines the metrics based on which the effectiveness of our algorithm as a critical path finding algorithm is evaluated. The experimental set-up consists of a dual socket system consisting of Intel Xeon E5620 CPU running at 2.4 GHz with 24GB DDR3 RAM. The system is running Linux kernel ver 3.0.40-1. The code was compiled using GCC version 4.7.1 with `-O3' optimization flag.

\subsection{Randomly generated workloads}
\label{sec:rand-gene-grap}
In order to not bias the results towards any particular application, we present comparisons of our algorithm to its contemporaries on synthetically generated random graphs. We use a modified version of the random graph generator from~\cite{topcuoglu2002performance}. In the next subsection, we present four comparison metrics on which the relative performance of the three algorithms is compared. We generated four sets of input \ul{r}andomly \ul{g}enerated \ul{g}raphs (RGG) using the random graph generator: \textit{RGG-classic}, \textit{RGG-low}, \textit{RGG-medium} and \textit{RGG-high}. 

\textit{RGG-classic} is the first set of input application graphs that we generated to mimic the random graphs generated in the work presented by~\cite{topcuoglu2002performance} and \cite{arabnejad2014list}. These graphs use the heterogeneity factor that is embedded in them to generate the execution times of a given task on the different processors. Following on from the random graph generator used in the literature, the execution time for task $t_i$ on processor $p_j$ is randomly chosen from the following range:
\begin{equation}
w_i \times (1-\dfrac{\beta}{2}) \leq w_{i,j} \leq w_i \times (1+\dfrac{\beta}{2})
\end{equation}

\begin{figure}[t!]
	\centering
	\subfloat[Application graph]{
		\includegraphics[angle=0, scale=0.25]{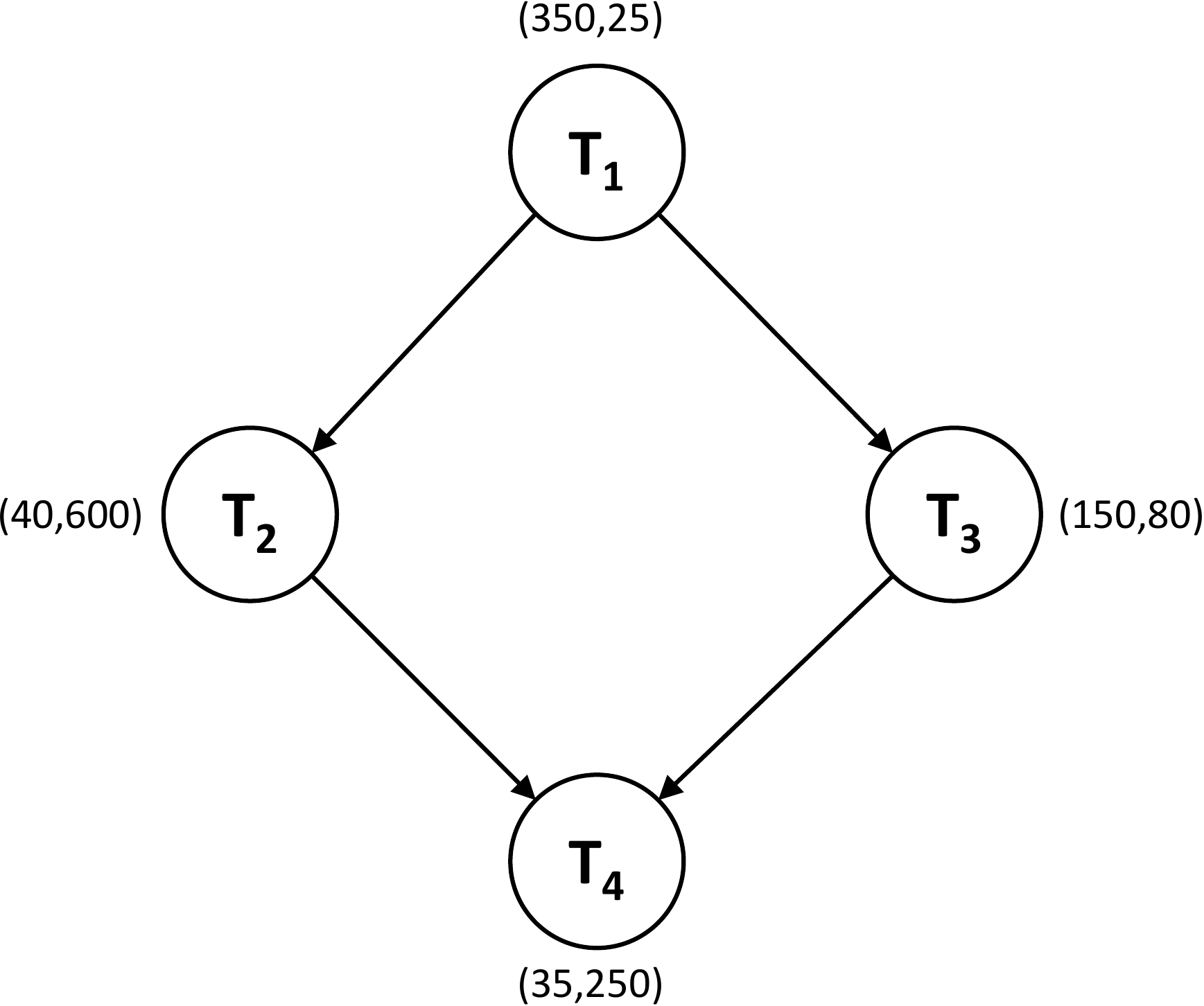}
		\label{fig:rgg-dag-app}
	}
	\qquad
	\centering
	\subfloat[Resource graph]{
		\includegraphics[angle=0, scale=0.25]{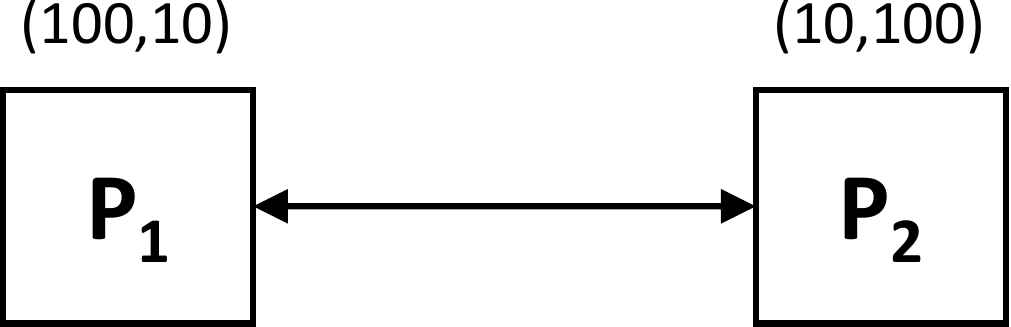}
		\label{fig:rgg-dag-res}
	}
	\caption{Sample graphs with 2 node-weights}
	\label{fig:rgg-dag}
\end{figure}

\noindent The possible range of values for $\beta$ is $0 \leq \beta \leq 1$, which means that $w_{i,j}$ can only possibly take values between $\dfrac{w_i}{2}$ and $\dfrac{3 \times w_i}{2}$. This implies that for any processor graph, a task can only be 3 times as fast on the fastest processor as it is on the slowest processor. This level of heterogeneity might not be representative of clusters which have certain processors with hardware accelerators. This is the major source of inspiration for us to generate the other three workloads. In the case of \textit{RGG-low}, \textit{RGG-medium} and \textit{RGG-high}, we use a modified version of the random graph generator. Every task in the modified graph from each of these workloads contains two node-weights as shown in figure~\ref{fig:rgg-dag}. Table~\ref{tab:rgg-dag-exec-time} shows the corresponding execution times of the tasks from figure~\ref{fig:rgg-dag-app} on the processors from figure~\ref{fig:rgg-dag-res}. 

\ifdefined\longversion
\begin{table}[b!]
\else
\begin{table}[t!]
\fi
	\centering
	\begin{tabular}{|c|c|c|}
		\hline
		& \textbf{P1} & \textbf{P2} \\
		\hline
		\textbf{T1} & 6 & 35.25\\
		\hline
		\textbf{T2} & 60.18 & 10 \\
		\hline
		\textbf{T3} & 9.5 & 15.8 \\
		\hline
		\textbf{T4} & 25.35 & 6 \\
		\hline
	\end{tabular}
	\caption{Execution time for the application and processor graph from figure~\ref{fig:rgg-dag}}
	\label{tab:rgg-dag-exec-time}
\end{table}

\begin{equation}
Cost(t_i, p_j)= \left[\dfrac{w^t_1(t_i)}{W^r_1(p_j)} + \dfrac{w^t_0(t_i)}{W^r_0(p_j)}\right]
\label{eq:latency}
\end{equation}

These execution times are calculated using a simple two-part cost model based on equation~\ref{eq:latency}. Every task and resource has two weights. The execution time of a task on a resource is given by the sum of the ratio of the corresponding node weights. We draw inspiration for this cost model from~\cite{vasudevan2014improved}. This cost model yields a higher variability in execution times with some tasks being fast on certain processors; while those processors being not universally faster for all tasks in the application graph. 

Consider the graphs from figure~\ref{fig:rgg-dag}; the value of node-weights of the tasks and processors determine the execution time of these tasks. We generated the same set of six processor graphs for the RGG-low, RGG-medium and RGG-high workloads. While creating said processor graphs, the values for the two node-weights are chosen from two intervals : $\{\mathcal{I}_1, \mathcal{I}_2\}$. At every node, a random number between 0 and 1 is chosen and if it is lower than $\beta$, the first node-weight is chosen from $\mathcal{I}_1$ and the second node-weight is chosen from $\mathcal{I}_2$. If it is higher than $\beta$ however, the two intervals are interchanged. This process of using the intervals to fill in the node-weights of the nodes in the graph is adopted for the application graphs too. For the workloads mentioned above, the following intervals were used:
\begin{itemize}
	\item Resource graph -- $\mathcal{I}_1 = \{10^2, 10^3\}$ and $\mathcal{I}_2=\{10^3, 10^4\}$
	\item \textit{RGG-low} -- $\mathcal{I}_1 = \{10^2, 10^3\}$ and $\mathcal{I}_2=\{10^3, 10^4\}$
	\item \textit{RGG-medium} -- $\mathcal{I}_1 = \{10^2, 10^3\}$ and $\mathcal{I}_2=\{10^4, 10^5\}$
	\item \textit{RGG-high} -- $\mathcal{I}_1 = \{10^2, 10^3\}$ and $\mathcal{I}_2=\{10^5, 10^6\}$
\end{itemize}

The result of this kind of workload generation, enables us to create workloads that have significantly different execution times. These workloads have the same structure, but differ in the execution times of the tasks as discussed in the previous paragraph. In order to generate the structure of the graphs, we use the random graph generator from the literature with the following parameters:
\begin{itemize}
	\item $n$ -- Number of tasks in the graph; -- $\textbf{\{128, 256, 512, 1024, 2048, 4096, 8192, 16384\}}$
	\item $o$ -- The average outdegree of a node in the graph; -- $\textbf{\{2, 4, 8\}}$
	\item $c$ -- Communication-to-Computation ratio (CCR)\ifdefined\longversion. It is the ratio of the weight of an edge leaving a vertex (i.e. communication cost) to the vertex weight (i.e. the computation cost). In order to incorporate heterogeneity in the communication backbone, the weight is chosen randomly in the range $w_i \times c \times ( 1-\dfrac{\beta}{2} ), w_i \times c \times ( 1+\dfrac{\beta}{2} )$; where $w_i$ represents the computation cost or the weight of the vertex and $\beta$ denotes the heterogeneity factor as described below; \fi-- $\textbf{\{0.001, 0.01, 0.1, 1, 5, 10\}}$
	\item $\alpha$ -- Shape parameter of the graph\ifdefined\longversion. The height of the graph depends on this parameter as $\dfrac{\sqrt{n}}{\alpha}$. The width of the graph is randomly chosen from a uniform distribution with a mean equal to $\alpha \times \sqrt{n}$. Hence smaller values of $\alpha$ gives tall and skinny graphs, while larger values of $\alpha$ gives short and fat graphs; \fi-- $\textbf{\{0.1, 0.25, 0.75, 1.0\}}$
	\item $\beta$ -- Heterogeneity factor of the graph\ifdefined\longversion. This parameter dictates the weights of the vertices in the graphs (i.e. computation costs) which is randomly chosen from the following range:
	\begin{equation}
	w_i \times (1-\dfrac{\beta}{2}) \leq w_{i,j} \leq w_i \times (1+\dfrac{\beta}{2})
	\end{equation}
	
	\noindent where $w_i$ is the weight of the vertex or the computation cost which is chosen randomly from a uniform distribution in the range $[0, 2 \times w_{DAG}]$. $w_{DAG}$ is the average computation cost of the graph and is chosen randomly. This is the way heterogeneity is incorporated into the application graphs throughout this paper unless otherwise stated explicitly; \fi-- $\textbf{\{10, 25, 50, 75, 95\}}$
	\item $\gamma$ -- Skewness parameter of the graph\ifdefined\longversion. It denotes how computation is spread across the graph. Smaller values of $\gamma$ gives uniformly distributed graphs while larger values give skewed graphs where pockets of the graph are more computationally intensive compared to other parts of the graph; \fi-- $\textbf{\{0.1, 0.25, 0.5, 0.75, 0.95\}}$
\end{itemize}

With the above configuration of parameters, a total of 14400 graphs were created. Each of these randomly generated graphs are scheduled on six different processor graphs ($p$ -- $\textbf{\{2, 4, 8, 16, 32, 64\}}$; where $p$ is the number of processors). This amounts to 86400 experiments (an experiment corresponds to an input application DAG, processor graph pair) for every workload and a total of 345600 experiments across all the workloads. To our best knowledge, our experiments are the only ones to use application graphs that have a large number of nodes (between 128 and 16384) as benchmarks. Previous evaluation of other heuristics such as HEFT and PEFT have a maximum of 500 nodes in the randomly generated graphs. \ifdefined\longversion\footnote{To explore the different workloads further, we encourage the readers to download the code from Github at \url{https://github.com/aravind-vasudevan/graphgen} and experiment with the different input parameters.}\fi

\ifdefined\longversion
\else
\vspace{-0.2cm}
\fi
\subsection{Real world graphs}
\label{sec:paper-crit-path-real-worl-grap}

In addition to the variants of the randomly generated graphs (RGG-classic, RGG-low, RGG-medium and RGG-high), we evaluate the performance of our algorithm on real-world applications, namely Gaussian elimination~\cite{cosnard1988parallel}, Fast Fourier Transform~\cite{topcuoglu1999task}, Molecular dynamics~\cite{kim1988general} and Epigenomics workflow~\cite{bharathi2008characterization}. As is a common trend in the scheduling research community~\cite{topcuoglu2002performance,arabnejad2014list}, we generate graphs based on the known structure of these real-world applications. These graphs are generated with differing values of some of the parameters discussed in section~\ref{sec:rand-gene-grap}. Since the structure is known for these applications, the $\alpha$ parameter of the graphs cannot be changed. The range of values for the other parameters are as follows: $\{\textbf{0.001, 0.01, 0.1, 0.5, 1, 5, 10}\}$ for $c$ (CCR) and $\{\textbf{10, 25, 50, 75, 95}\}$ for $\beta$ (heterogeneity). These real world applications are also run on the six different processor graphs as mentioned in the previous section.

\ifdefined\longversion
\begin{figure}[t]
	\centering
	\hspace*{\fill}%
	\csubfloat[Gaussian Elimination]{
		\includegraphics[angle=0, scale=0.15]{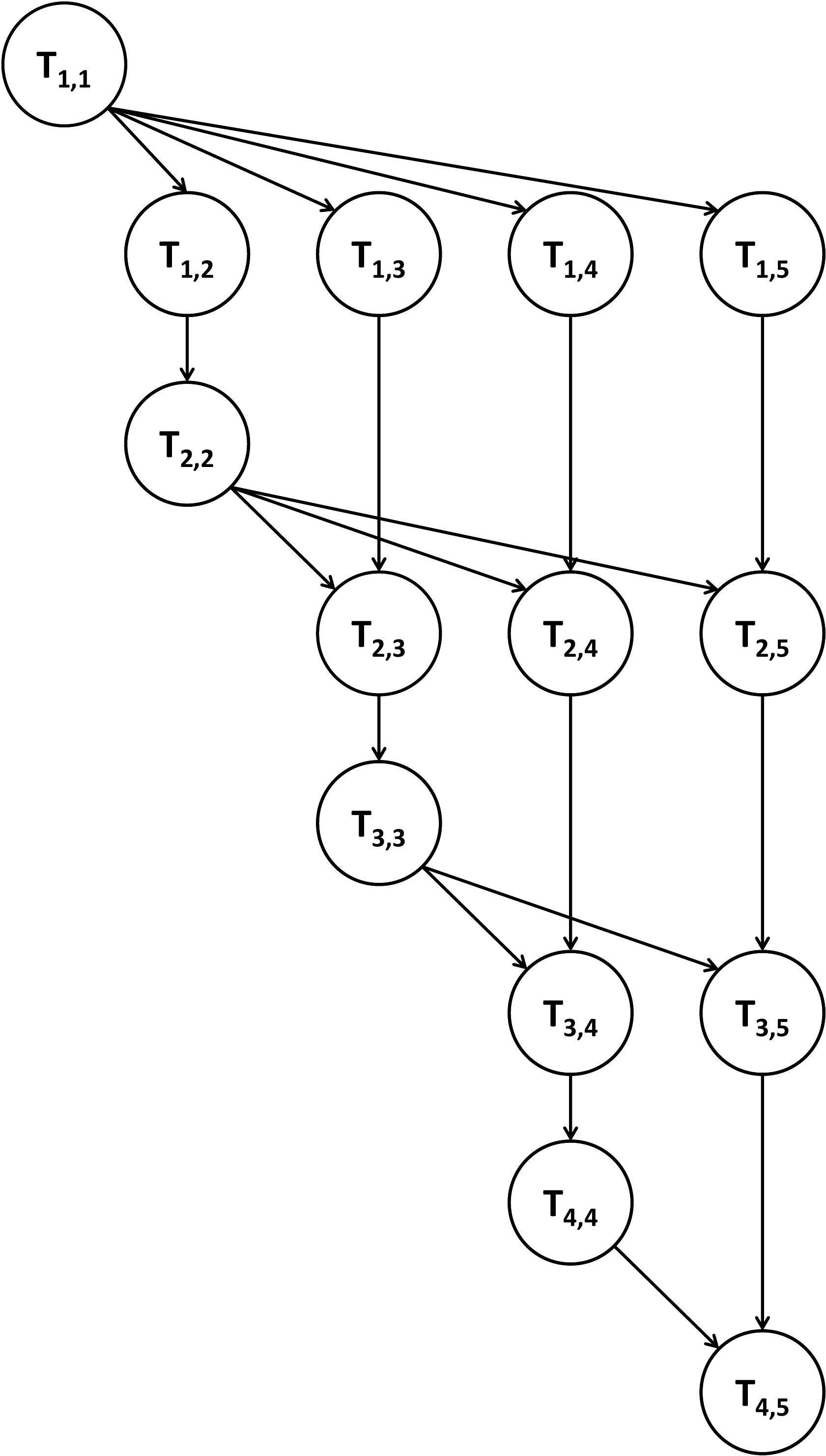}
		\label{fig:gaussian}
	}
	\centerhfill
	\csubfloat[Fast Fourier Transform]{
		\includegraphics[angle=0, scale=0.2]{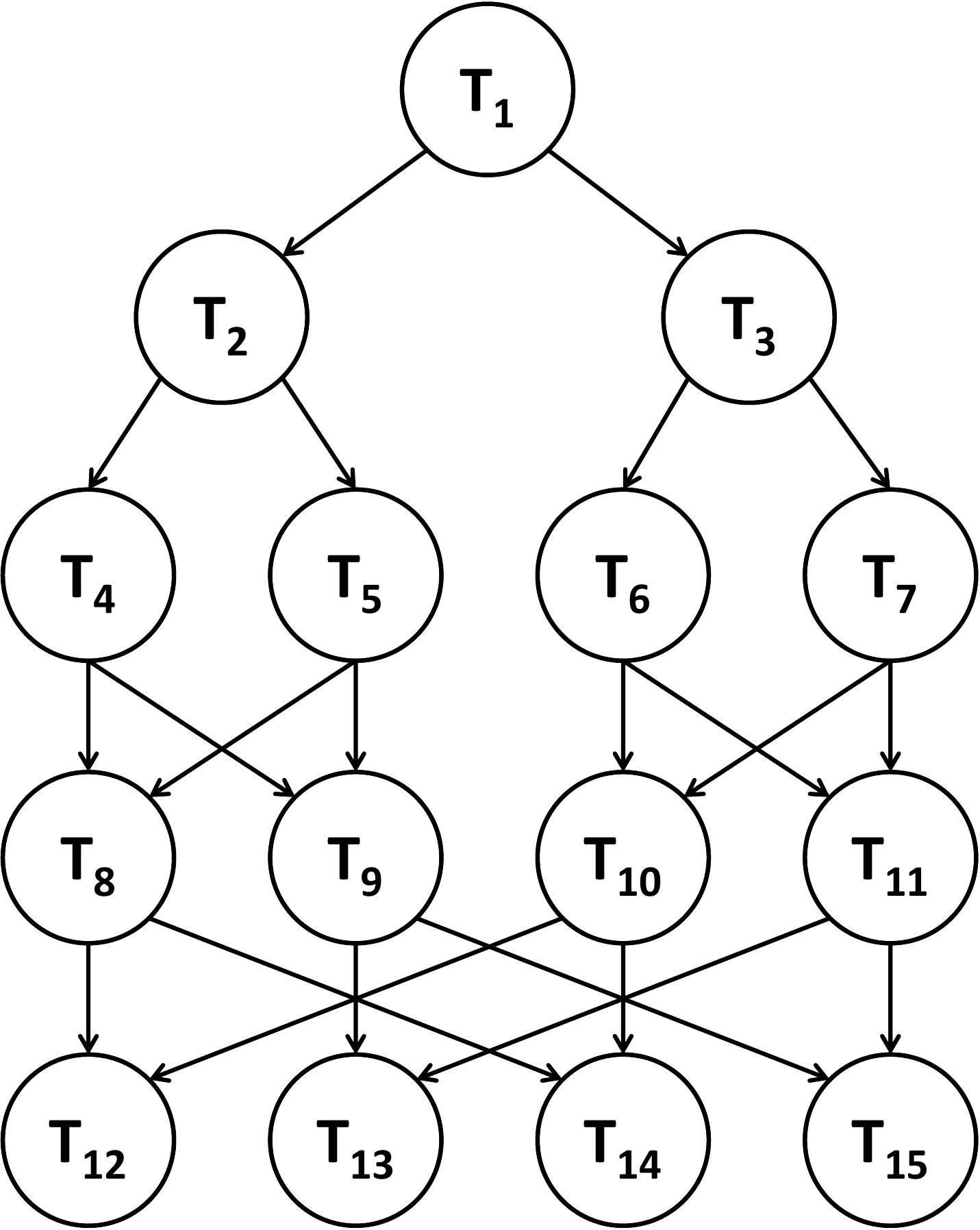}
		\label{fig:fft}
	}
	\hspace*{\fill}%
	\caption[Application DAGs for Gaussian Elimination and Fast Fourier Transform]{Application DAGs for Gaussian Elimination and Fast Fourier Transform. Redrawn from~\cite{wu1990hypertool,cosnard1988parallel,topcuoglu2002performance}}
	\label{fig:gaussian-fft}
\end{figure}

\subsubsection{Fast Fourier Transform (FFT)}
\label{sec:paper-crit-path-fft}
Figure~\ref{fig:fft} shows the task graph of another real world application, \textit{Fast Fourier Transform}. The FFT algorithm can be split into two parts: recursive calls and the butterfly operation~\cite{topcuoglu2002performance} represented by the dashed line in the figure. All the tasks above the line represent the recursive calls and the ones below are the butterfly operation tasks. For a given input vector of size $m$ which is a power of two, there are $2 \times m-1$ recursive calls and $m \times {log}_2 m$ butterfly operations. This application is especially unique in that all the paths in this application are the \textit{critical-path} and they all have the same weight~\cite{cosnard1988parallel}.

\subsubsection{Gaussian Elimination (GE)}
\label{sec:paper-crit-path-gaus}
\textit{Gaussian Elimination}~\cite{wu1990hypertool,cosnard1988parallel}, is an algorithm which solves a linear system of equations by performing a sequence of operations on the associated matrix of coefficients. Figure~\ref{fig:gaussian} shows the task graph for the Gaussian Elimination algorithm operating on a matrix of size 5. The total number of tasks in a Gaussian Elimination graph is given by $\dfrac{m^2+m-2}{2}$ and in the case of the figure, the number of tasks when $m=5$ is 14.

\subsubsection{Molecular Dynamics (MD)}
\label{sec:paper-crit-path-mole-dyna}
\begin{figure}[hb!]
	\centering
	\includegraphics[angle=0, scale=0.45]{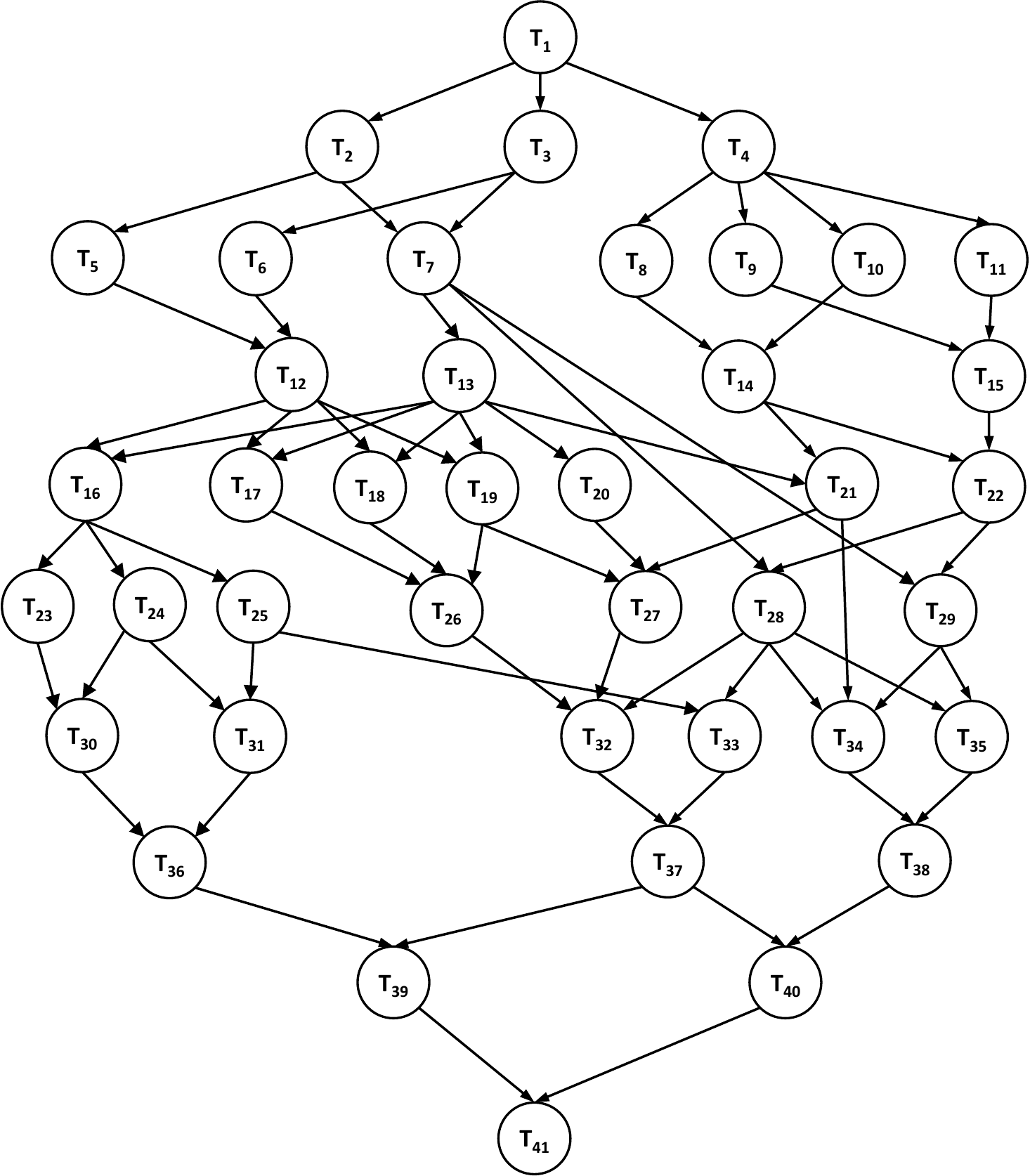}
	\caption[Application DAG for Molecular dynamics code]{Application DAG for Molecular dynamics code. Redrawn from~\cite{kim1988general}}
	\label{fig:molecular-dynamics}
\end{figure}
A commonly found application in the literature is the modified molecular dynamic code from~\cite{kim1988general}. The task graph of this code is presented in Figure~\ref{fig:molecular-dynamics}. This application serves as a benchmark for scheduling algorithms due the shape of its irregular task graph. The task graph was modified by Browne~\cite{kim1988general} from its original structure in order to increase the number of tasks and edges. They also modified the computation and communication times of the tasks and edges while synthetically generating the architecture on which this task graph was run. This was done in an attempt to increase the variability in the graph. In a similar vein, all the scheduling algorithms presented in this paper were tested on synthetically generated application and processor graphs, unless explicitly stated otherwise.

\subsubsection{Epigenomics Workflow (EW)}
\label{sec:paper-crit-path-epig-work}
The epigenomics workflow is a data processing pipeline that automates the execution of various genome sequencing operations. It maps the epigenetic state of the human cells on a genome-wide scale. Parts of this application can be split into independent chunks (split on inputs) which can be executed in parallel. The outputs from these independent chunks are further processed to filter noise and contaminating sequences. The graph has a very compact parallel structure and is generally wider than it is taller. 

\else
\vspace{-0.2cm}
\fi
\subsection{Comparison metrics}
\label{sec:paper-crit-path-comp-metr}

We compare the algorithms based on the following comparison metrics : critical path length (CPL), schedule length (makespan), speedup, schedule length ratio (SLR), slack and a pairwise comparison of number of occurrences of better solutions which are common heuristics used to compare the performance of scheduling algorithms~\cite{topcuoglu2002performance,kwok1996dynamic,arabnejad2014list,ahmad1998exploiting,topcuoglu1999task}.

\ifdefined\longversion

\subsubsection{Critical path length (CPL)}
\label{sec:paper-crit-path-crit-path-leng}
As we have discussed in section~\ref{sec:paper-crit-path-form}, the critical path is the longest path from the entry node to the exit node in the application graph. The length of the critical path in turn becomes a key metric as it serves as a hard lower bound for the schedule length (makespan). As our algorithm is primarily a critical path finding algorithm, this metric is of key importance and we compare the lengths of the paths produced by our algorithm and CPOP for a given input application graph and processor graph pair. HEFT is not a critical path based scheduling algorithm and hence we cannot present the statistics for it under this comparison metric.

\subsubsection{Speedup}	
\label{sec:paper-crit-path-effi}

Speedup is defined as the ratio of the sequential execution time to the parallel execution time (\textit{makespan}). The sequential execution time is calculated by assigning all tasks onto the processor which minimizes the total execution time of the task graph as shown in the following equation :

\begin{equation}
\label{eq:spee}
Speedup = \frac{min_{p_j\in P}[\Sigma_{t_i \in T}\costTable(t_i, p_j)]}{makespan}
\end{equation}

\noindent In equation~\ref{eq:spee}, the numerator represents the sequential execution time of the input application graph for the given processor graph. This value is independent of the choice of the scheduling algorithm and is therefore a constant for all the three algorithms (our critical path algorithm, CPOP and HEFT) under scrutiny here. This implies that the speedup is the makespan normalised against the sequential execution time which is constant across all the algorithms compared. Hence, speedup is often used as a better replacement metric for the makespan as it returns a normalised  score.

\subsubsection{Scheduling length ratio (SLR)}
\label{sec:paper-crit-path-sche-leng-rati}

The most commonly used metric to compare the performance of scheduling algorithms is the \textit{makespan}. Its use as a comparison metric has been well established in the literature~\cite{topcuoglu2002performance, kwok1996dynamic, topcuoglu1999task, braun2001comparison}. But in order to normalize the schedule length against any topology/processor graph, we adopt the normalized schedule length (NSL)~\cite{daoud2008high} which is also called the scheduling length ratio (SLR). It is defined as follows :

\begin{equation}
\label{eq:slr}
SLR = \dfrac{makespan}{\underset{t_i \in CP}{\Sigma}min_{p_j \in P}[\costTable(t_i, p_j)]}
\end{equation}

\noindent where $CP$ is the critical path. The denominator\footnote{The denominator of SLR is often confused with the numerator of speedup. They are not the same as the task set to which they are applied to is different. In the denominator of the SLR only tasks from the CP are considered, while the numerator of speedup considers all the tasks in the task graph.} of the equation gives the sum of the computation costs of the critical path tasks assuming they are assigned to the processors which minimize their individual execution times. The SLR of an application DAG (under an optimal assignment or using any other scheduling algorithm) cannot be less than one as no valid schedule of tasks on the processors can produce a smaller makespan than the denominator. Since the critical path serves as lower bound for the makespan, one can identify that the denominator value might be smaller\footnote{It is equal in the case where the input application graph is a linear DAG} than the \textit{true} critical path; as this formulation ignores communication cost and hence produces shorter critical path lengths than the true critical path length.

\subsubsection{Slack}
\label{sec:paper-crit-path-slac-metr}
Slack is a commonly used metric in the context of comparing scheduling algorithms~\cite{shi2006robust}. It represents the ability of a schedule to deal with delays in the execution of some tasks. It represents how accommodative a schedule is and acts as representative of robustness in the scheduling algorithms literature~\cite{boloni2002robust}. The slack of a task represents the time window within which the task can be delayed without extending the makespan. Slack is defined as,

\begin{equation}
Slack = \dfrac{\mathlarger{\mathlarger{\sum}}_{t_i \in V}M-b_{level}(t_i)-t_{level}(t_i)}{\numberOfTasks}
\end{equation}

\noindent It is important to note that makespan and slack are conflicting metrics. Makespan is representative of the efficiency of the scheduling algorithms in terms of its capability to lower the execution time of the application DAG; whereas, slack is representative of the forgiving nature of the schedule. \footnote{If one were to reduce the problem \textit{ad absurdum}, a schedule that never finishes or finishes at infinity will have the highest (infinite) slack.}

Static algorithms however, deal with the time dependence of certain application DAGs by using stochastic models where task execution times are random variables as discussed in~\cite{adam1974comparison}. Braun et al. also suggest that scheduling algorithms having higher slack are more \textit{robust} for DAGs who employ a stochastic model for task execution times. In our experiments, we do not use any such dynamic DAGs as our workloads are comprised entirely of static DAGs. In this case, a higher slack usually means that the schedule is not \textit{tight} enough. However, if a scheduling algorithm creates a schedule with low SLR and high slack, it means that there is still scope for improvement in the schedule and hence even lower SLR values can be obtained. 

\fi

\section{Results and analysis}
\label{sec:paper-crit-path-resu-anal}

\ifdefined\longversion

\begin{table}[t!]
	\resizebox{\linewidth}{!}{%
		\centering
		\begin{tabular}{|c|c|c|c|c|}
			\hline
			\textbf{Workload} & \textbf{\# of experiments} & \textbf{CPL(\%)} & \textbf{makespan(\%)}\\
			\hline
			\multirow{3}{*}{RGG-classic}& 99346 & Longer & 60.06 & 26.95 \\
			\cline{2-5}
			& 99346 & Equal & 39.93 & 57.12\\
			\cline{2-5}
			& 99346 & Shorter & 0 & 15.9\\
			\hline
			\multirow{3}{*}{RGG-low}& 100800 & Longer & 40.61 & 23.15\\
			\cline{2-5}
			& 100800 & Equal & 0.46 & 0.89\\
			\cline{2-5}
			& 100800 & Shorter & 58.92 & 75.94\\
			\hline
			\multirow{3}{*}{RGG-medium}& 100800 & Longer & 16.52 & 7.96\\
			\cline{2-5}
			& 100800 & Equal & 0.33 & 1.74\\
			\cline{2-5}
			& 100800 & Shorter & 83.14 & 90.29\\
			\hline
			\multirow{3}{*}{RGG-high}& 100800 & Longer & 15.20 & 7.66\\
			\cline{2-5}
			& 100800 & Equal & 0.8 & 2.64\\
			\cline{2-5}
			& 100800 & Shorter & 83.99 & 89.69\\
			\hline
			
		\end{tabular}
	}
	\caption{Percentage of instances in the experiments where CEFT's CPL and makespan are longer, equal or shorter than CPOP's corresponding values}
	\label{tab:ceft-cpop}
\end{table}

\else

\begin{table}[t!]
			\centering
	\resizebox{0.85\linewidth}{!}{%

		\begin{tabular}{|c|c|c|c|}
			\hline
			\textbf{Workload} & & \textbf{CPL(\%)} & \textbf{makespan(\%)}\\
			\hline
			\multirow{3}{*}{RGG-classic}& Longer & 60.06 & 26.95 \\
			\cline{2-4}
			& Equal & 39.93 & 57.12\\
			\cline{2-4}
			& Shorter & 0 & 15.9\\
			\hline
			\multirow{3}{*}{RGG-low}& Longer & 40.61 & 23.15\\
			\cline{2-4}
			& Equal & 0.46 & 0.89\\
			\cline{2-4}
			& Shorter & 58.92 & 75.94\\
			\hline
			\multirow{3}{*}{RGG-medium}& Longer & 16.52 & 7.96\\
			\cline{2-4}
			& Equal & 0.33 & 1.74\\
			\cline{2-4}
			& Shorter & 83.14 & 90.29\\
			\hline
			\multirow{3}{*}{RGG-high}& Longer & 15.20 & 7.66\\
			\cline{2-4}
			& Equal & 0.8 & 2.64\\
			\cline{2-4}
			& Shorter & 83.99 & 89.69\\
			\hline
			
		\end{tabular}
	}
	\caption{Percentage of instances in the experiments where CEFT's CPL and makespan are longer, equal or shorter than CPOP's corresponding values}
	\label{tab:ceft-cpop}
\end{table}

\fi

\begin{figure}[t!]
	\centering
	\includegraphics[width=0.92\linewidth]{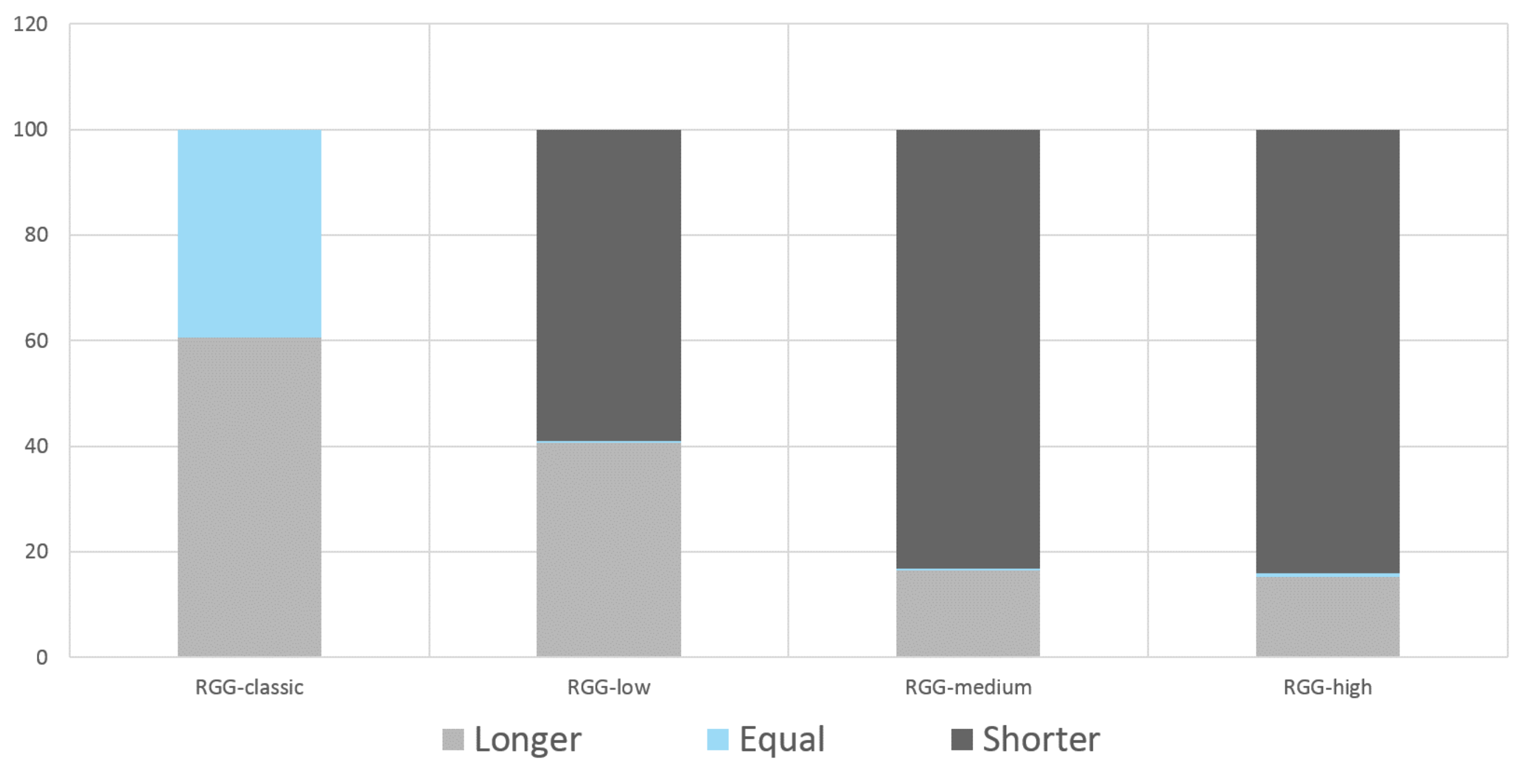}
	\caption{Percentage of instances CEFT's CPL is longer, equal or shorter compared to CPOP's CPL}
	\label{fig:cpl-comparison}
\end{figure}

\begin{figure}[b!]
	\centering
	\includegraphics[width=0.92\linewidth]{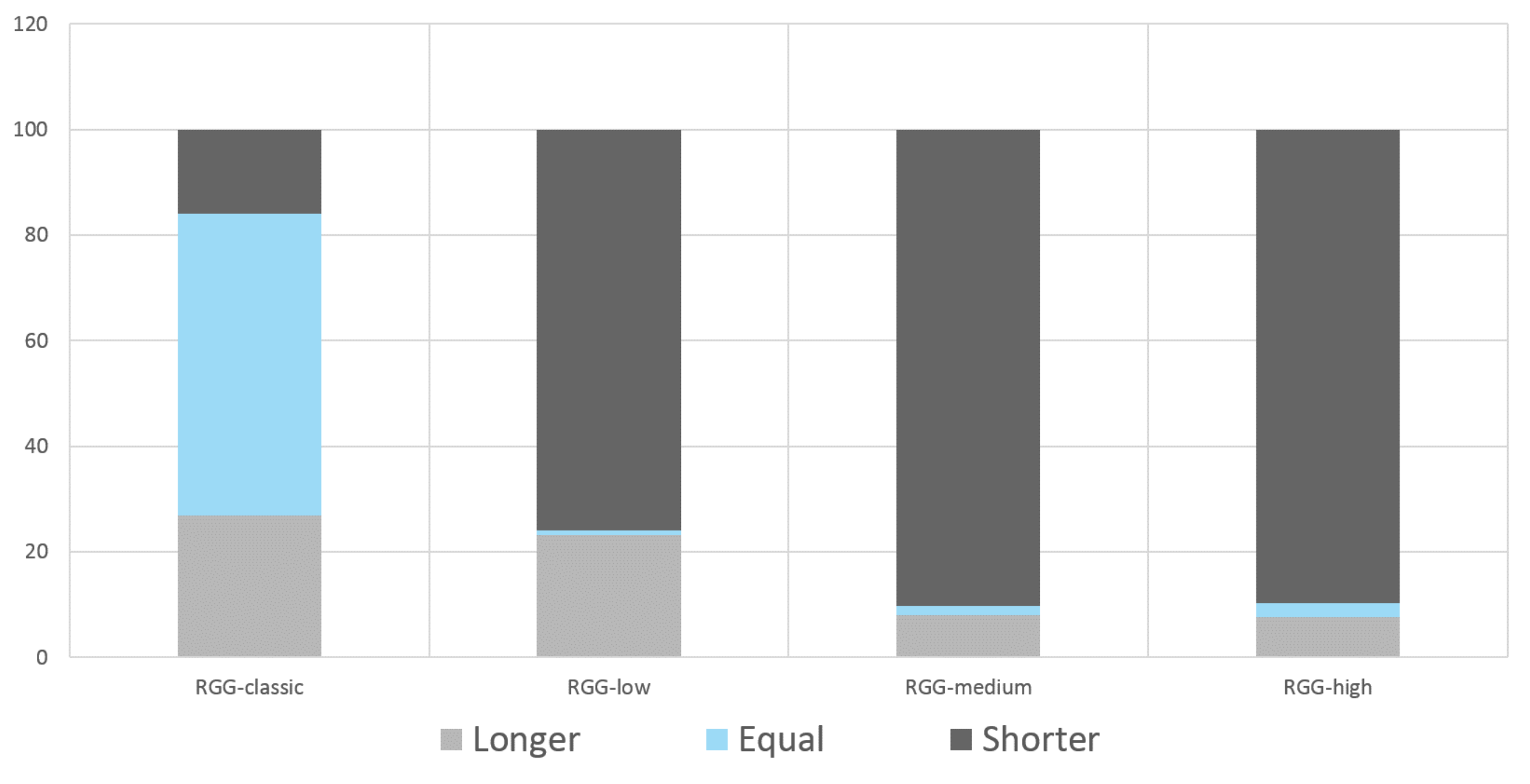}
	\caption{Percentage of instances CEFT's makespan is longer, equal or shorter compared to CPOP's makespan}
	\label{fig:makespan-comparison}
\end{figure}

\ifdefined\longversion
\else
\begin{figure*}[h]
	\centering
	\hspace*{\fill}%
	\csubfloat[RGG-classic]{
		\includegraphics[width=0.22\linewidth]{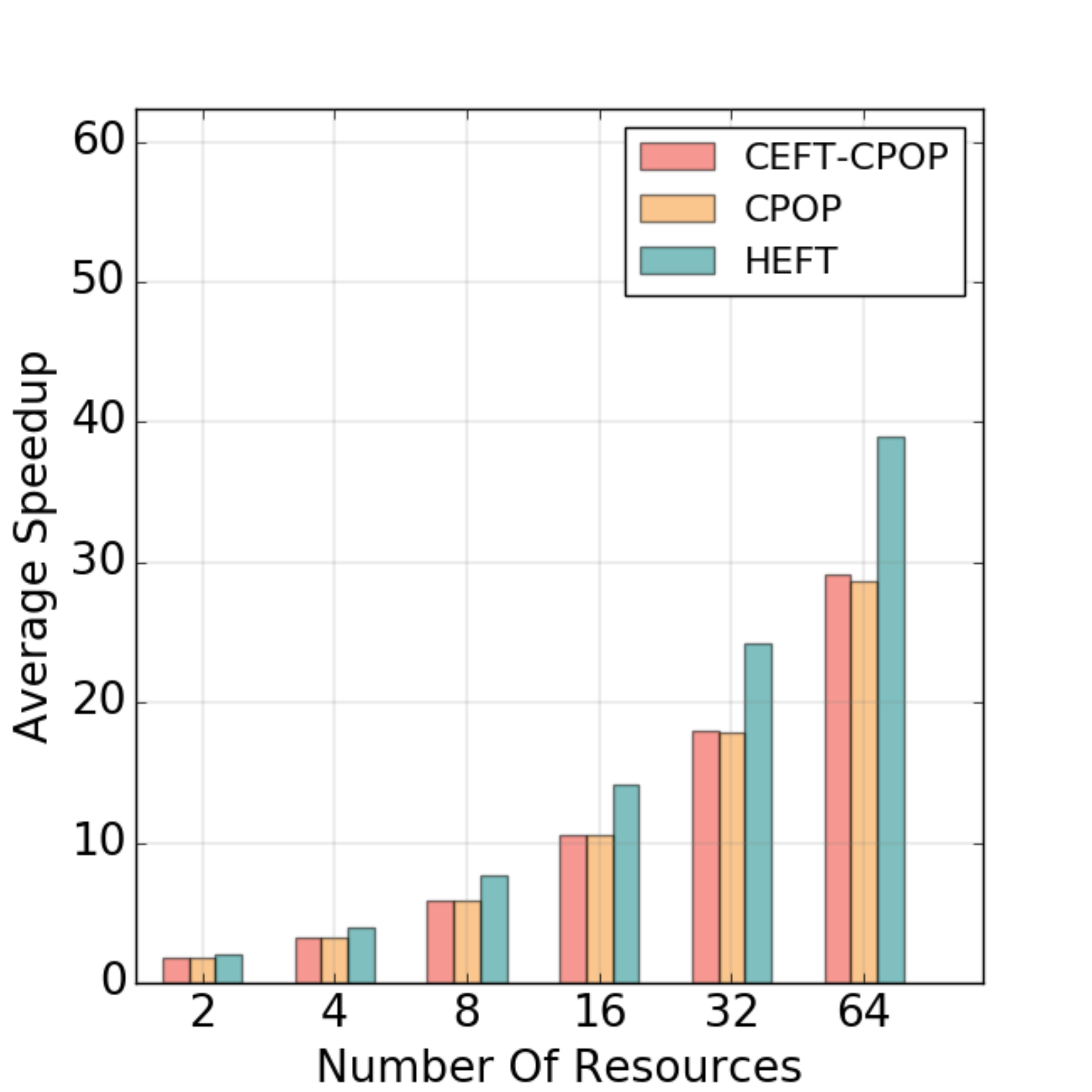}
		\label{fig:rgg-clas-res-speedup}
	}
	\centerhfill
	\csubfloat[RGG-low]{
		\includegraphics[width=0.22\linewidth]{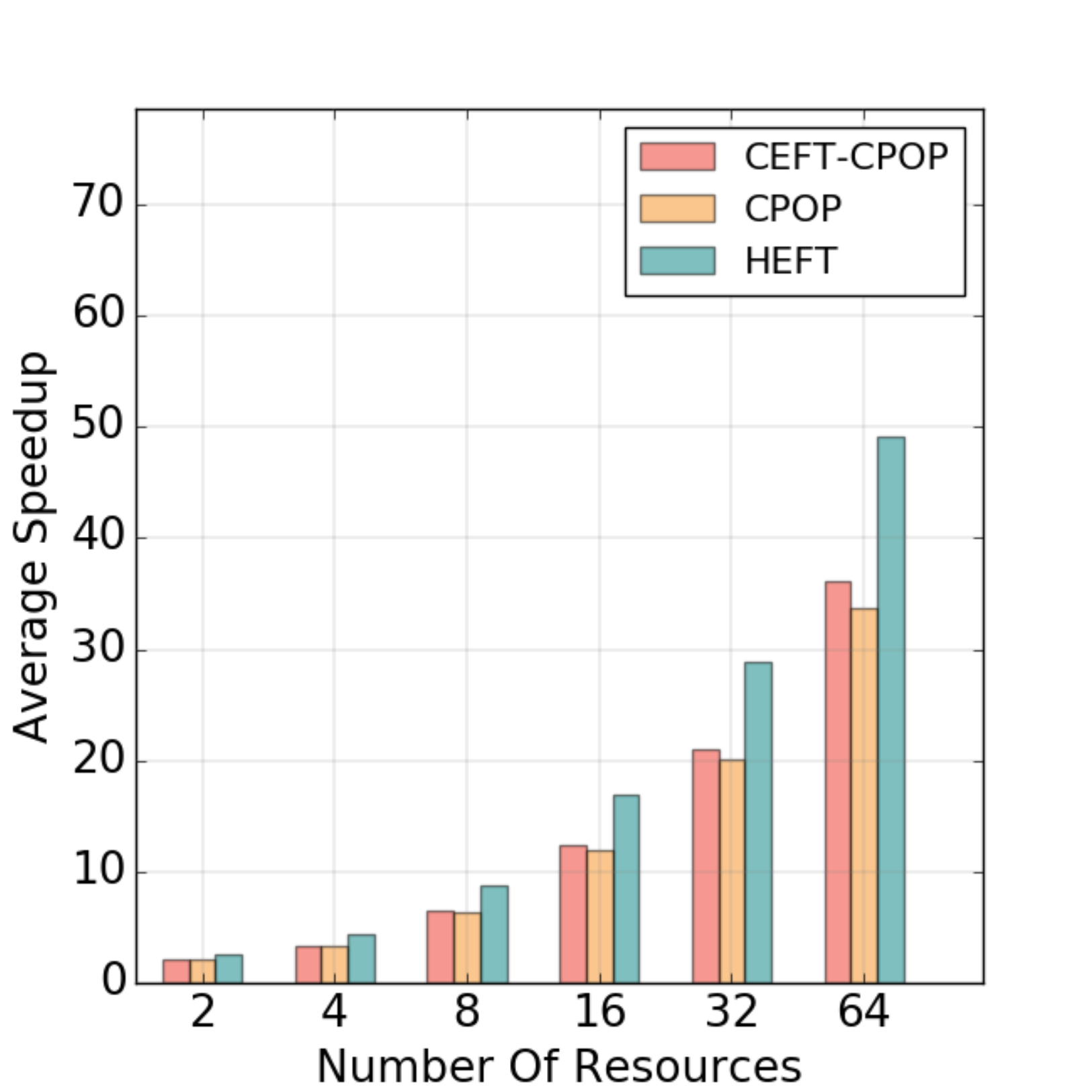}
		\label{fig:rgg-low-res-speedup}
	}
	\centerhfill
	\csubfloat[RGG-medium]{
		\includegraphics[width=0.22\linewidth]{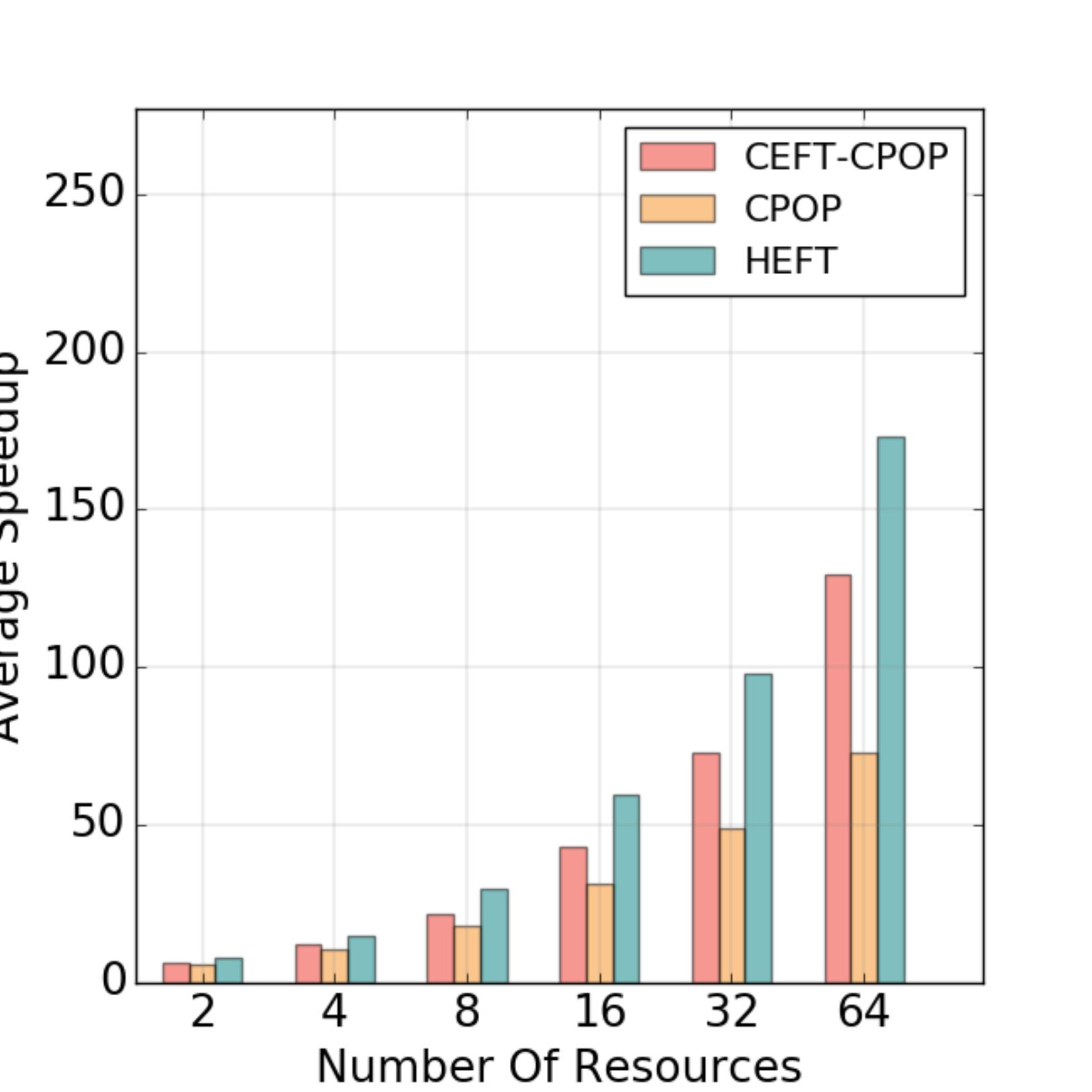}
		\label{fig:rgg-med-res-speedup}
	}
	\centerhfill
	\csubfloat[RGG-high]{
		\includegraphics[width=0.22\linewidth]{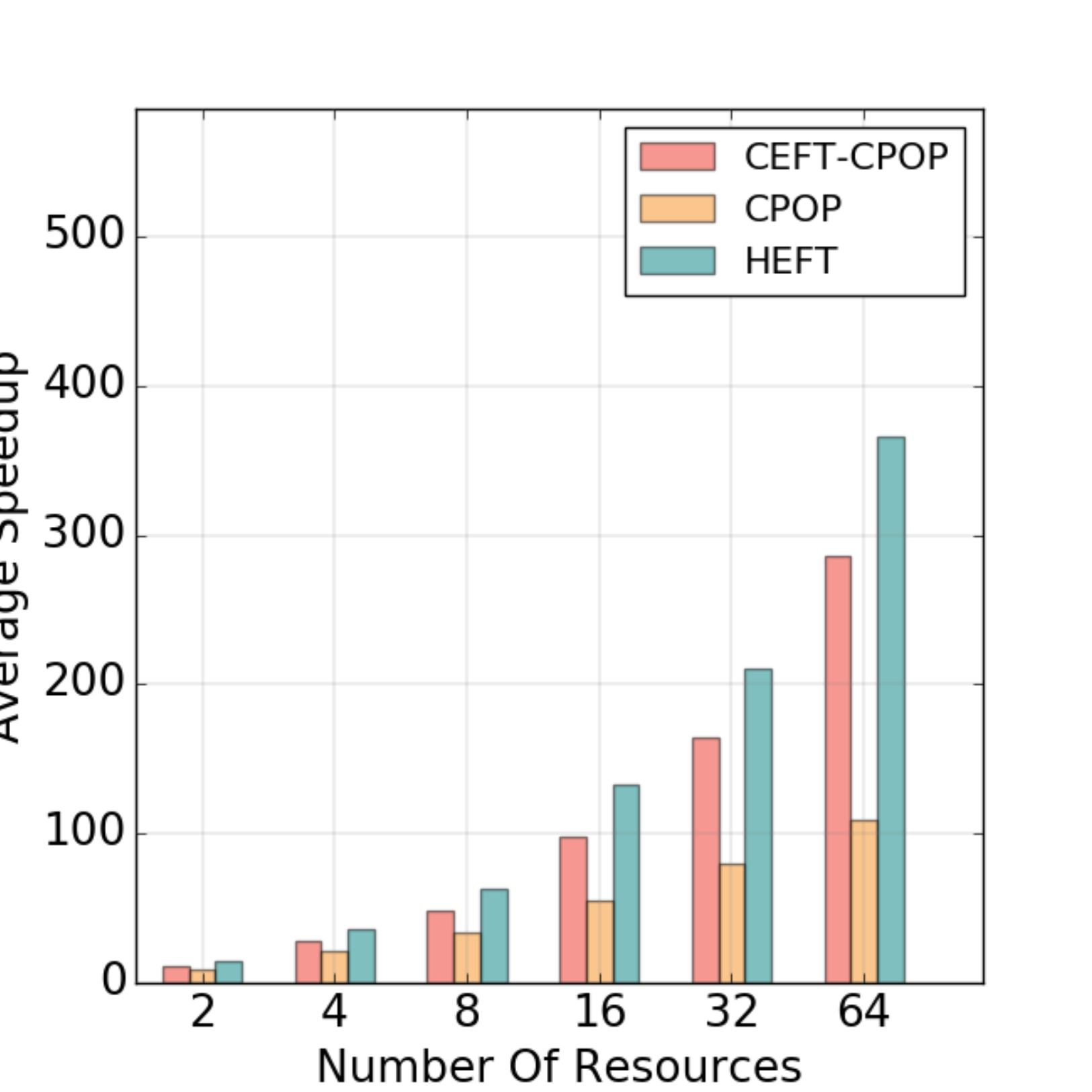}
		\label{fig:rgg-high-res-speedup}
	}
	\hspace*{\fill}%
	\caption{Comparing speedup across different workloads in terms of the number of processors in the processor graph. Higher is better.}
	\label{fig:res-speedup}
\end{figure*}

\begin{figure*}[t!]
	\centering
	\begin{minipage}{.49\textwidth}
		\centering
		\hspace*{\fill}%
		\csubfloat[RGG-classic]{
			\includegraphics[width=0.49\linewidth]{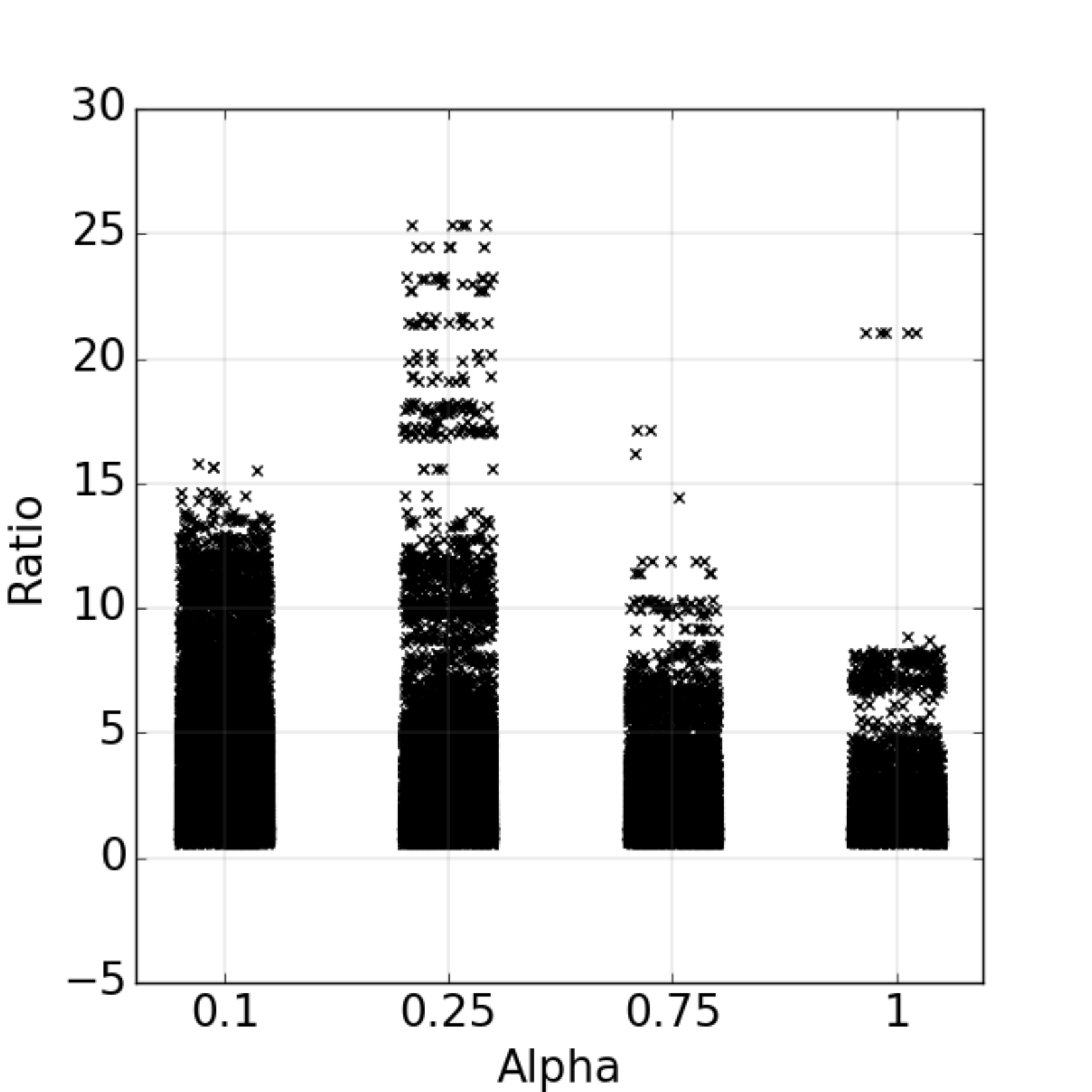}
			\label{fig:rgg-clas-res-cpl}
		}
		\centerhfill
		\csubfloat[RGG-high]{
			\includegraphics[width=0.49\linewidth]{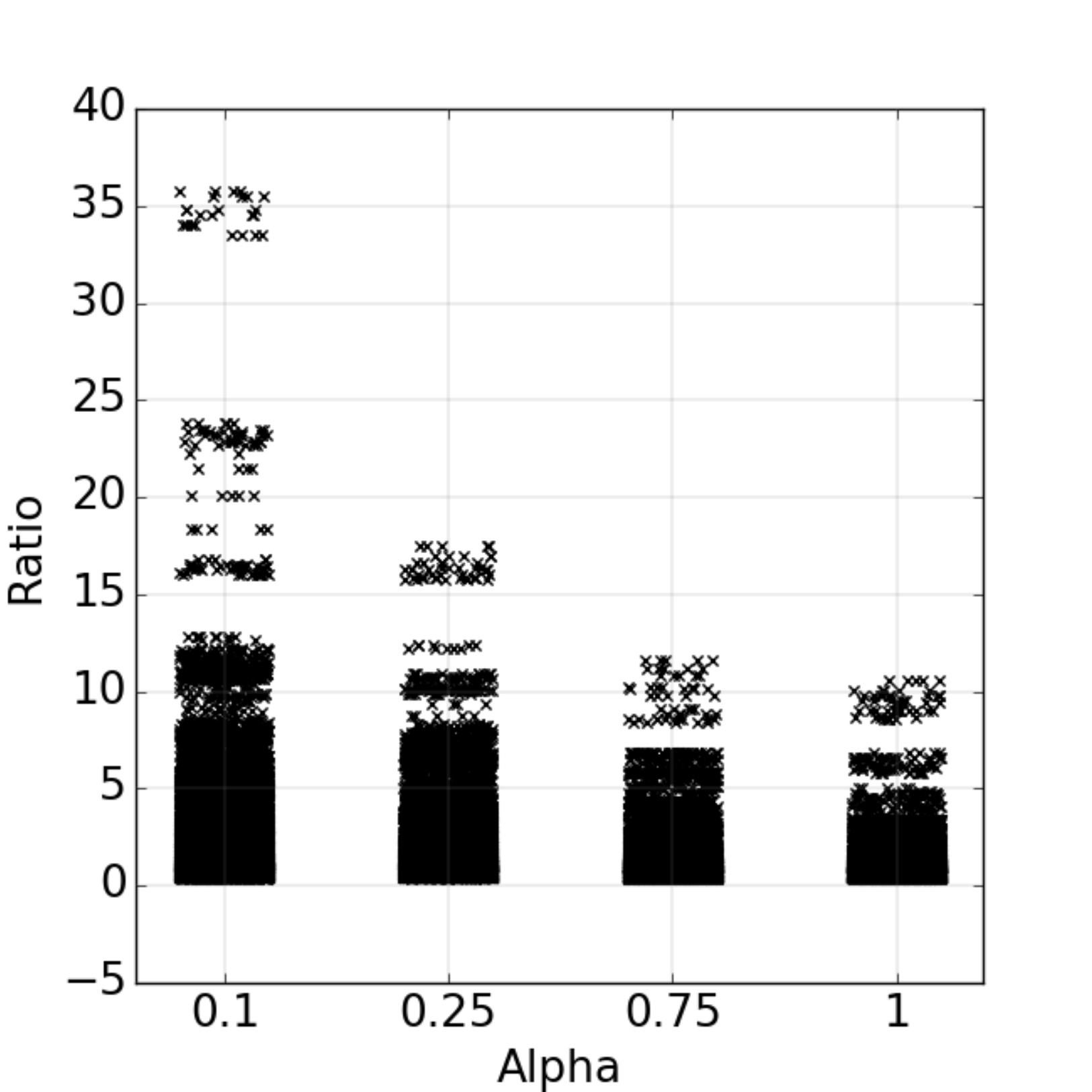}
			\label{fig:rgg-high-cpl}
		}
		\hspace*{\fill}%
		\caption{Comparing the lengths of the critical paths across RGG-classic and RGG-high workloads in terms of $\alpha$ of the application graph}
		\label{fig:alpha-cpl}
	\end{minipage}
	\hfill
	\begin{minipage}{.24\textwidth}
		\includegraphics[width=.97\linewidth]{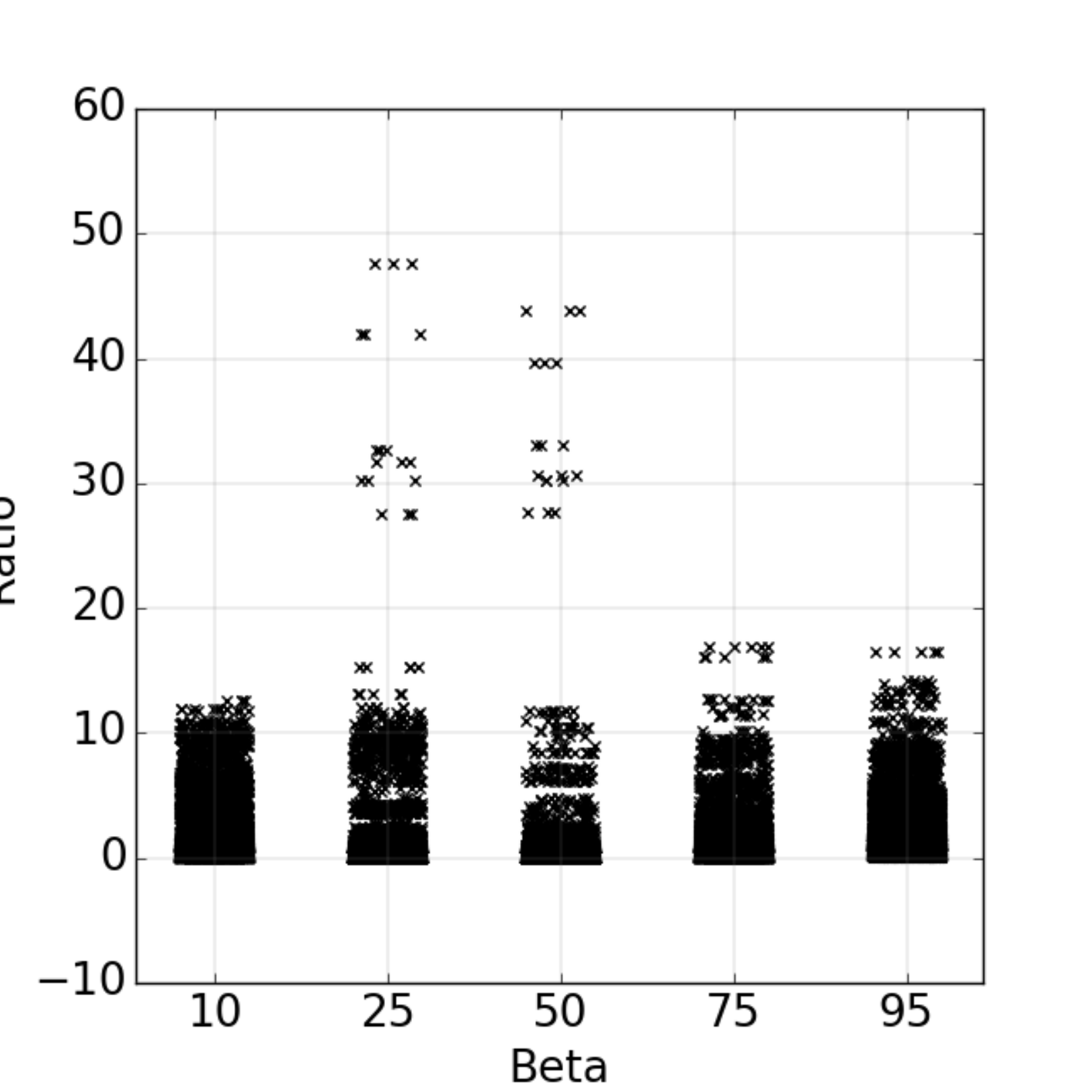}
		\caption{Comparing CPL for RGG-medium in terms of different values of $\beta$ in the input graphs}
		\label{fig:cpl-not-changed-for-beta}
	\end{minipage}
	\hfill
	\begin{minipage}{.24\textwidth}
		\centering
		\includegraphics[width=.97\linewidth]{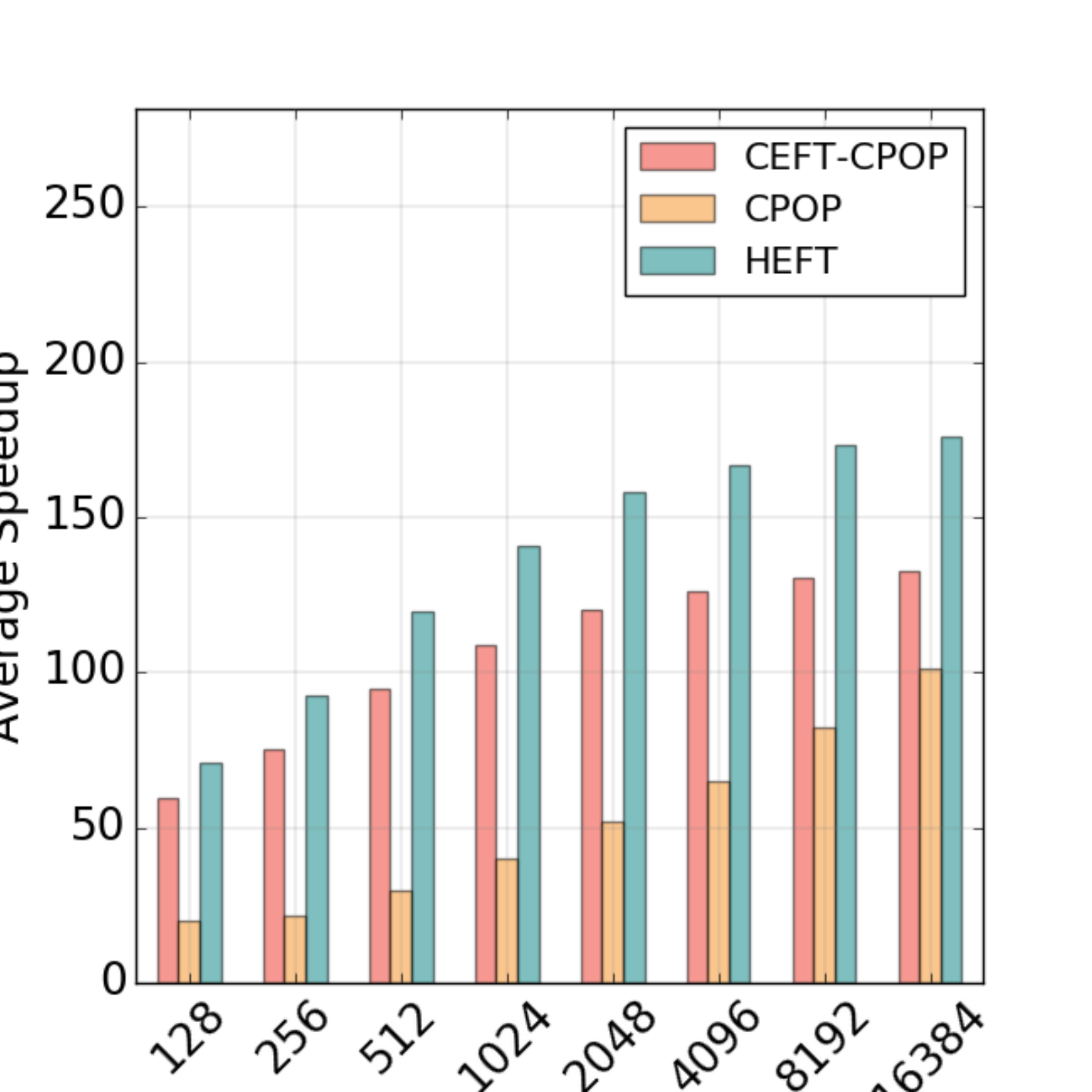}
		\caption{Comparing speedup for RGG-high in terms of number of tasks in the input graphs. Higher is better.}
		\label{fig:task-speedup}
	\end{minipage}%
\end{figure*}

\fi

In this section we compare the performance of our critical path finding algorithm (CEFT) against the current state of the art critical path algorithm (CPOP). We also present a brief comparison of the extension of our critical path algorithm (CEFT-CPOP) to function as a scheduling algorithm and compare its results against CPOP. Since the only difference between CEFT-CPOP and CPOP is the method by which the critical paths are found and mapped, makespan related metrics between these two algorithms help us clearly understand the effects of finding the right critical path.

Table~\ref{tab:ceft-cpop} compares CEFT and CPOP in terms of the critical path lengths produced and corresponding makespans. Figures~\ref{fig:cpl-comparison} and \ref{fig:makespan-comparison} put table~\ref{tab:ceft-cpop} into graphical context. We can observe from these graphs that CEFT produces either longer or same length critical paths as CPOP in the classic workload. However, when heterogeneity is better expressed, we produce shorter makespans in about 83\% of the cases. This is similarly reflected in the corresponding makespans produced by CEFT. Note however, that the table only provides the percentage of the number of instances in which path lengths and corresponding makespans are \textit{longer}, \textit{equal} or \textit{shorter} and discloses nothing about the relative quality of the solutions obtained. 

Figures~\ref{fig:rgg-clas-res-cpl} and \ref{fig:rgg-high-cpl} on the other hand help understand the relative quality of the solutions obtained by the two algorithms. Both the plots shown here are scatter plots. As the density of the points in the scatter plot is so high, we chose to offset the points that are on the line corresponding to a particular $\alpha$, by a small random amount (in the x-axis; within a preset range) to form a ``bar'' that better displays how the ratios are distributed. All the points inside the bar correspond to the value of $\alpha$ that the bar sits on top of. As the graphs become wider (with increasing values of $\alpha$), the critical path lengths found by CEFT become shorter. This stems from the fact that, while no other application graph parameter changes, the increase in the width of the graph gives rise to more shorter paths from the source task to the exit task. Since the objective of CEFT is to find the longest shortest path from all the possible paths, the critical path lengths produced by it decrease as well. This holds true in the case of the high heterogeneity workloads as well (RGG-high). \footnote{\ifdefined\longversion At this juncture, we have to mention that this way of representing the critical path length ratio is a bit misleading. The density of the points in the lower portions of the graph is not clearly visible and hence it seems like CEFT always produces longer critical paths than CPOP. While this is true in the case of RGG-classic, the CPL produced by our algorithm (CEFT) is shorter in 83.99\% of the experiments in RGG-high. Another note on the graph, is that it appears as if t\else T\fi here are some values that look like they are below the zero line (which is impossible since the critical path length ratio between any two algorithms can never be negative). This is because the plot uses an 'x' marker to plot the points.}

\ifdefined\longversion
\begin{figure*}[ht!]
	\centering
	\begin{minipage}{.46\textwidth}
		\centering
		\hspace*{\fill}%
		\csubfloat[RGG-classic]{
			\includegraphics[width=0.49\linewidth]{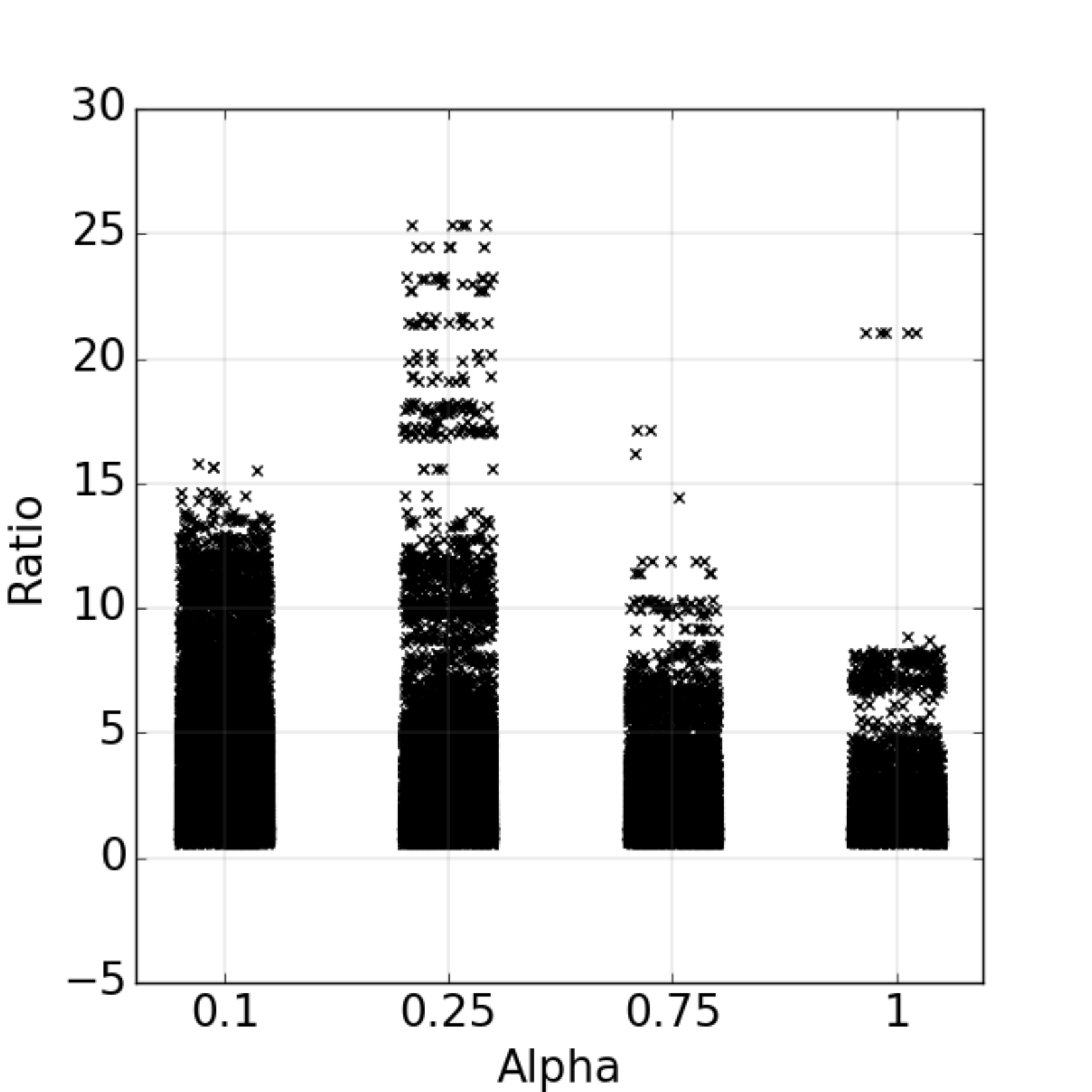}
			\label{fig:rgg-clas-res-cpl}
		}
		\centerhfill
		\csubfloat[RGG-high]{
			\includegraphics[width=0.49\linewidth]{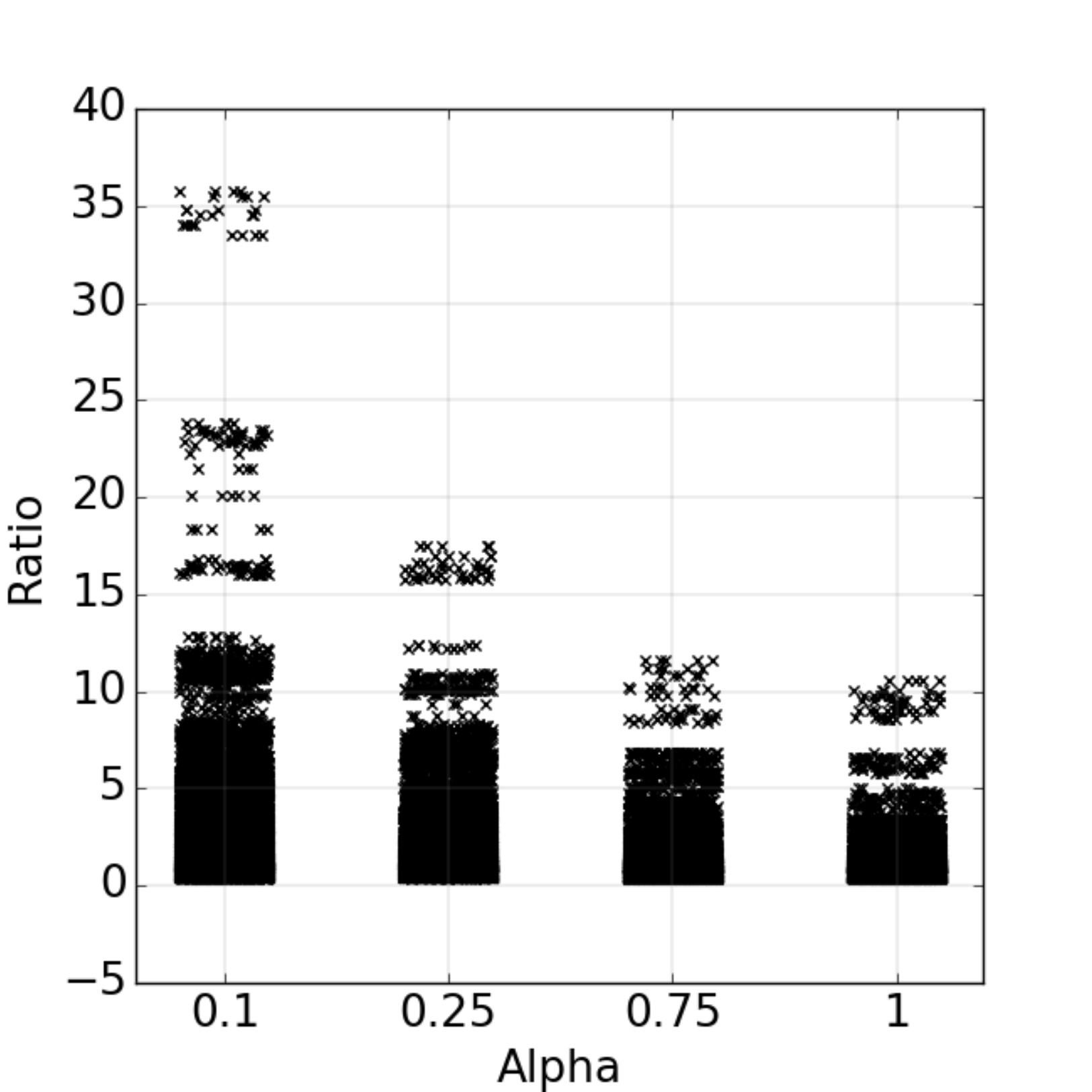}
			\label{fig:rgg-high-cpl}
		}
		\hspace*{\fill}%
		\caption{Comparing the lengths of the critical paths across RGG-classic and RGG-high workloads in terms of $\alpha$ of the application graph}
		\label{fig:alpha-cpl}
	\end{minipage}
	\hfill
	\begin{minipage}{.22\textwidth}
		\includegraphics[width=.97\linewidth]{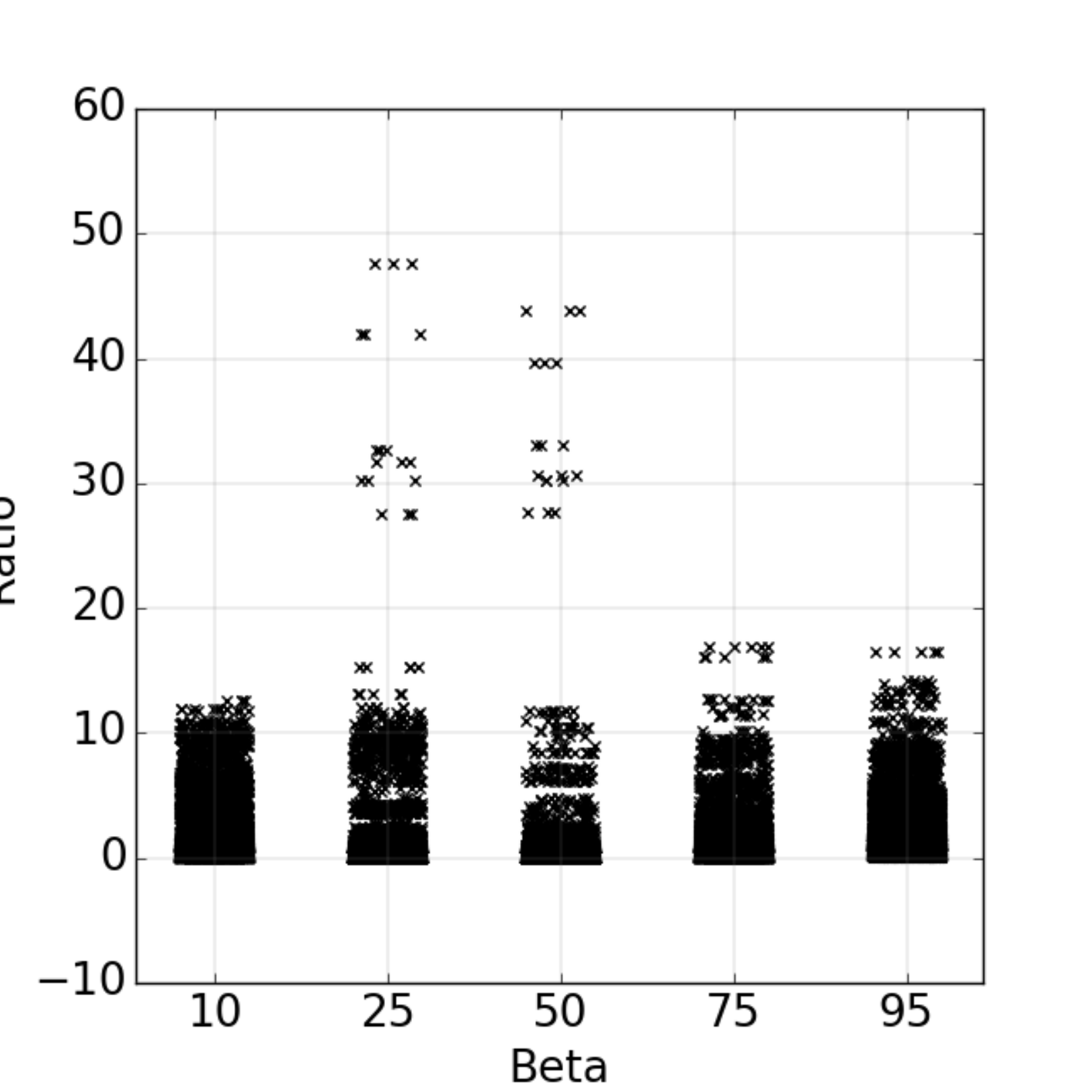}
		\caption{Comparing CPL for RGG-medium in terms of different values of $\beta$ in the input graphs}
		\label{fig:cpl-not-changed-for-beta}
	\end{minipage}
	\hfill
	\begin{minipage}{.22\textwidth}
		\centering
		\includegraphics[width=.97\linewidth]{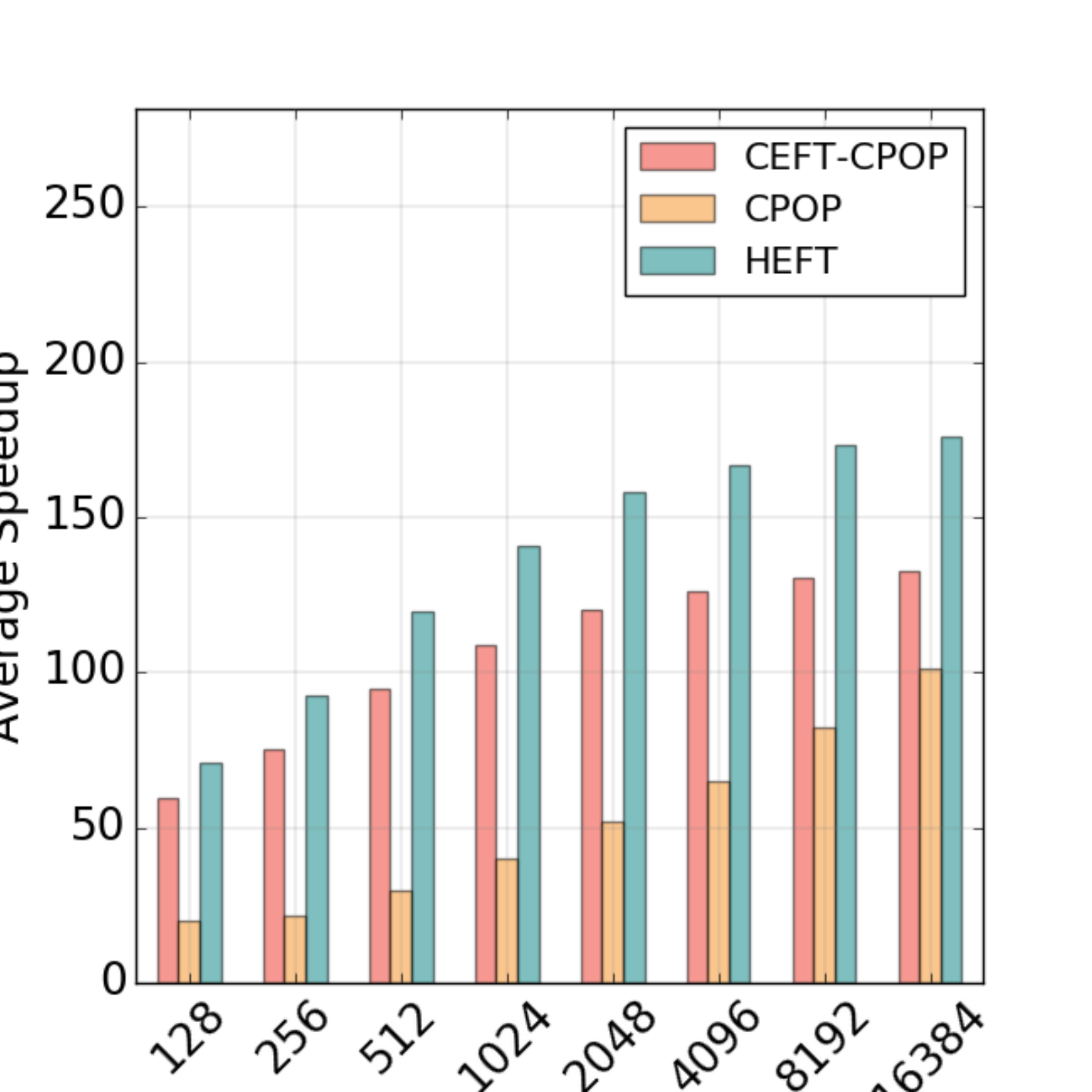}
		\caption{Comparing speedup for RGG-high in terms of number of tasks in the input graphs. Higher is better.}
		\label{fig:task-speedup}
	\end{minipage}%
\end{figure*}
\fi

\ifdefined\longversion
\begin{figure*}[ht]
	\centering
	\hspace*{\fill}%
	\csubfloat[RGG-classic]{
		\includegraphics[width=0.22\linewidth]{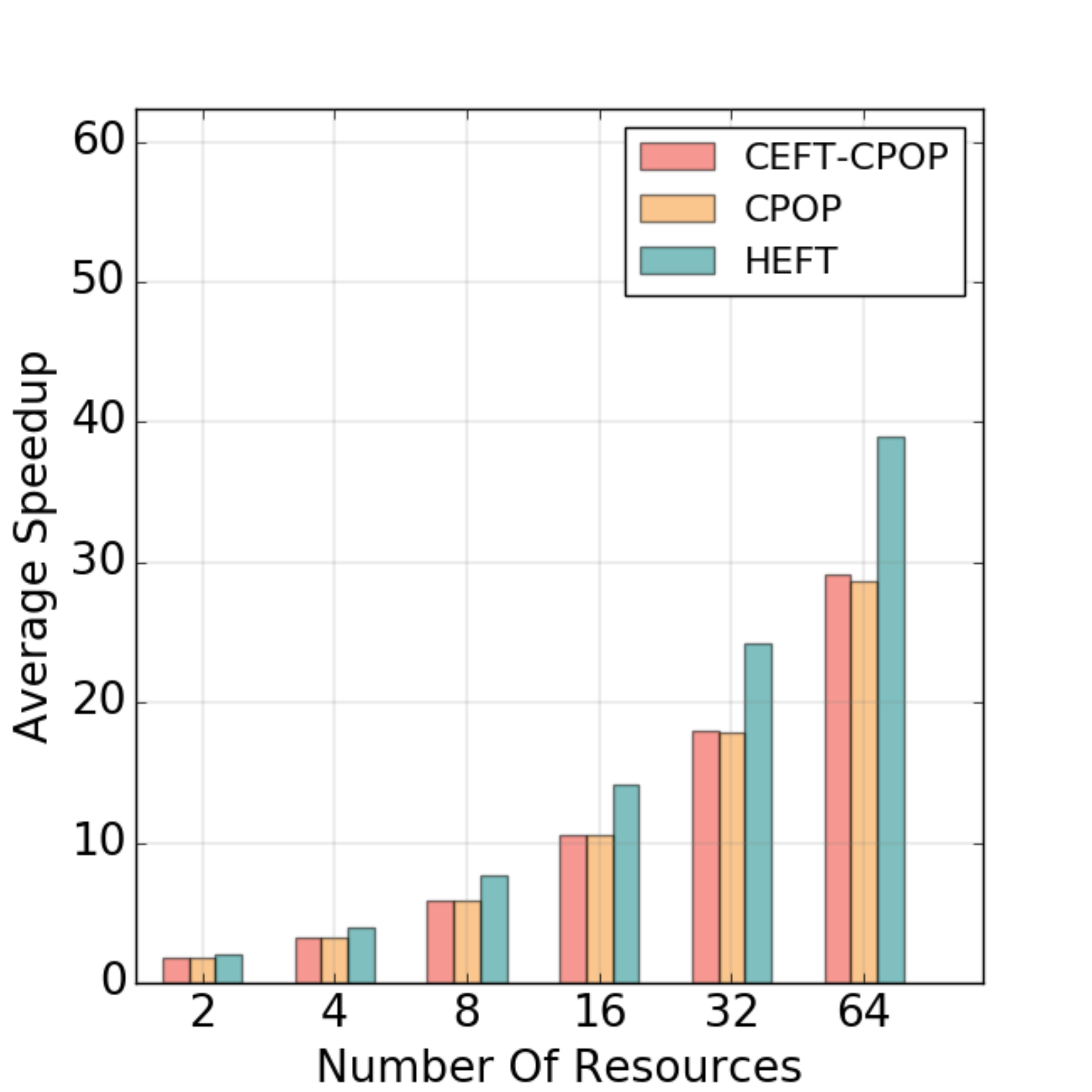}
		\label{fig:rgg-clas-res-speedup}
	}
	\centerhfill
	\csubfloat[RGG-low]{
		\includegraphics[width=0.22\linewidth]{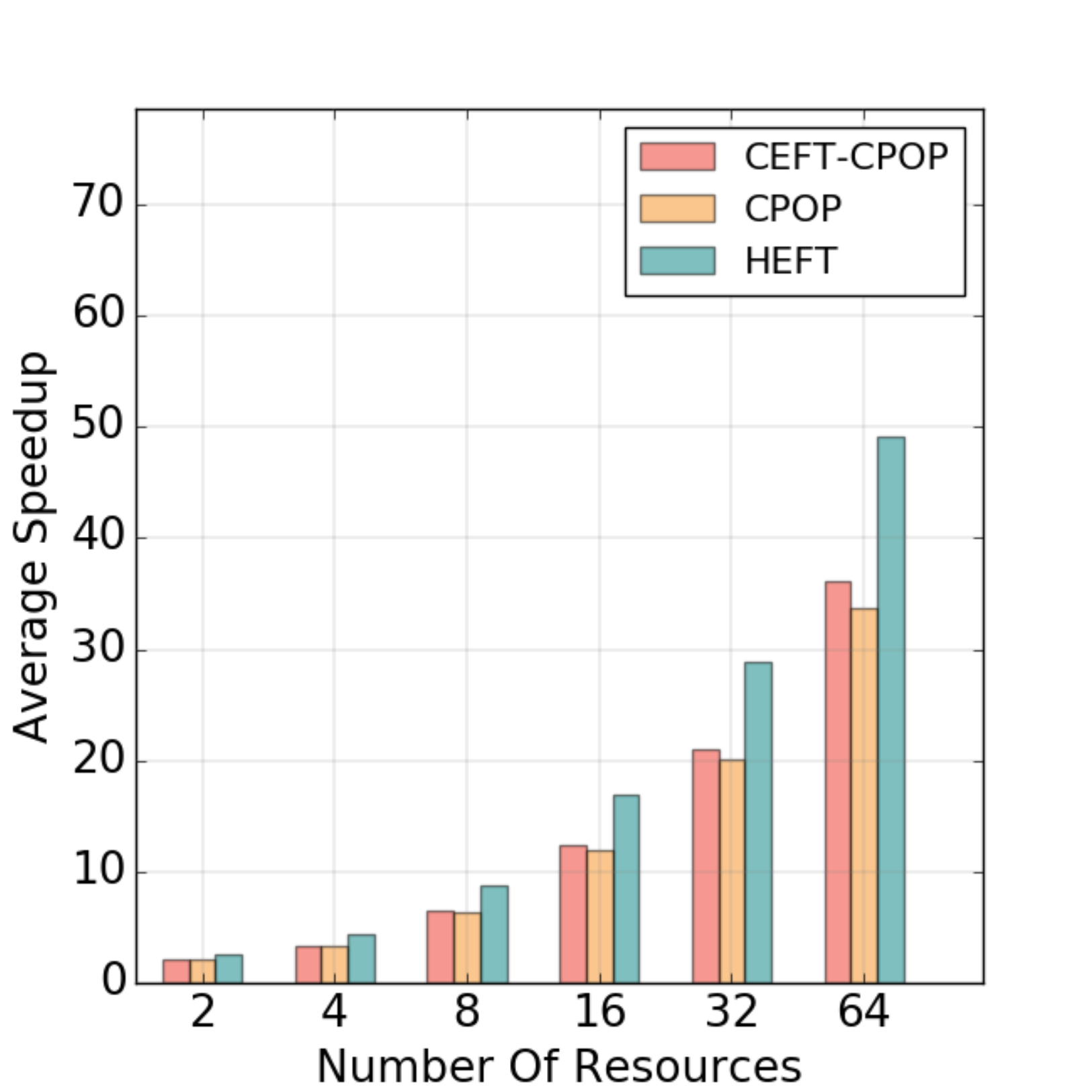}
		\label{fig:rgg-low-res-speedup}
	}
	\centerhfill
	\csubfloat[RGG-medium]{
		\includegraphics[width=0.22\linewidth]{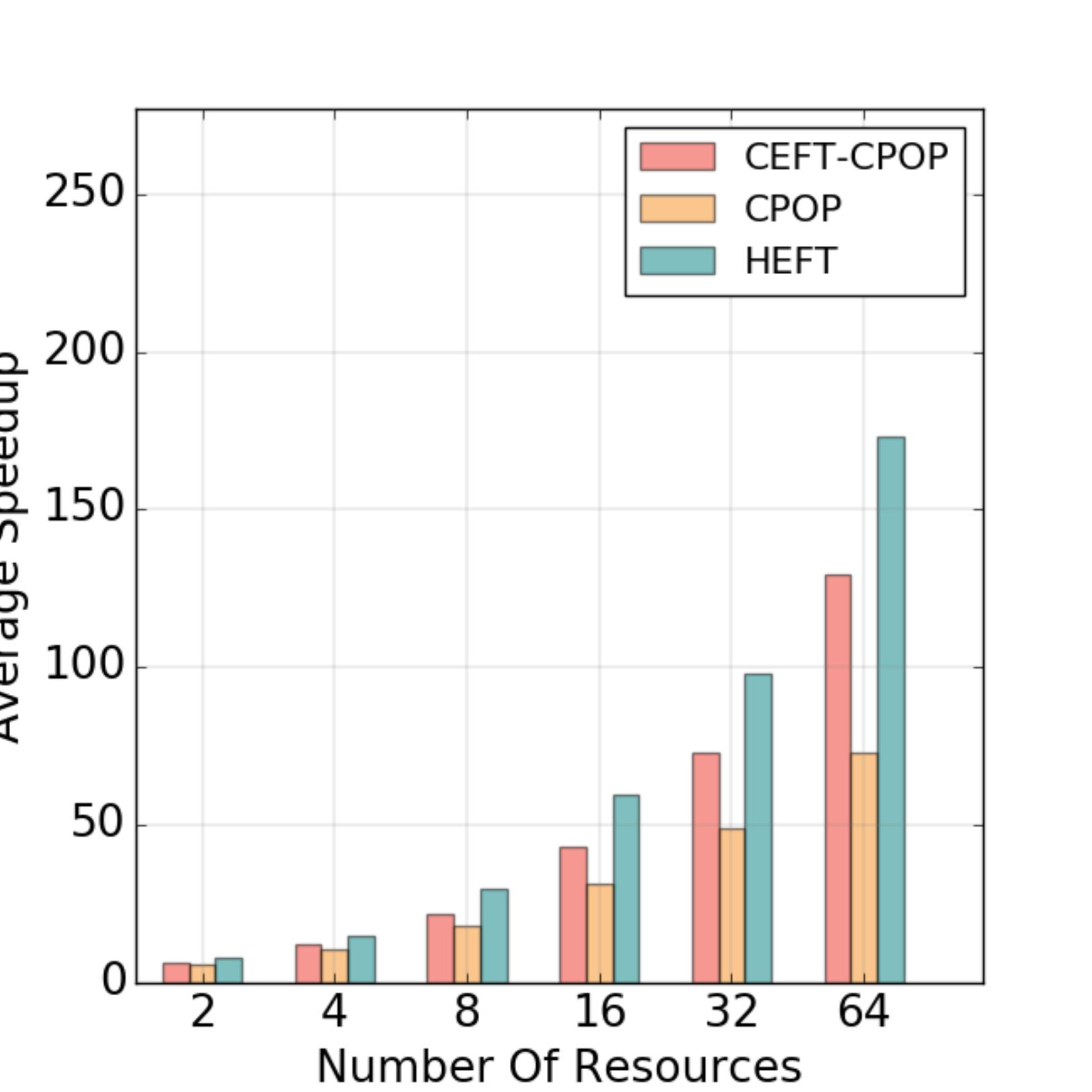}
		\label{fig:rgg-med-res-speedup}
	}
	\centerhfill
	\csubfloat[RGG-high]{
		\includegraphics[width=0.22\linewidth]{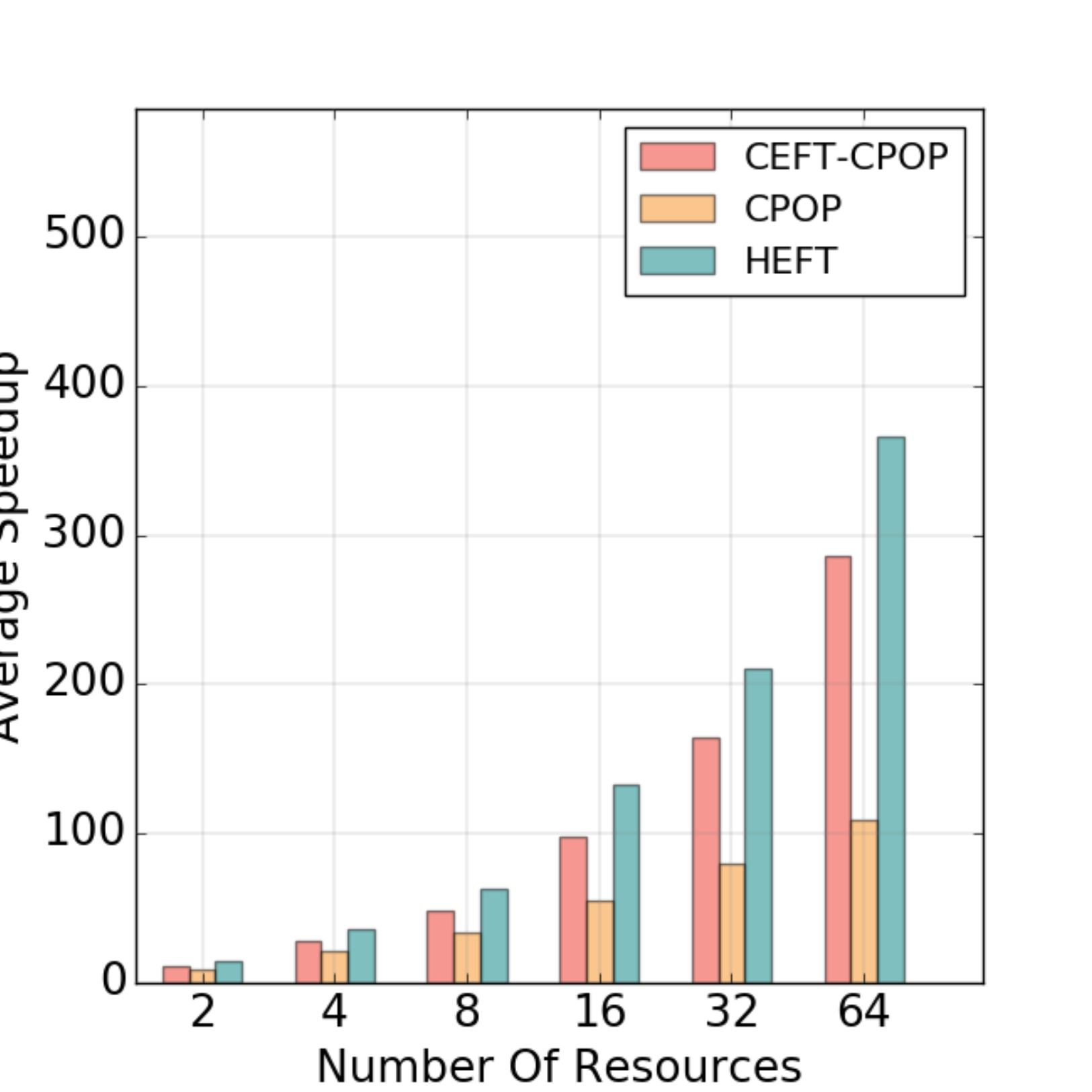}
		\label{fig:rgg-high-res-speedup}
	}
	\hspace*{\fill}%
	\caption{Comparing speedup across different workloads in terms of the number of processors in the processor graph. Higher is better.}
	\label{fig:res-speedup}
\end{figure*}

\fi

It is evident from table~\ref{tab:ceft-cpop} that as heterogeneity in the workload becomes more apparent, CEFT outperforms CPOP in terms of both the critical path length and the makespan. In RGG-classic CEFT never produces a critical path that was shorter than the critical paths produced by CPOP which resulted in makespans that were longer in 26.95\% of the experiments. Another interesting thing to note here is that even in the case where heterogeneity is well expressed (RGG-high) our algorithm does not perform as well as CPOP only in 7.66\% of the experiments. 

\ifdefined\longversion
\else
\begin{figure*}[h]
	\centering
	\hspace*{\fill}%
	\csubfloat[RGG-classic]{
		\includegraphics[width=0.22\linewidth]{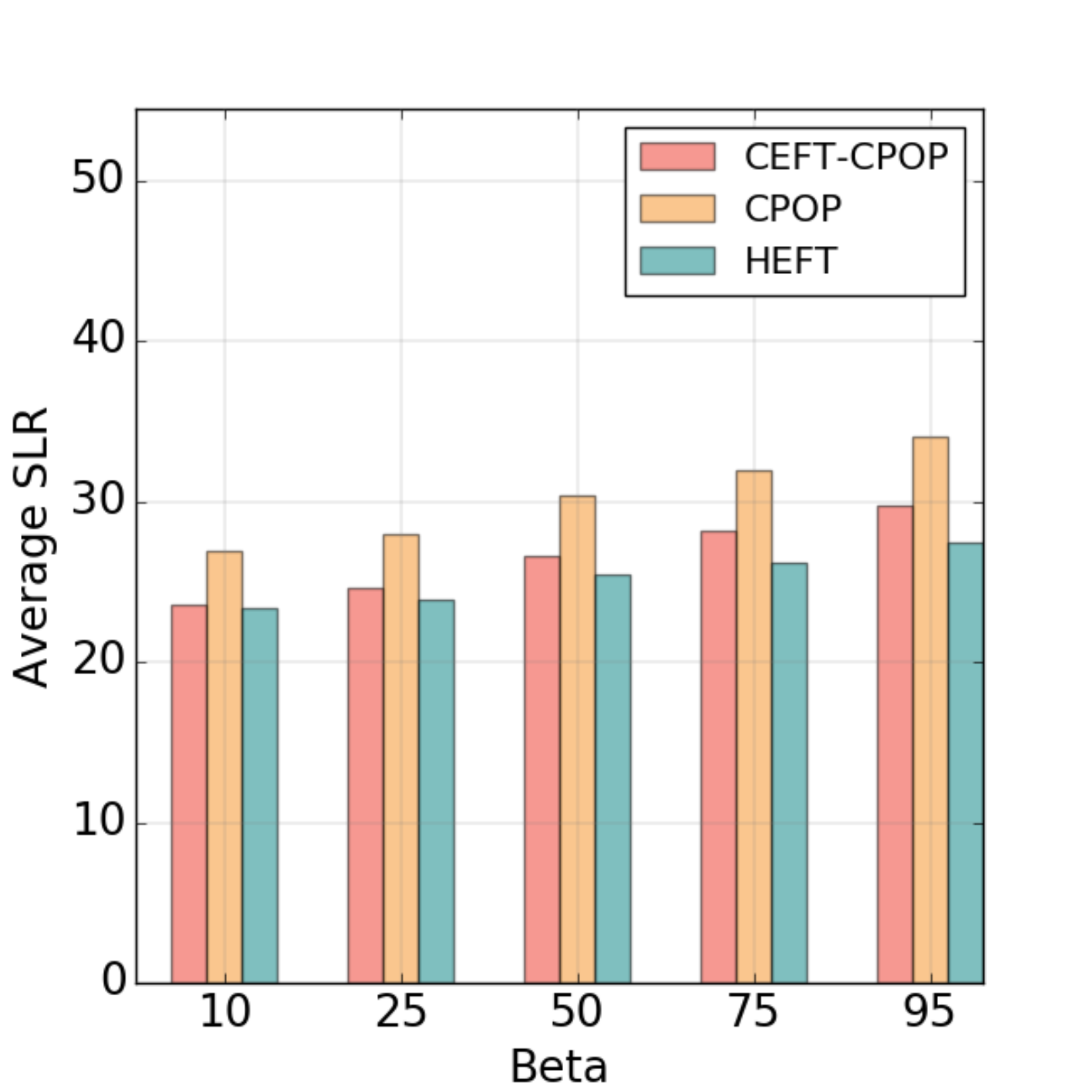}
		\label{fig:rgg-clas-beta-slr}
	}\centerhfill
	\csubfloat[RGG-low]{
		\includegraphics[width=0.22\linewidth]{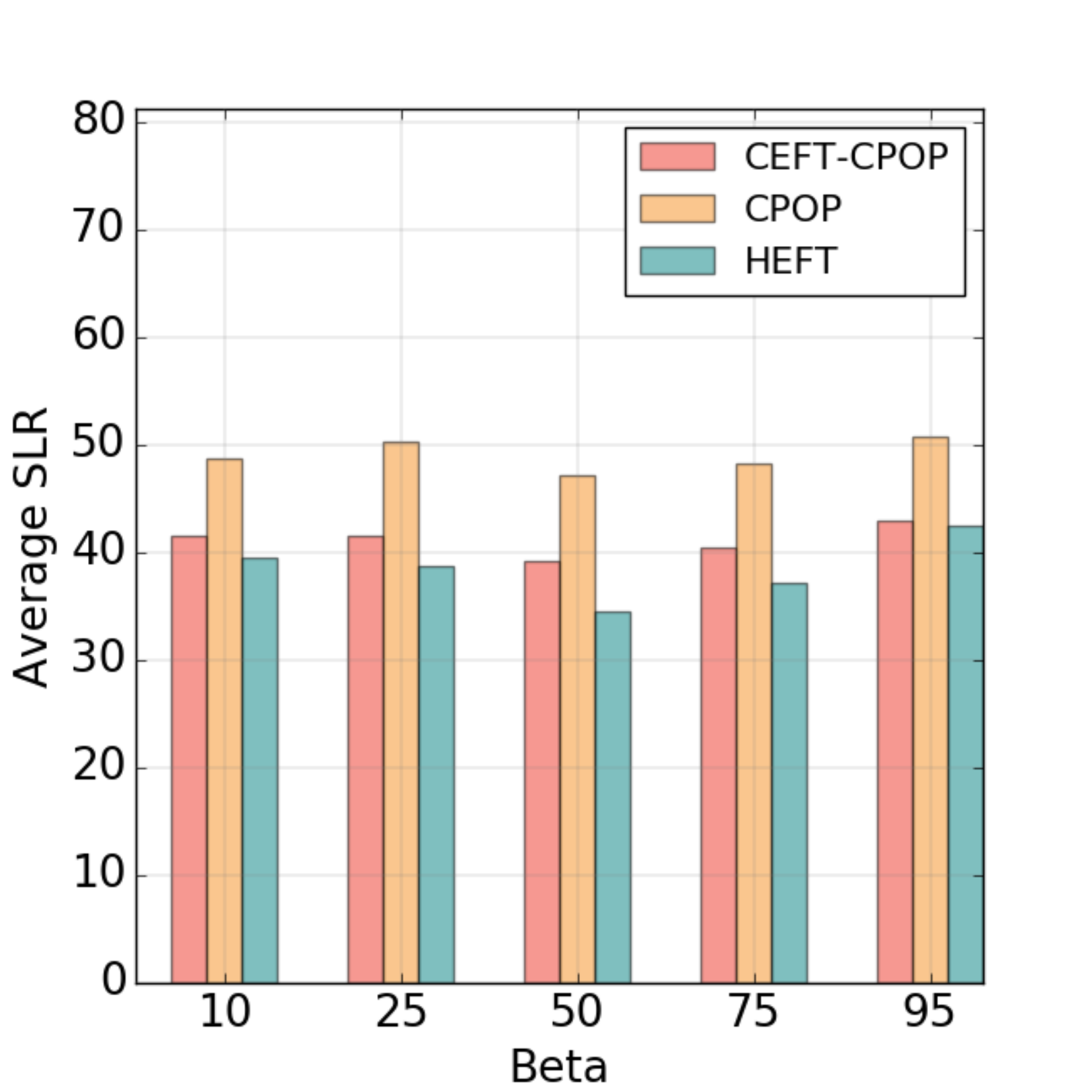}
		\label{fig:rgg-low-beta-slr}
	}\centerhfill
	\csubfloat[RGG-medium]{
		\includegraphics[width=0.22\linewidth]{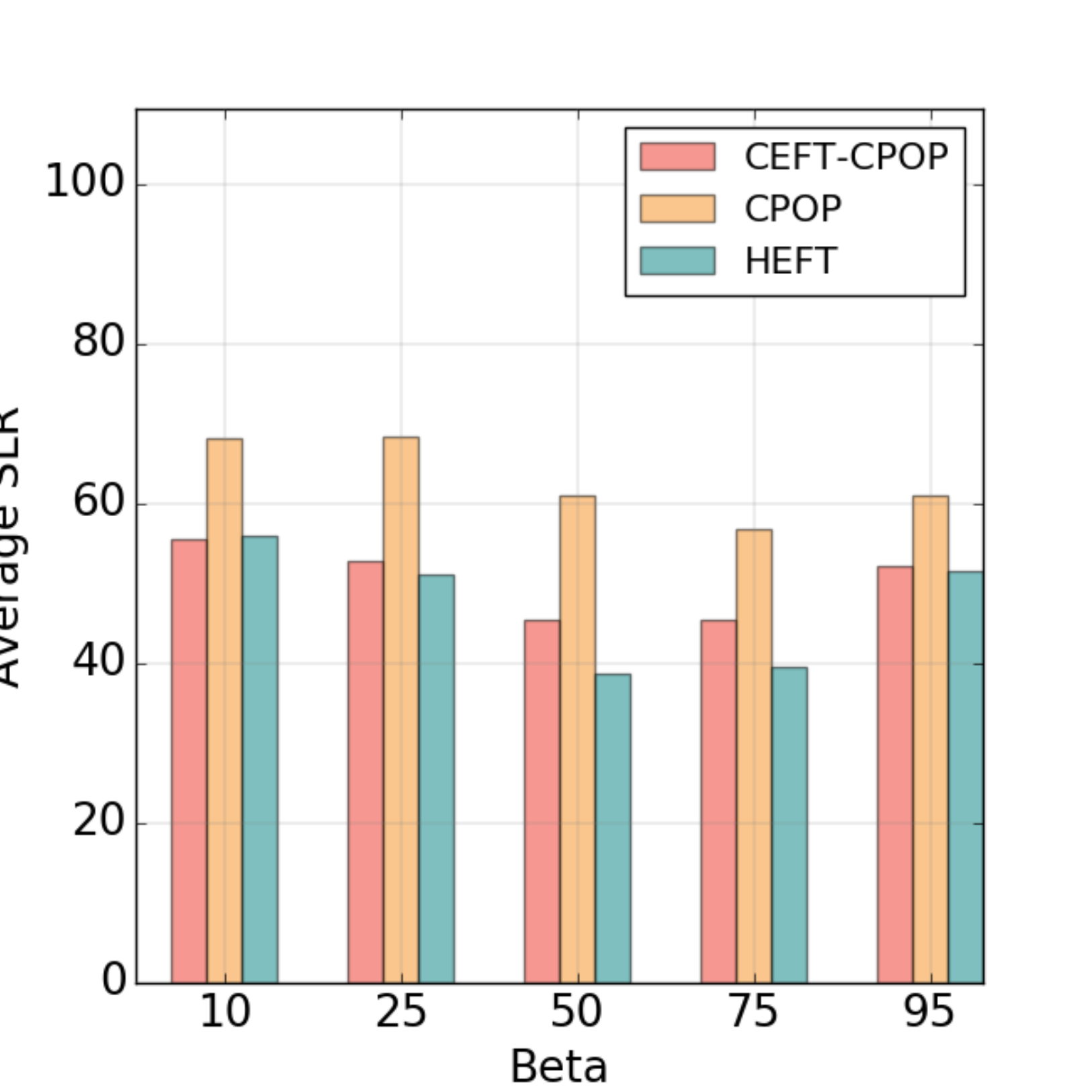}
		\label{fig:rgg-med-beta-slr}
	}\centerhfill
	\csubfloat[RGG-high]{
		\includegraphics[width=0.22\linewidth]{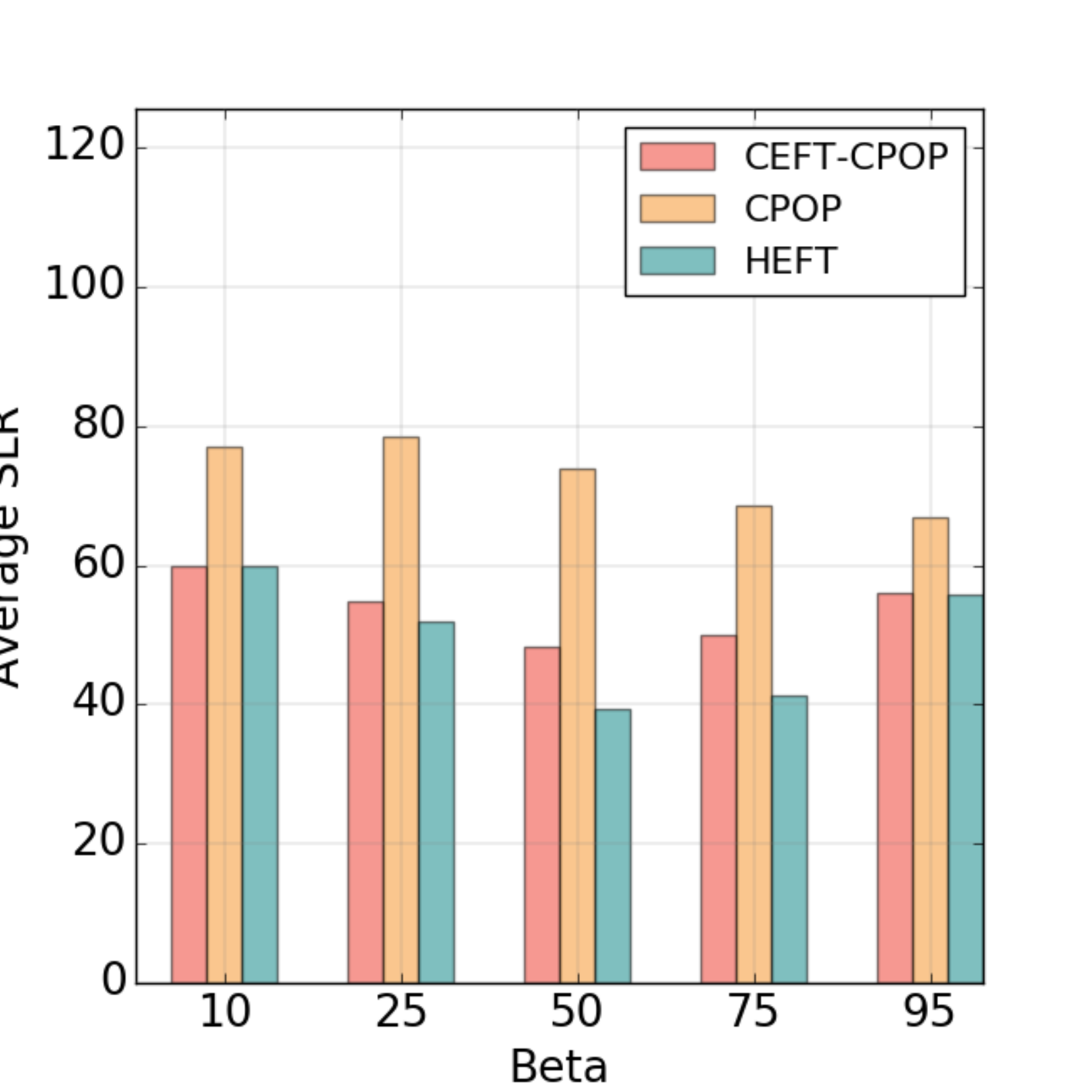}
		\label{fig:rgg-high-beta-slr}
	}\hspace*{\fill}
	\caption{Comparing SLR across different workloads in terms of $\beta$ of the input graphs. Lower is better.}
	\label{fig:Beta-slr}
\end{figure*}
\fi

In the most heterogeneous workload RGG-high, CEFT produces shorter critical path lengths in 83.99\% of the experiments which lead to shorter makespans in 89.69\% of the experiments. There seems to be a strong correlation between shorter critical path lengths and shorter makespans. However, one cannot conclude that shorter critical path lengths result in shorter makespans as it is important to identify the \textit{correct} shorter critical path which would lead to shorter makespans. From the results in this table, our algorithm does well in terms of selecting the correct critical paths.

\ifdefined\longversion

\begin{figure*}[ht]
	\centering
	\hspace*{\fill}%
	\csubfloat[RGG-classic]{
		\includegraphics[width=0.22\linewidth]{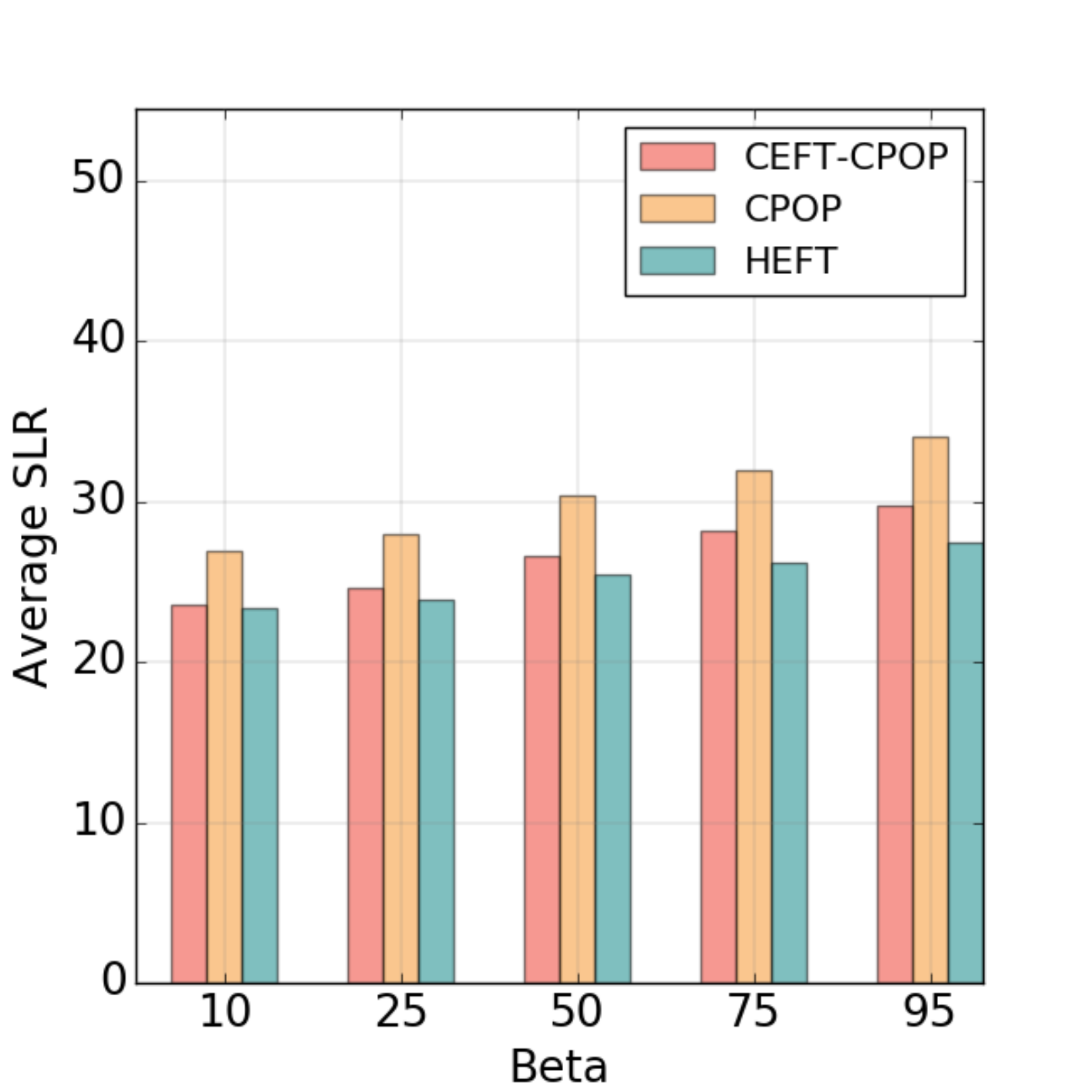}
		\label{fig:rgg-clas-beta-slr}
	}\centerhfill
	\csubfloat[RGG-low]{
		\includegraphics[width=0.22\linewidth]{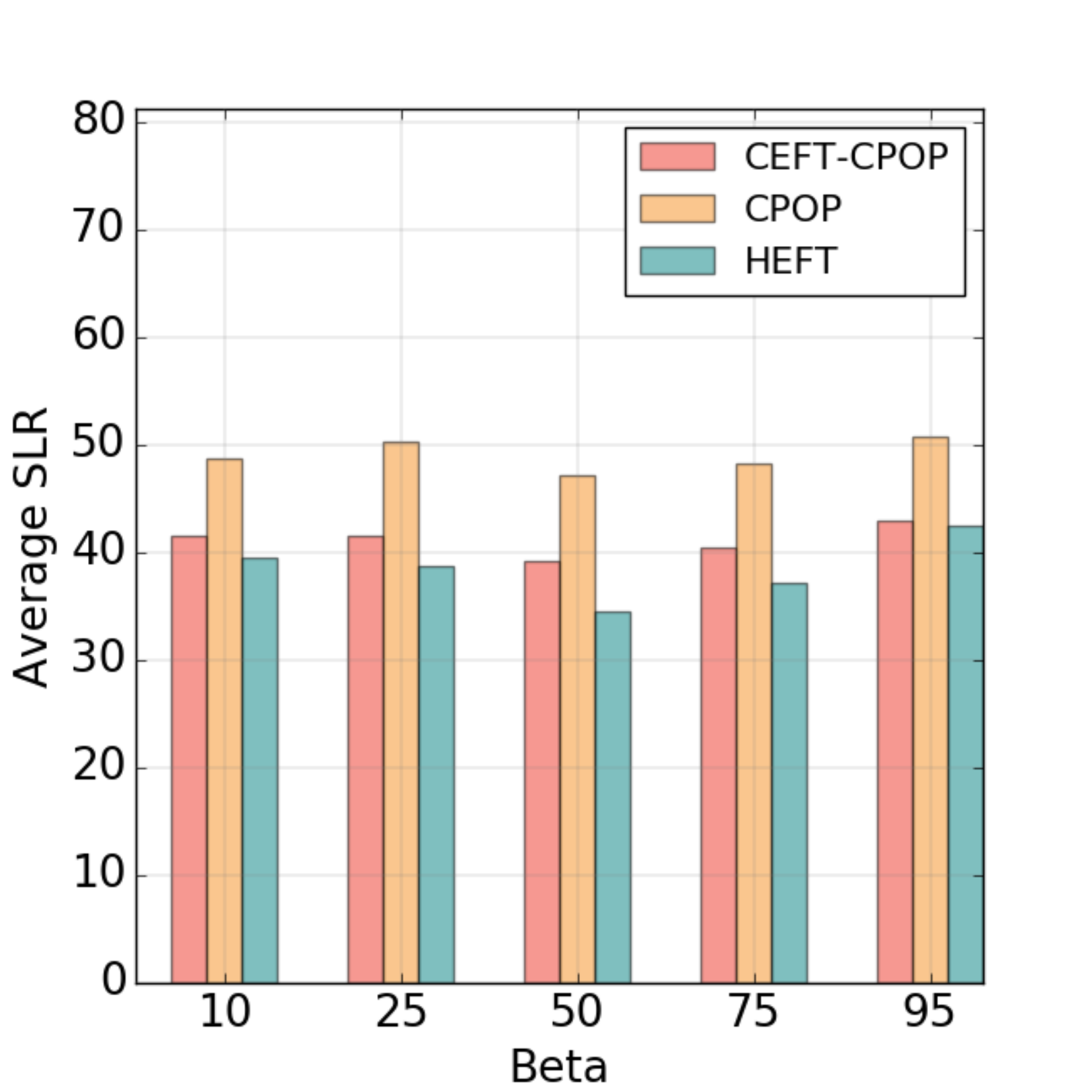}
		\label{fig:rgg-low-beta-slr}
	}\centerhfill
	\csubfloat[RGG-medium]{
		\includegraphics[width=0.22\linewidth]{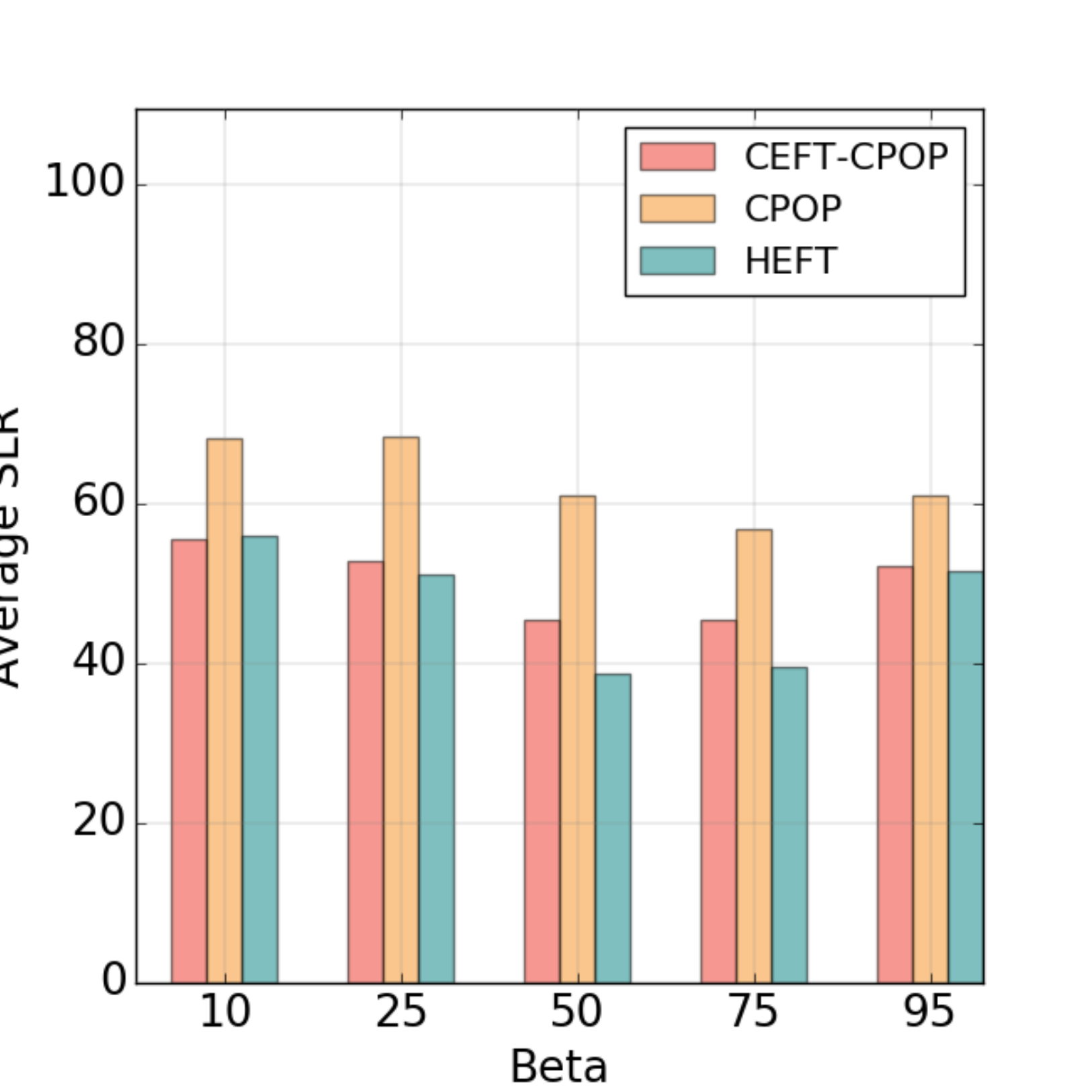}
		\label{fig:rgg-med-beta-slr}
	}\centerhfill
	\csubfloat[RGG-high]{
		\includegraphics[width=0.22\linewidth]{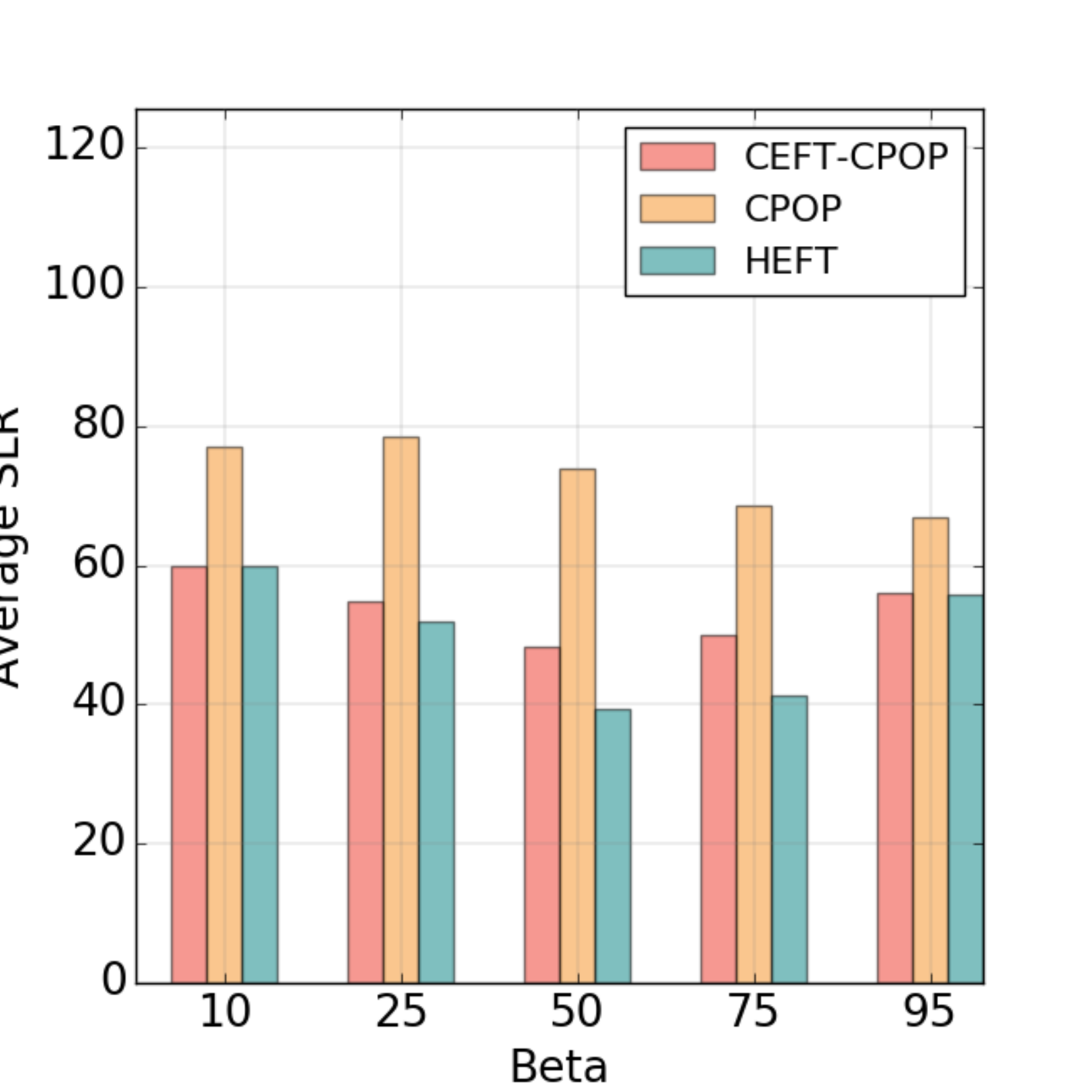}
		\label{fig:rgg-high-beta-slr}
	}\hspace*{\fill}
	\caption{Comparing SLR across different workloads in terms of $\beta$ of the input graphs. Lower is better.}
	\label{fig:Beta-slr}
\end{figure*}

\begin{figure*}[t]
	\centering
	\hspace*{\fill}%
	\csubfloat[RGG-classic]{
		\includegraphics[width=0.22\linewidth]{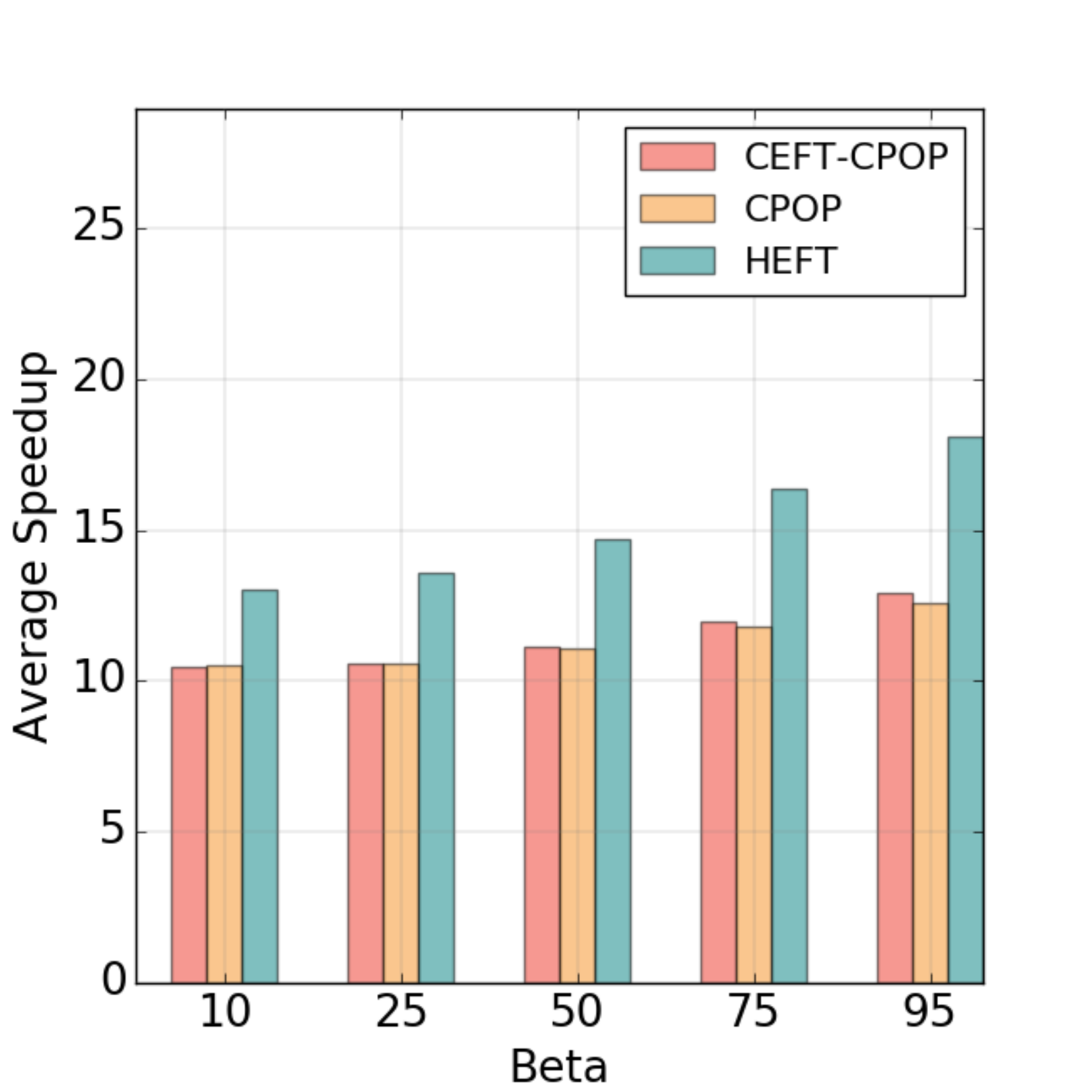}
		\label{fig:rgg-clas-beta-speedup}
	}
	\centerhfill
	\csubfloat[RGG-low]{
		\includegraphics[width=0.22\linewidth]{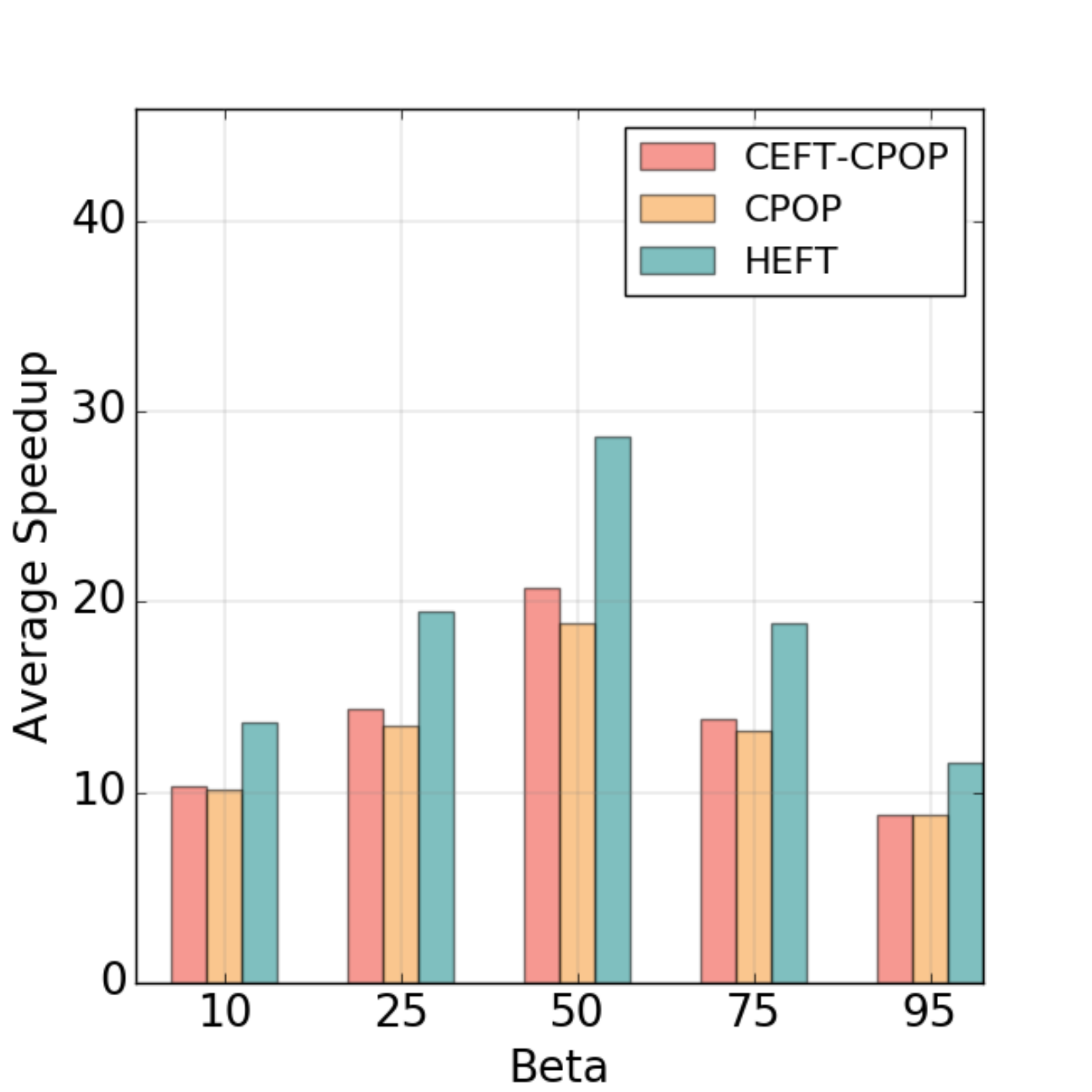}
		\label{fig:rgg-low-beta-speedup}
	}
	\centerhfill
	\csubfloat[RGG-medium]{
		\includegraphics[width=0.22\linewidth]{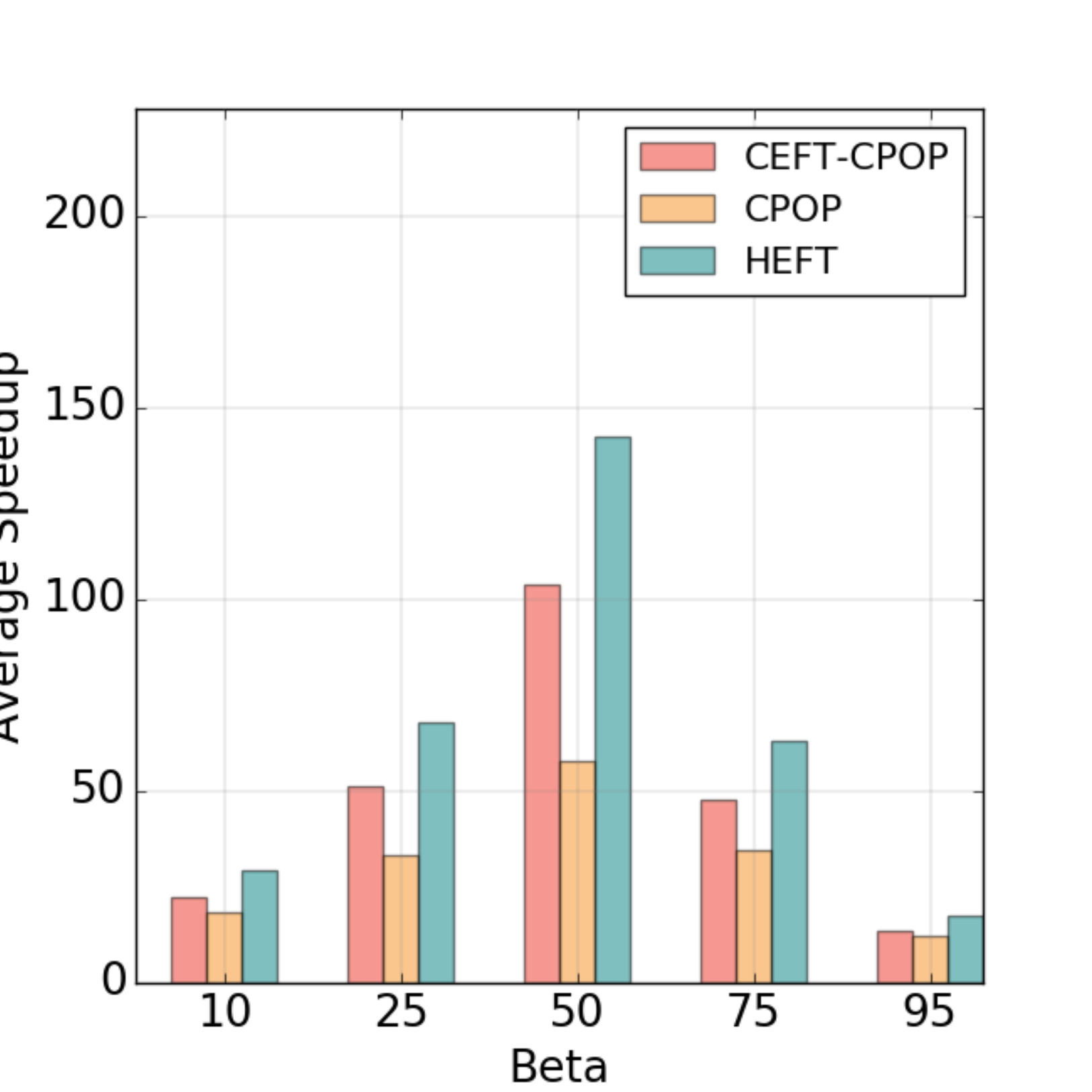}
		\label{fig:rgg-med-beta-speedup}
	}
	\centerhfill
	\csubfloat[RGG-high]{
		\includegraphics[width=0.22\linewidth]{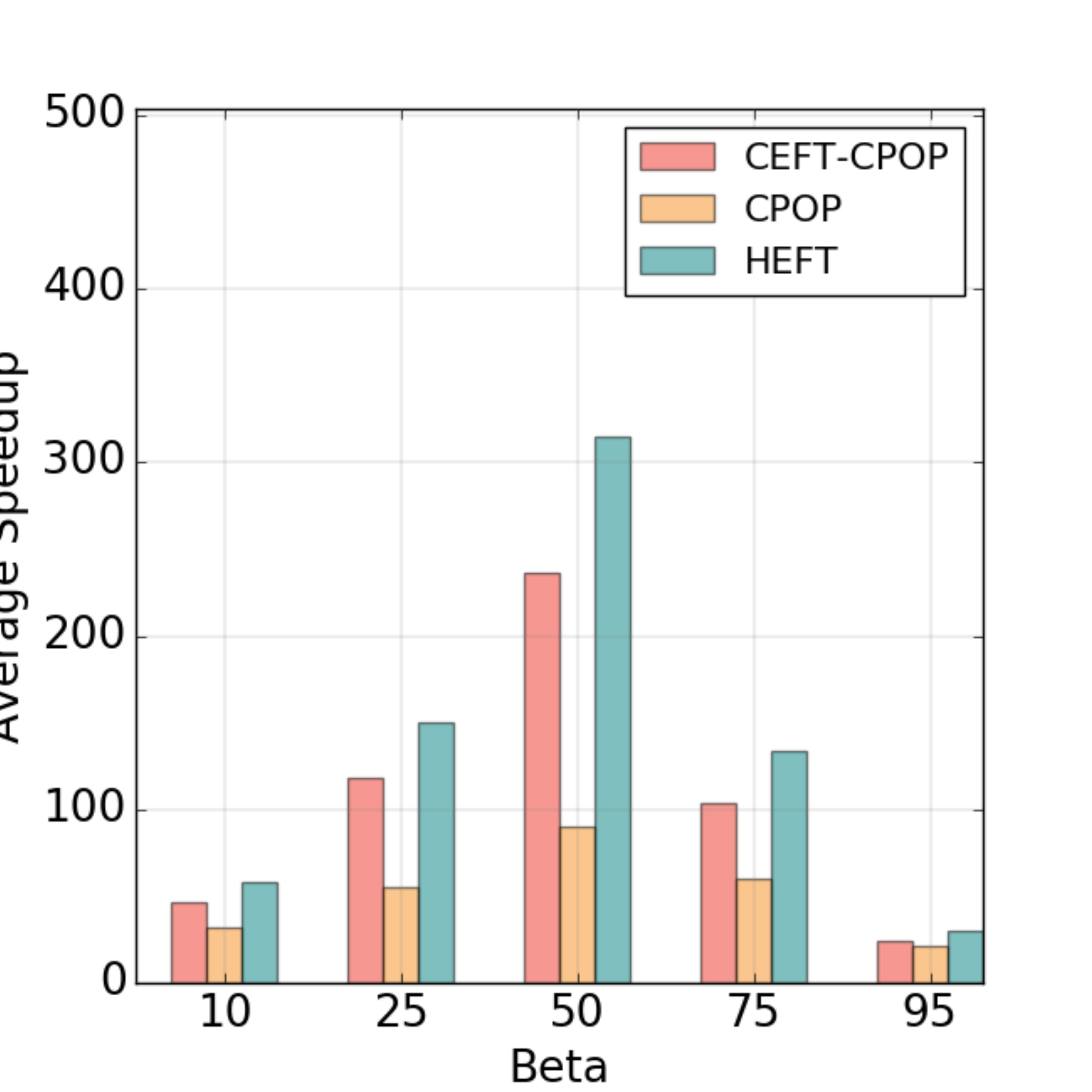}
		\label{fig:rgg-high-beta-speedup}
	}
	\hspace*{\fill}%
	\caption{Comparing speedup across different workloads in terms of $\beta$ of the input graphs. Higher is better.}
	\label{fig:beta-speedup}
\end{figure*}
\fi

Figure~\ref{fig:res-speedup} shows how the speedup metric \ifdefined\longversion from section~\ref{sec:paper-crit-path-effi} \fi varies when the number of processor is varied. As the number of processors is increased, the average speedup achieved is naturally higher which is clearly reflected in the graphs. In the standard workload, RGG-classic all the algorithms perform nearly equally. When the heterogeneity is increased, it is evident from figures~\ref{fig:rgg-med-res-speedup} and \ref{fig:rgg-high-res-speedup} that the average speedups achieved by CPOP become progressively lesser with increase in number of processors. This is mainly because of the fact that CPOP assigns all the tasks from the critical path onto a single processor. 

The choice of assigning the entire critical path on a single processor might prove to be an excellent choice in scenarios where the communication-to-computation ratio (CCR) is very high, thereby making communication costs high. But in the majority of the cases, where the computation costs dominate the communication costs, this proves to be a wrong decision and the makespans suffer in kind. This trend of CPOP not being able to catch-up with CEFT-CPOP is evident in graphs based on metrics like number of tasks, $\alpha$, $\beta$ etc. Figure~\ref{fig:task-speedup}, highlights another interesting result. Upon careful inspection, the average speedup of CEFT-CPOP is the highest among all three comparison algorithms, until the number of tasks cross 1024 (incidentally, 512 is the highest number of tasks in synthetic workloads that have been used for testing the efficiency of scheduling algorithms previously~\cite{arabnejad2014list}).


In figure~\ref{fig:Beta-slr}, we compare the schedule length ratio metric across the different workloads through this context. It is evident from the graphs that RGG-classic and RGG-low exhibit similar SLR patterns (RGG-low has a slightly lower SLR value on average). It is interesting to note however, that in RGG-medium and RGG-high, our algorithm produces the lowest average SLR value when $\beta\sim=50$. Setting $\beta$ to values close to 50 during the input graph generation, generates a good mix of tasks that require the different types of processors in the processor graph\footnote{This is better understood by referring back to our discussion of the two intervals $\mathcal{I}_1$ and $\mathcal{I}_2$ from section~\ref{sec:rand-gene-grap}. When $\beta\approx50$, there are approximately an equal number of tasks that use $\mathcal{I}_1$ for its first node-weight $\mathcal{I}_2$ for its second, as the number of tasks that use it vice-versa. This results in the most varied execution times across tasks. When $\beta$ goes farther away from the mean value of 50, it results in graphs that have more tasks that confirm to a specific ordering of the two intervals, thereby making the tasks more similar which leads to a less varied execution time table.}.

When there is such a good mix of the types of tasks, it is easier to schedule it onto  different types of processors as contention for the same kind of processor would be low. However, when $\beta$ goes away from 50 either side, it leads to increased demand for a certain type of processor, hence increasing the contention among tasks. This leads to increased makespan values when $\beta$ is farther away from 50. However, the critical path lengths calculated by our algorithm remains unaffected as shown figure~\ref{fig:cpl-not-changed-for-beta}, as one does not need to account for processor availability while calculating the critical path (since, this is equivalent to scheduling a linear DAG where all processors are available whenever a task is ready to be scheduled.

\ifdefined\longversion
This notion of the graph generator producing graphs that could potentially lead to lower makespans is further accentuated by figure~\ref{fig:beta-speedup}. In RGG-classic, since heterogeneity is incorporated differently (recall our discussion about the conventional way of implementing a random graph generator, like the one presented by Topcuoglu et al. in  \cite{topcuoglu2002performance}) the speedup values across the different algorithms are very similar and we do not observe the U-shaped curve from figure~\ref{fig:rgg-high-beta-slr}. However, in the heavier workloads, we can clearly see the curve forming again. Once again, CPOP's method of calculating the critical path lets it down. As CPOP assigns all the tasks from the critical path (which might be composed of any number of different types of tasks, i.e. tasks requiring different amounts of the different types of processors) onto the same processor, the makespan suffers which leads to reduced speedup.
\fi

Another important parameter in the graph generation process is $\alpha$ which dictates the \textit{width} of the graph. Lower values of $\alpha$ produce tall skinny graphs, while larger values produce larger wide graphs. It is evident from figure~\ref{fig:alph-slr}, that the average SLR produced by the schedules found by our algorithm are lower than both CPOP and HEFT for all the different values of $\alpha$.

\begin{figure*}[t]
	\centering
	\hspace*{\fill}%
	\csubfloat[Alpha - SLR]{
		\includegraphics[width=0.25\linewidth]{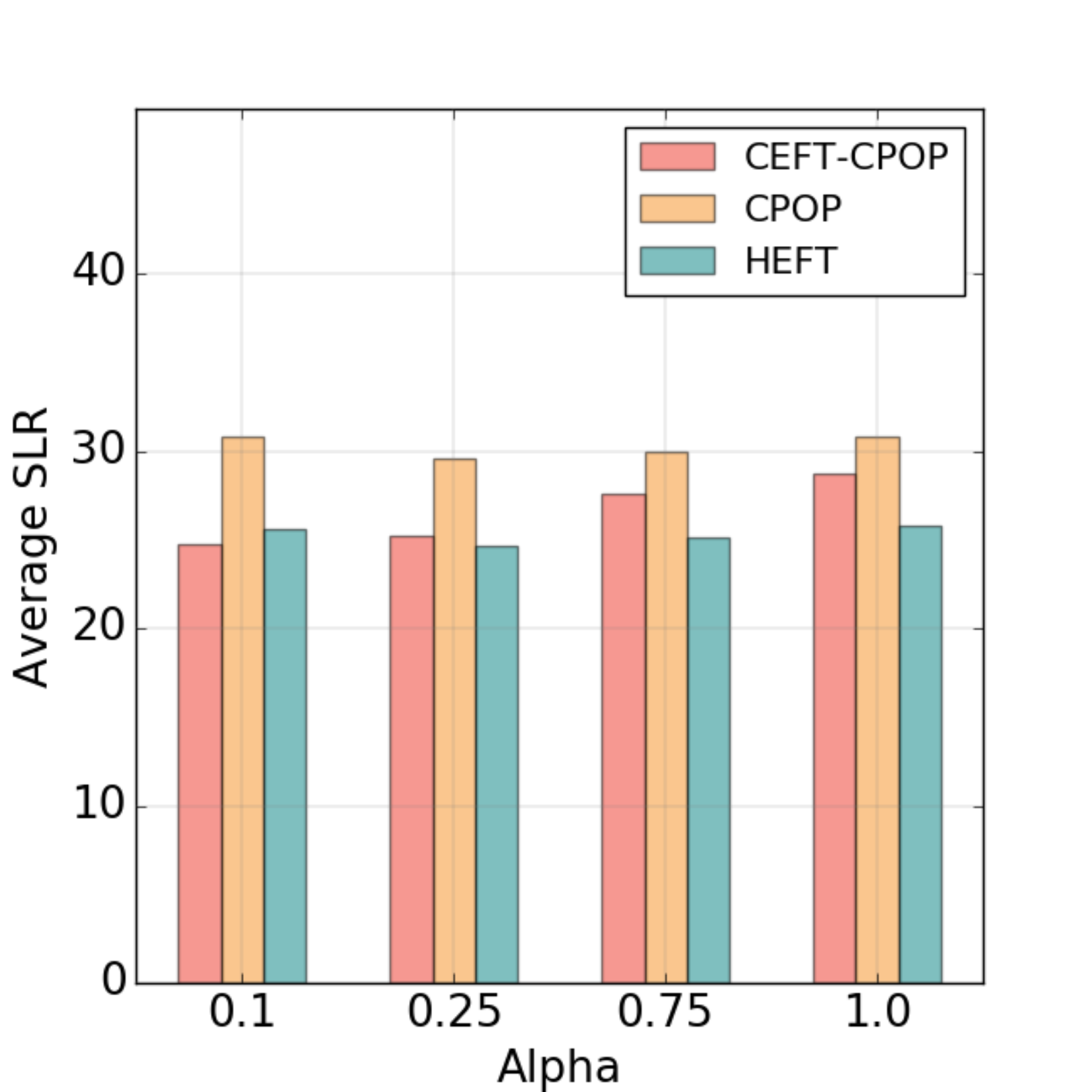}
		\label{fig:alph-slr}
	}%
	\centerhfill
	\csubfloat[CCR - SLR]{
		\includegraphics[width=0.25\linewidth]{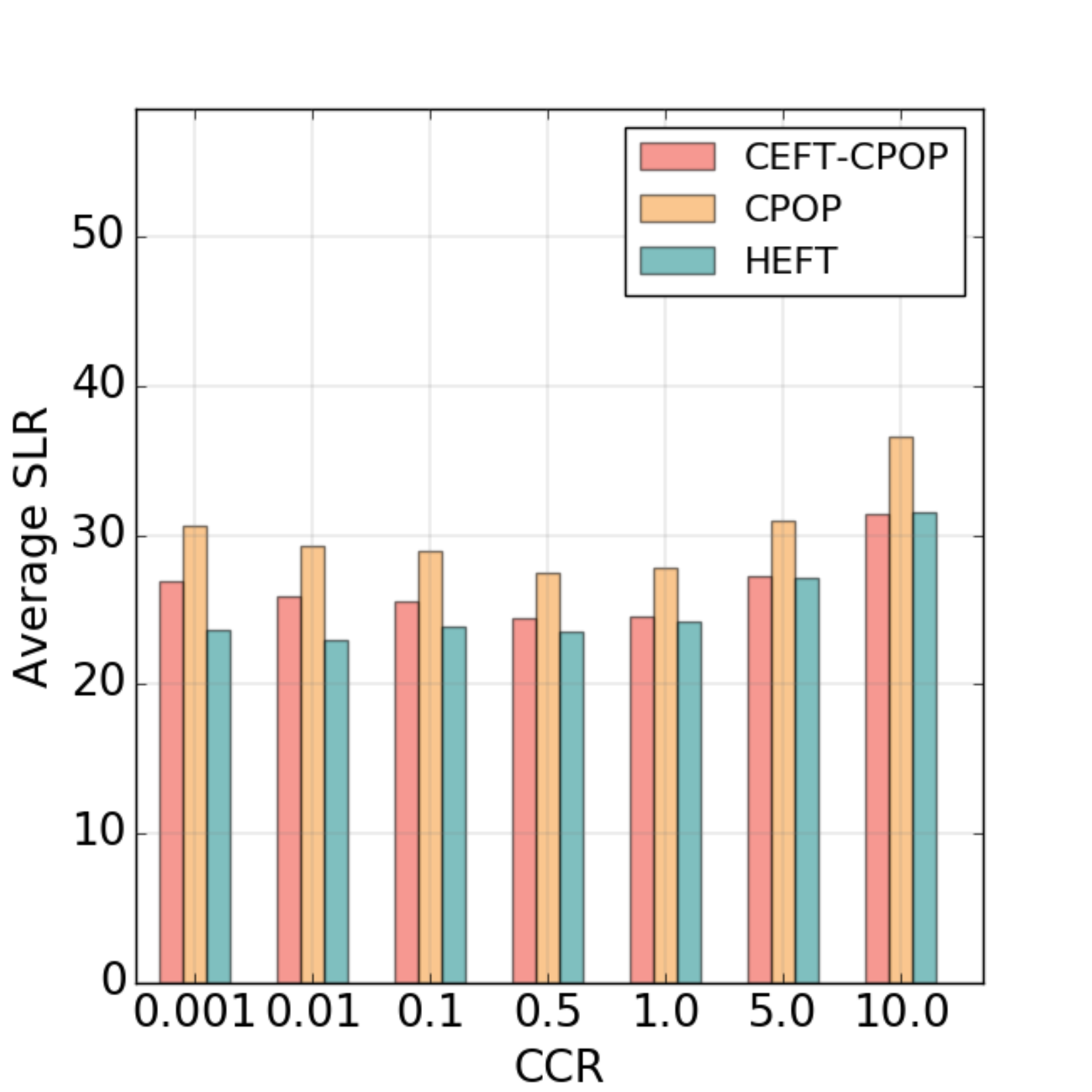}
		\label{fig:ccr-slr}
	}%
	\centerhfill
	\csubfloat[CCR - Slack]{
		\includegraphics[width=0.25\linewidth]{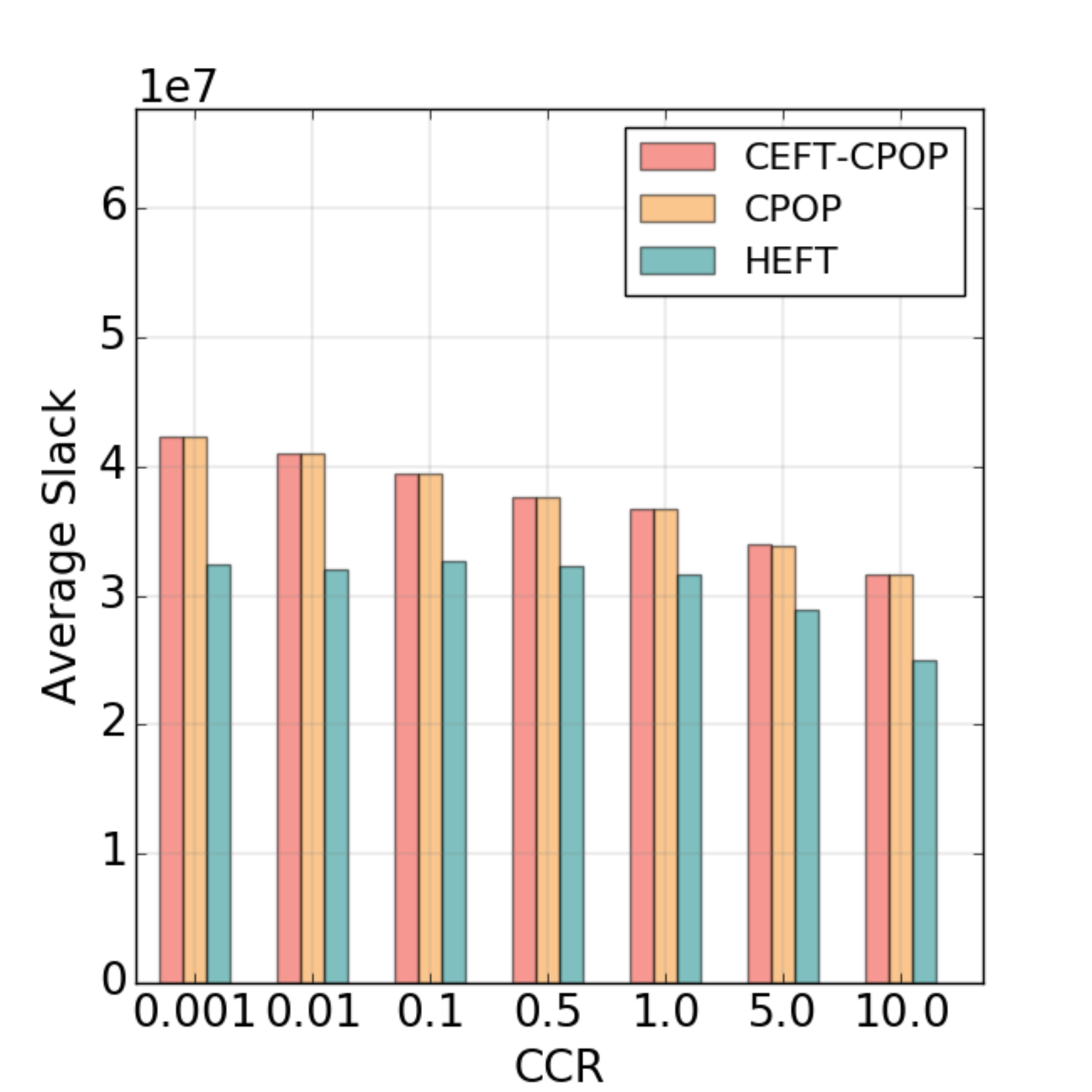}
		\label{fig:ccr-slack}
	}%
	\hspace*{\fill}%
	\caption{Comparing Slack and SLR for RGG-classic in terms of $\alpha$ and CCR of the input graphs. Lower is better for SLR.}
	\label{fig:alph-ccr-slr}
\end{figure*}

For smaller values of $\alpha$, our algorithm's SLR is lower than CPOP's SLR by $\sim$19\% and lower than HEFT's SLR by $\sim$6\%, which we denote by the tuple $[19, 6]$. This gap however reduces as the graph becomes wider, with the gap between our algorithm and HEFT vanishing at high values of $\alpha$. We do not highlight the results obtained from the other workloads as they exhibit similar patterns to RGG-classic.


As the workload gets heavier (RGG-classic to RGG-high) the gap between average SLR produced by our algorithm (CEFT-CPOP) and CPOP is increased. In the case of RGG-high, for smaller values of $\alpha$, CEFT-CPOP's SLR is lower than CPOP's SLR by $\sim$34\%. 

In terms of robustness, the value of slack increases for all three algorithms for increasing values of $\alpha$. The schedules produced for thinner graphs have a lower tolerance to accommodate delays in execution of certain tasks. In the trivial case of the thinnest graph (which is a linear DAG), any schedule produced by a static scheduling algorithm will have zero slack as there is no possibility to overlap computation and communication (due to the serial nature of the graph). As the graph gets wider, there is more scope for overlapping computation with communication which in turns helps in schedules being more accommodative to delays in task execution; this in turn increases the slack as the graphs get wider. 


\begin{figure}[b]
	\centering
	\hspace*{\fill}%
	\csubfloat[Number of tasks]{
		\includegraphics[width=0.49\linewidth]{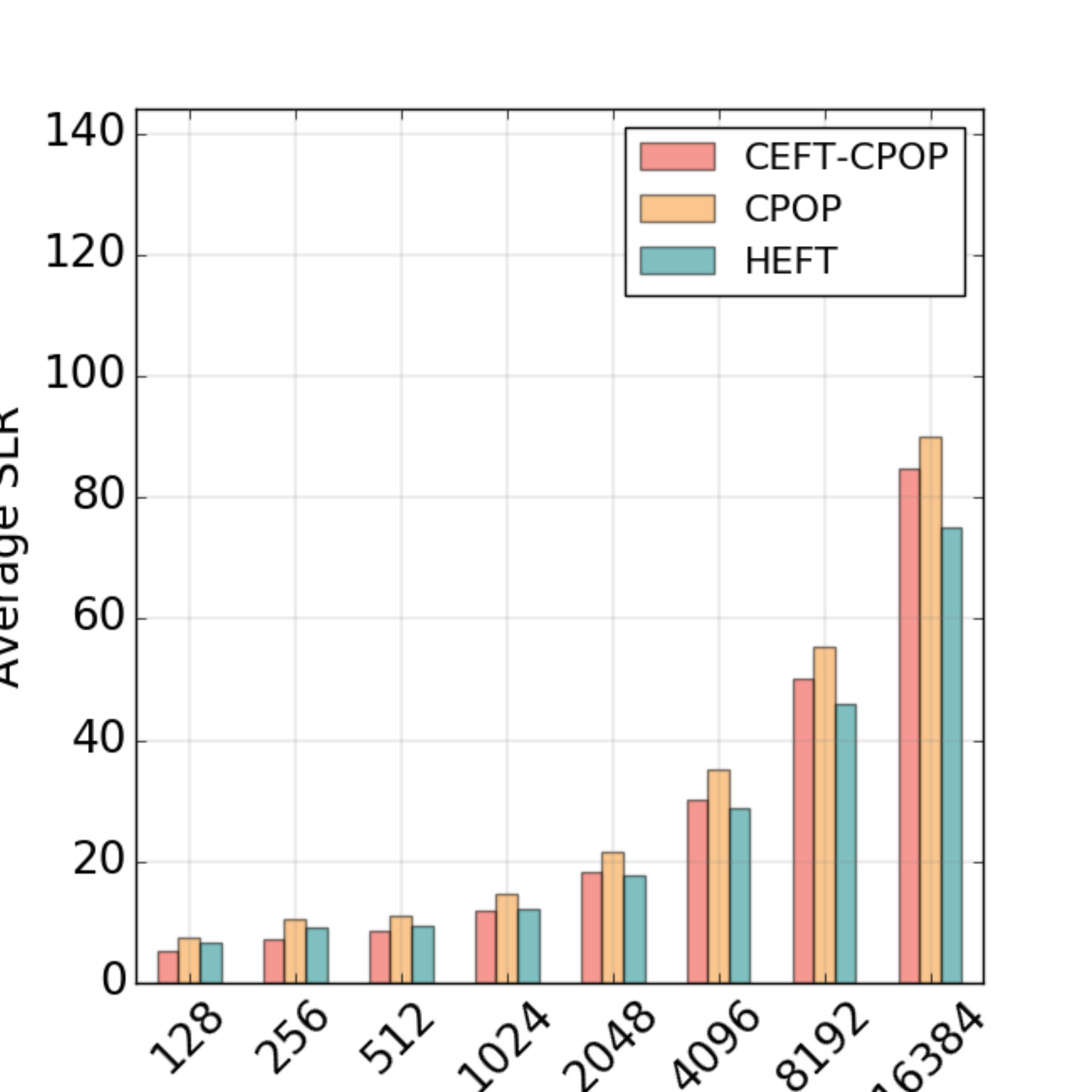}
		\label{fig:rgg-clas-task-slr}
	}
	\centerhfill
	\csubfloat[Number of resources]{
		\includegraphics[width=0.49\linewidth]{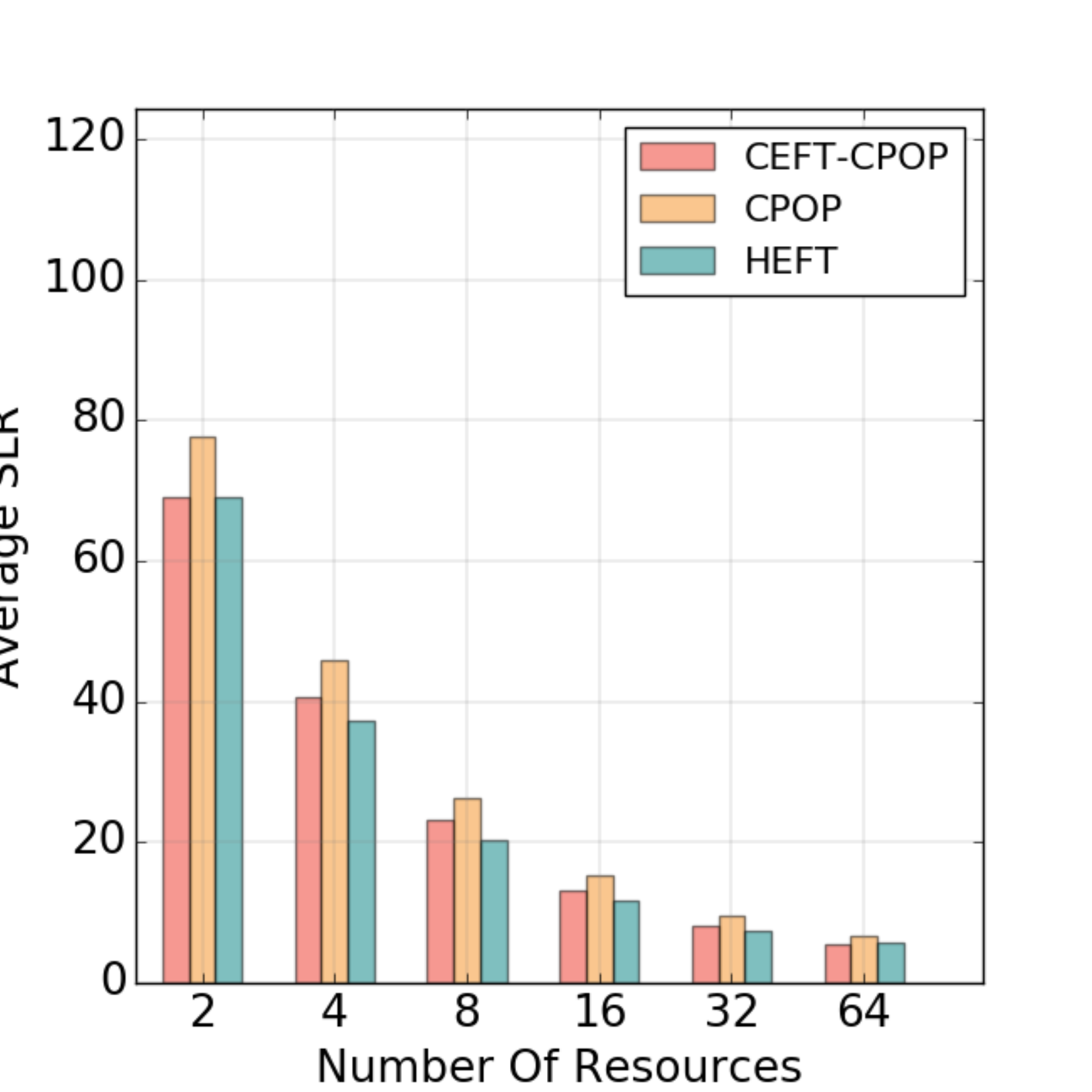}
		\label{fig:rgg-clas-res-slr}
	}
	\hspace*{\fill}%
	\caption{Comparing SLR across different workloads in terms of number of tasks and resources. Lower is better.}
	\label{fig:task-res-slr}
\end{figure}

As the value of the communication-to-computation ratio (CCR) increases, interprocessor communication overhead dominates computation and hence, the performance of all three scheduling algorithms tends to degrade. This is shown by the use of the schedule length ration (SLR) metric in figure~\ref{fig:ccr-slr}. Our algorithm produces SLRs which are lower than CPOP's SLR by $\sim13\%$, for lower CCR values with HEFT producing better average SLRs. CEFT-CPOP continues to outperform CPOP for all the CCR values and produces similar average SLR to HEFT, for extremely large CCR values. 

Slack on the other hand, which is a measure of robustness\ifdefined\longversion (section~\ref{sec:paper-crit-path-slac-metr})\fi, decreases for all the three scheduling algorithms with increasing CCR values. This trend of the schedules becoming less tolerant to delays for increasing values of CCR is highlighted in figure~\ref{fig:ccr-slack}. Our algorithm provides the highest slack from the three algorithms compared here, while HEFT provides the lowest\ifdefined\longversion\footnote{Heterogeneous Earliest Finish Time (HEFT) is a \ul{greedy} list scheduling heuristic as we have discussed before. It provides the lowest slack compared to all the algorithms thereby making it less robust, but more efficient in terms of schedule length. However, our algorithm provides a slightly higher slack than HEFT, while still providing a much lower SLR value compared to HEFT. This leads us to believe that there is more room for optimization (as higher slack usually translates to larger windows of time in which a tasks start time can be moved) with the schedules generated by CEFT-CPOP which can be utilised to further decrease its makespan.}\fi. The slack produced by CPOP and CEFT-CPOP are similar (our algorithm produces slacks that are $\sim 1\%$ -- $\sim 2\%$ larger).


\begin{figure*}[ht]
	\centering
	\hspace*{\fill}%
	\csubfloat[FFT-medium]{
		\includegraphics[width=0.22\linewidth]{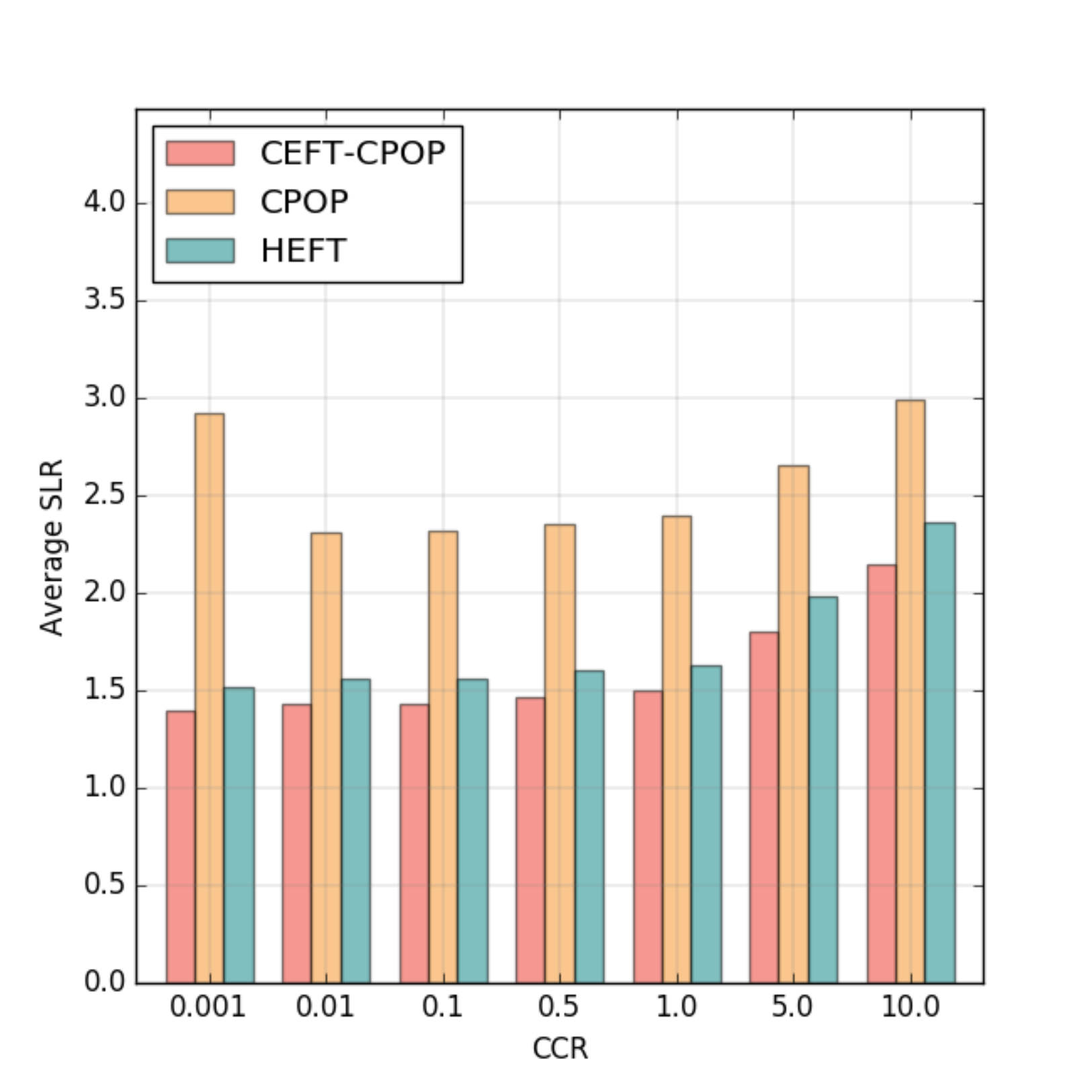}
		\label{fig:fft-med-ccr-slr}
	}
	\centerhfill
	\csubfloat[GE-medium]{
		\includegraphics[width=0.22\linewidth]{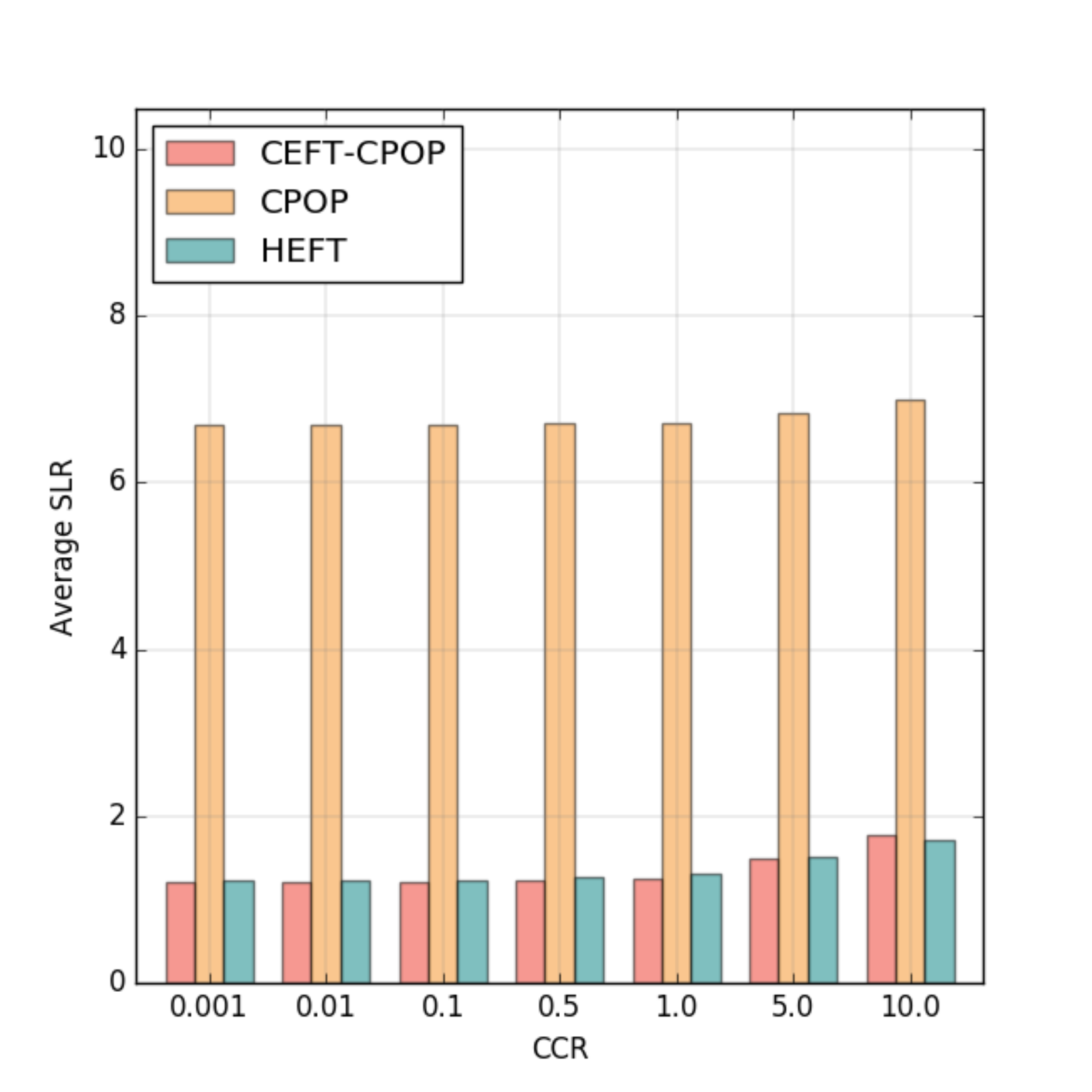}
		\label{fig:ge-med-ccr-slr}
	}
	\centerhfill
	\csubfloat[MD-medium]{
		\includegraphics[width=0.22\linewidth]{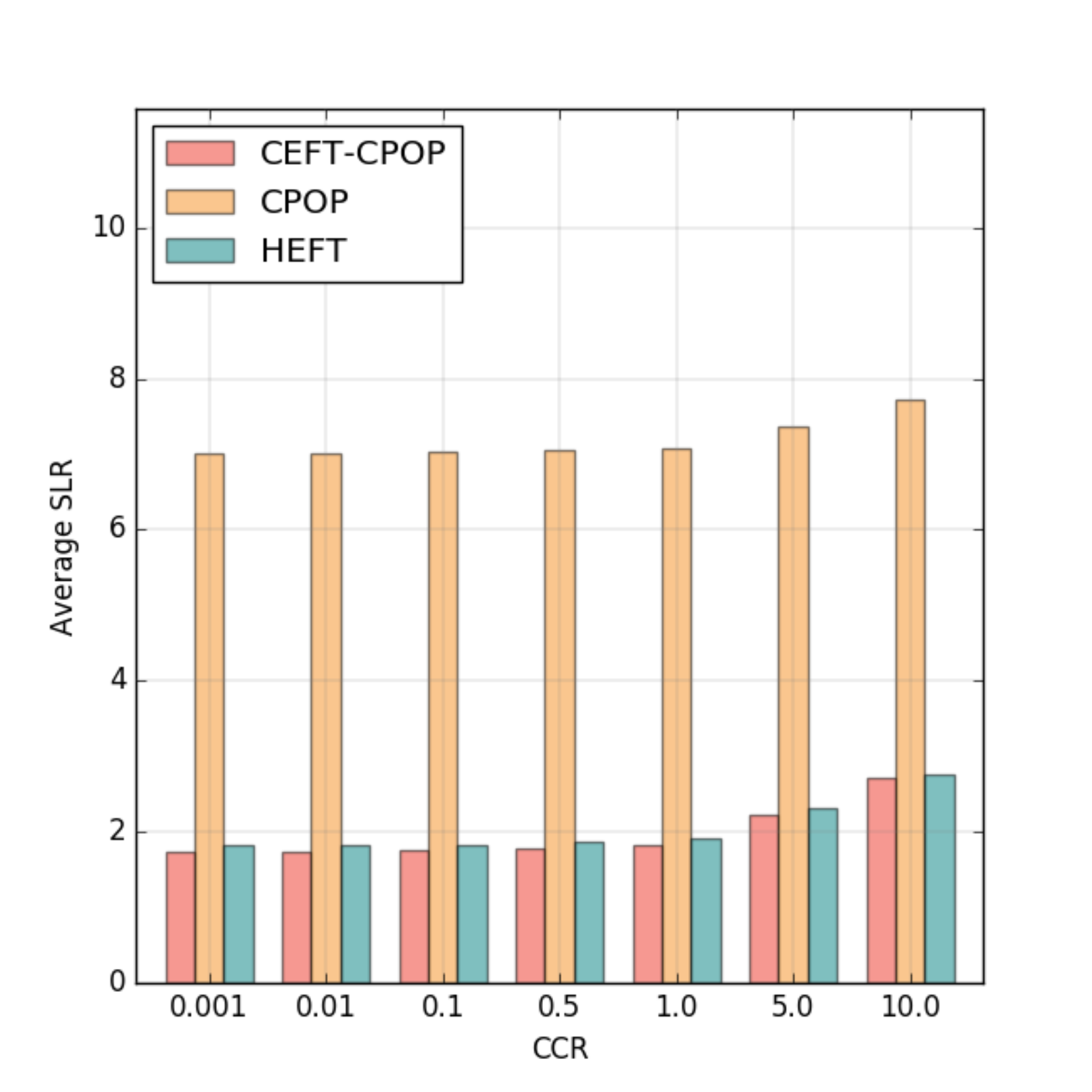}
		\label{fig:md-med-ccr-slr}
	}
	\centerhfill
	\csubfloat[EW-medium]{
		\includegraphics[width=0.22\linewidth]{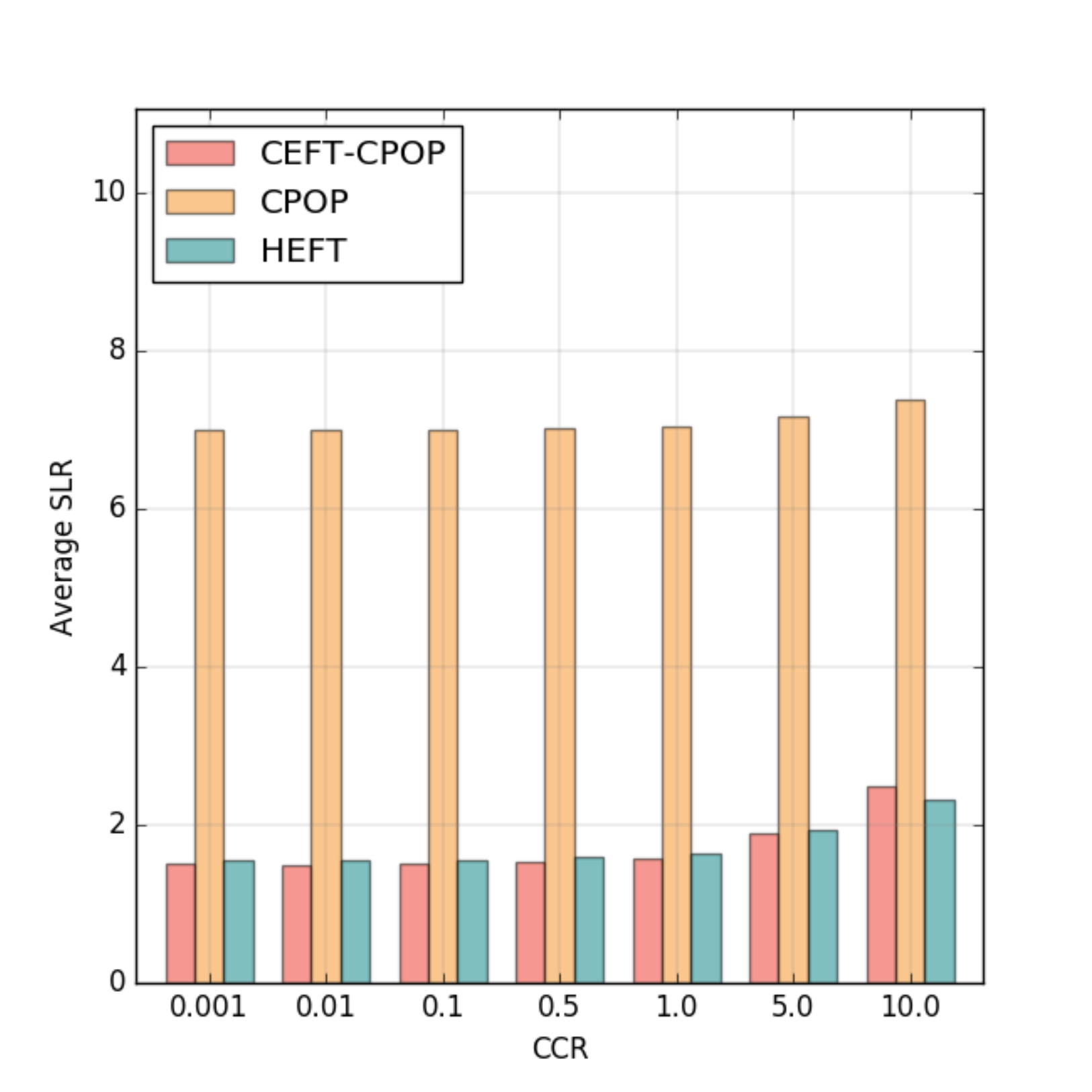}
		\label{fig:ew-med-ccr-slr}
	}
	\hspace*{\fill}%
	\caption{Comparing SLR across the different real-world benchmarks (medium variants) in terms of CCR of the input graphs. Lower is better.}
	\label{fig:realworld-med-ccr-slr}
\end{figure*}

\begin{figure*}[t]
	\centering
	\hspace*{\fill}%
	\csubfloat[FFT-classic]{
		\includegraphics[width=0.22\linewidth]{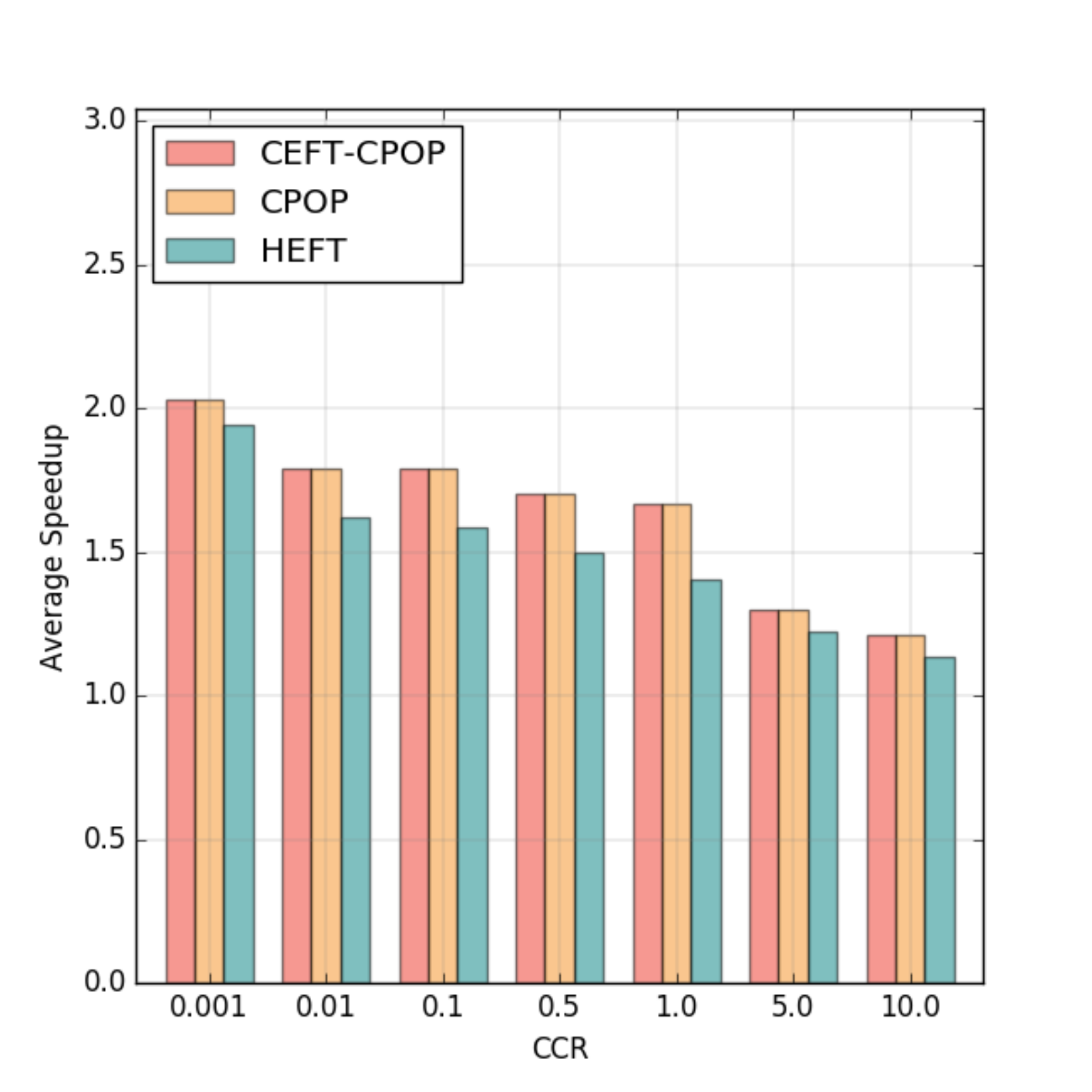}
		\label{fig:fft-clas-ccr-speedup}
	}
	\centerhfill
	\csubfloat[GE-classic]{
		\includegraphics[width=0.22\linewidth]{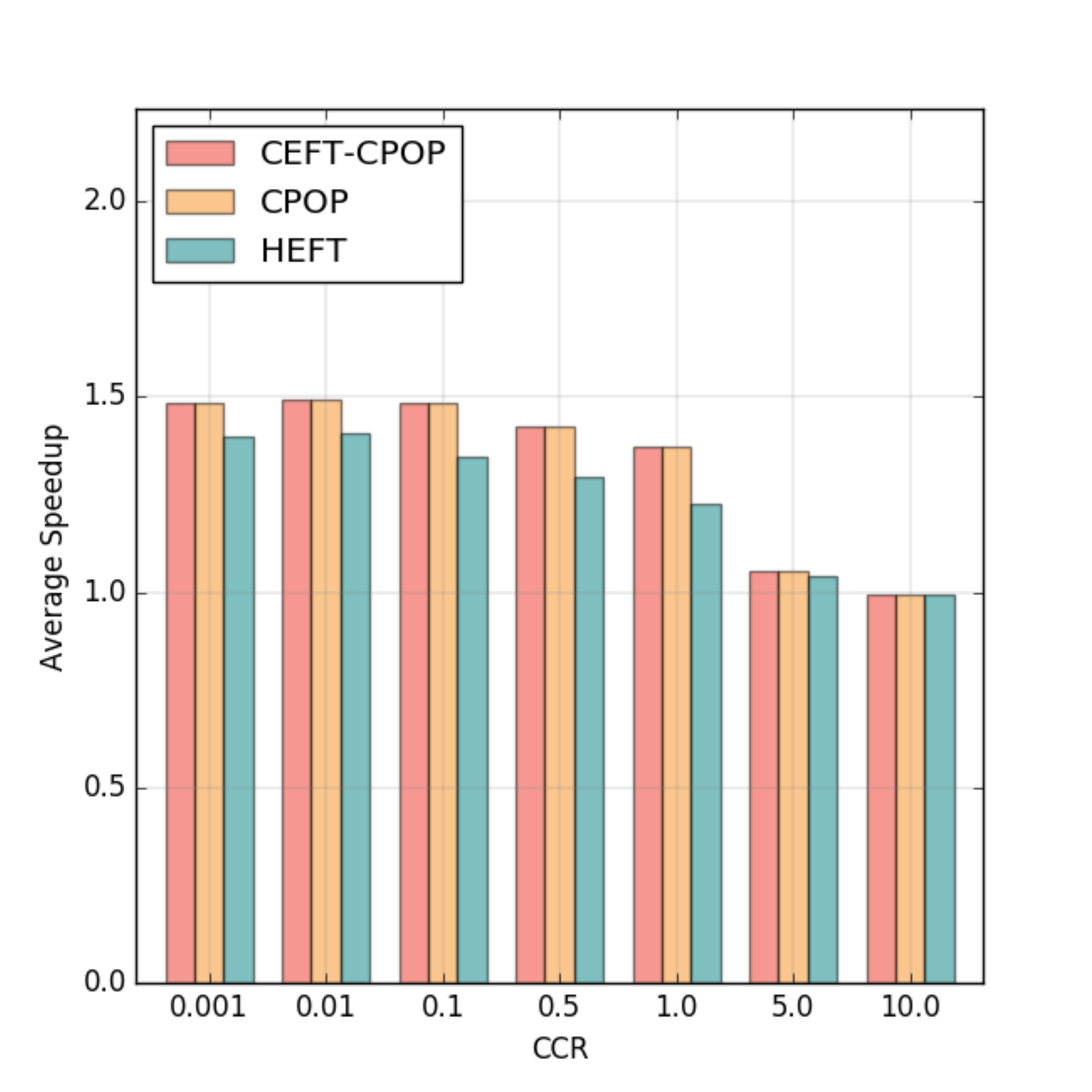}
		\label{fig:ge-clas-ccr-speedup}
	}
	\centerhfill
	\csubfloat[MD-classic]{
		\includegraphics[width=0.22\linewidth]{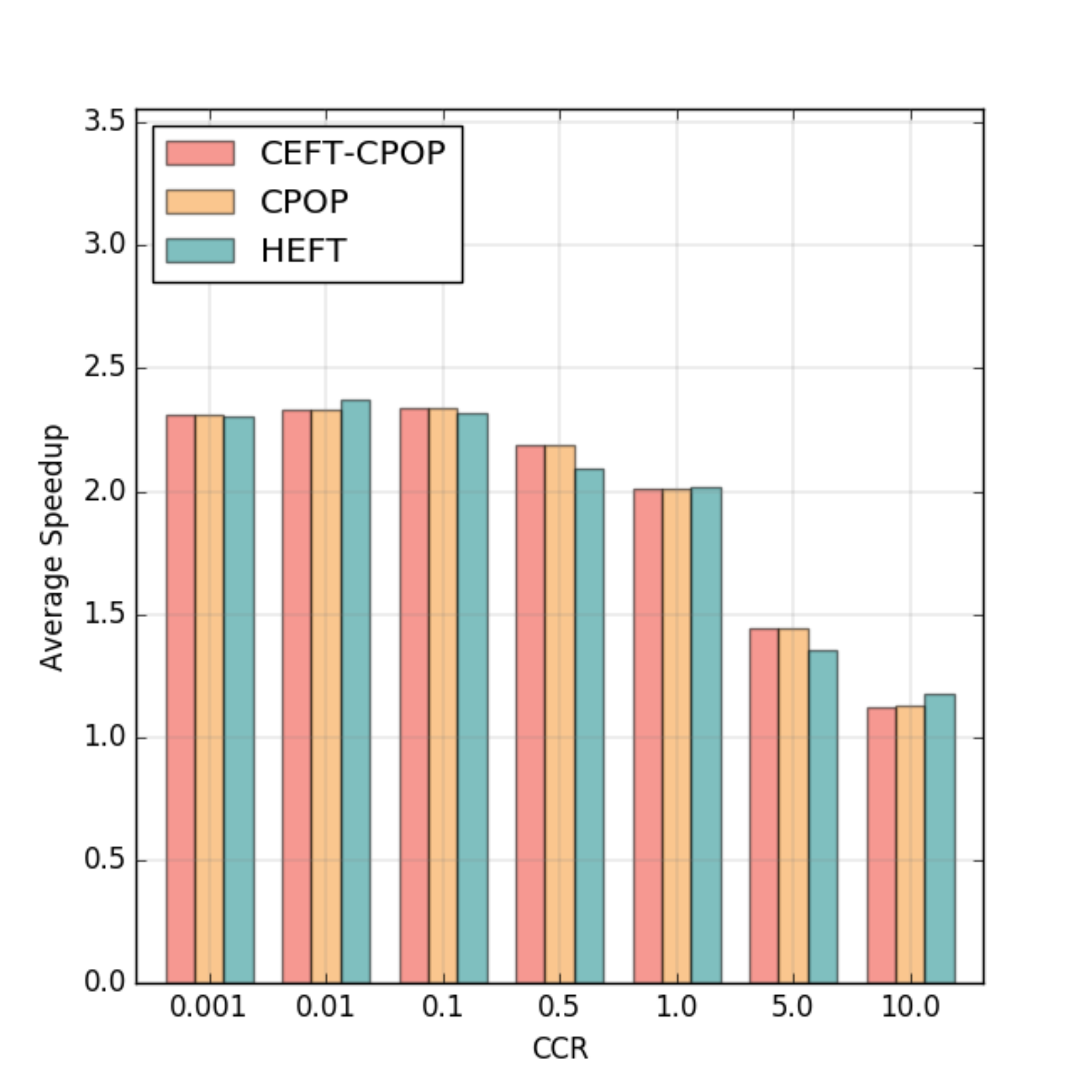}
		\label{fig:md-clas-ccr-speedup}
	}
	\centerhfill
	\csubfloat[EW-classic]{
		\includegraphics[width=0.22\linewidth]{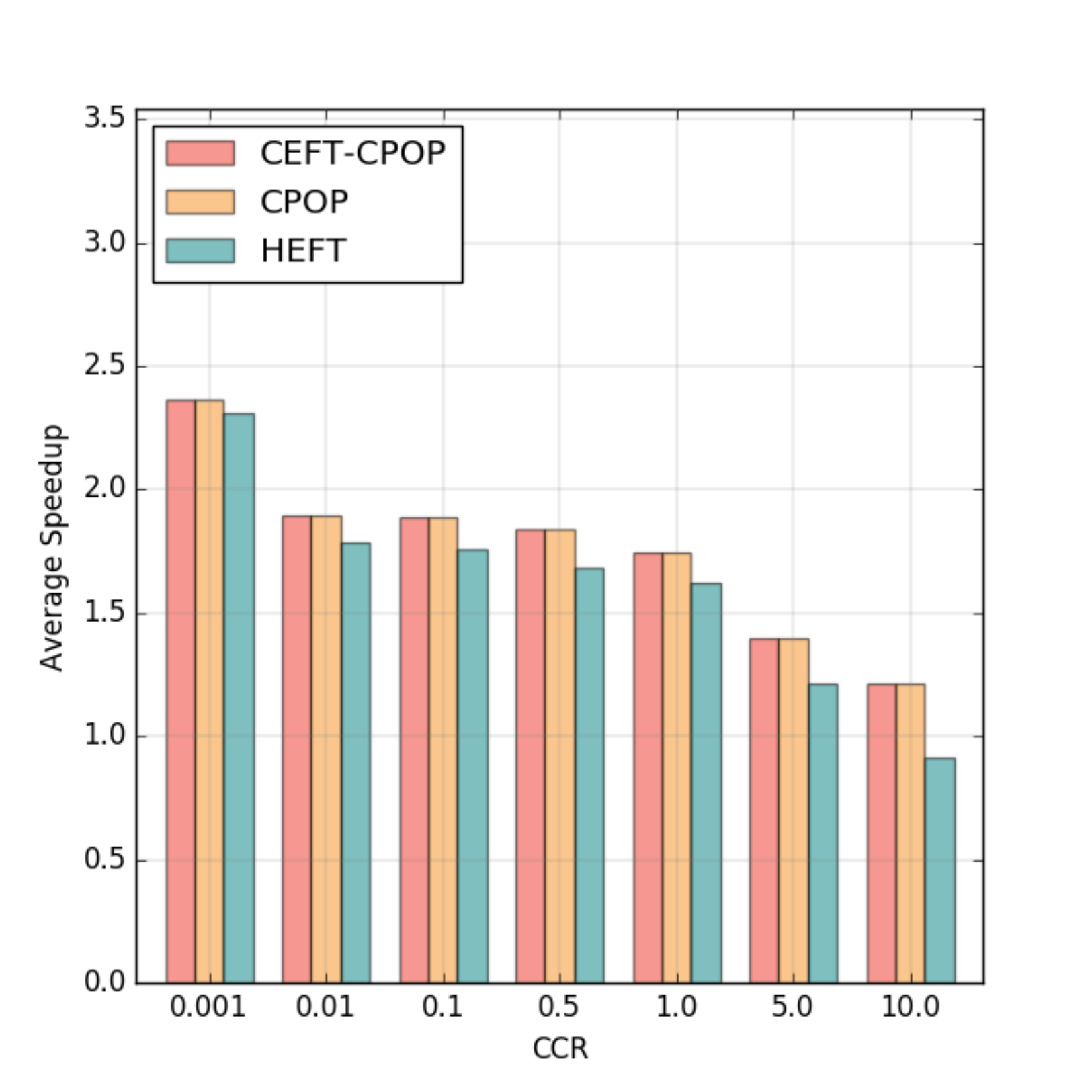}
		\label{fig:ew-clas-ccr-speedup}
	}
	\hspace*{\fill}%
	\caption{Comparing speedup across the different real-world benchmarks (classic variants) in terms of CCR of the input graphs. Higher is better.}
	\label{fig:realworld-clas-ccr-speedup}
\end{figure*}

In figure~\ref{fig:task-res-slr}, we present how the SLR varies with increasing number of tasks in the application DAG and increasing number of resources in the resource graph. Arabnejad et al. suggest that the decrease in performance of the scheduling algorithms with the increase in number of tasks is due to a marked increase in the number of concurrent tasks. According to them, algorithms that have lookahead features tend to suffer more as these algorithms tend to base the decision of scheduling the current task heavily on its children tasks. For some DAGs that have many concurrent tasks to schedule, the processor load is substantially changed by the concurrent tasks to be scheduled after the current task. Therefore, the conditions at the time of scheduling the current task are different than the conditions at the time of scheduling its child tasks. This implies that the decision made by the algorithm to to schedule the parent task might not be valid, hence leading to poorer solutions. Our algorithm, in spite of incorporating lookahead features, provides the lowest SLR of all the three algorithms compared for smaller number of tasks ($n=128$ to $n=1024$). For larger graphs, HEFT manages to produce better makespans and hence better SLR values, but our algorithm continues to outperform CPOP.

\ifdefined\longversion
\else
For more results and analysis we encourage the readers read a longer version of this paper on arxiv~\ref{ceft-arxiv}.
\fi

\ifdefined\longversion
\begin{figure*}[ht]
	\centering
	\hspace*{\fill}%
	\csubfloat[FFT-classic]{
		\includegraphics[width=0.22\linewidth]{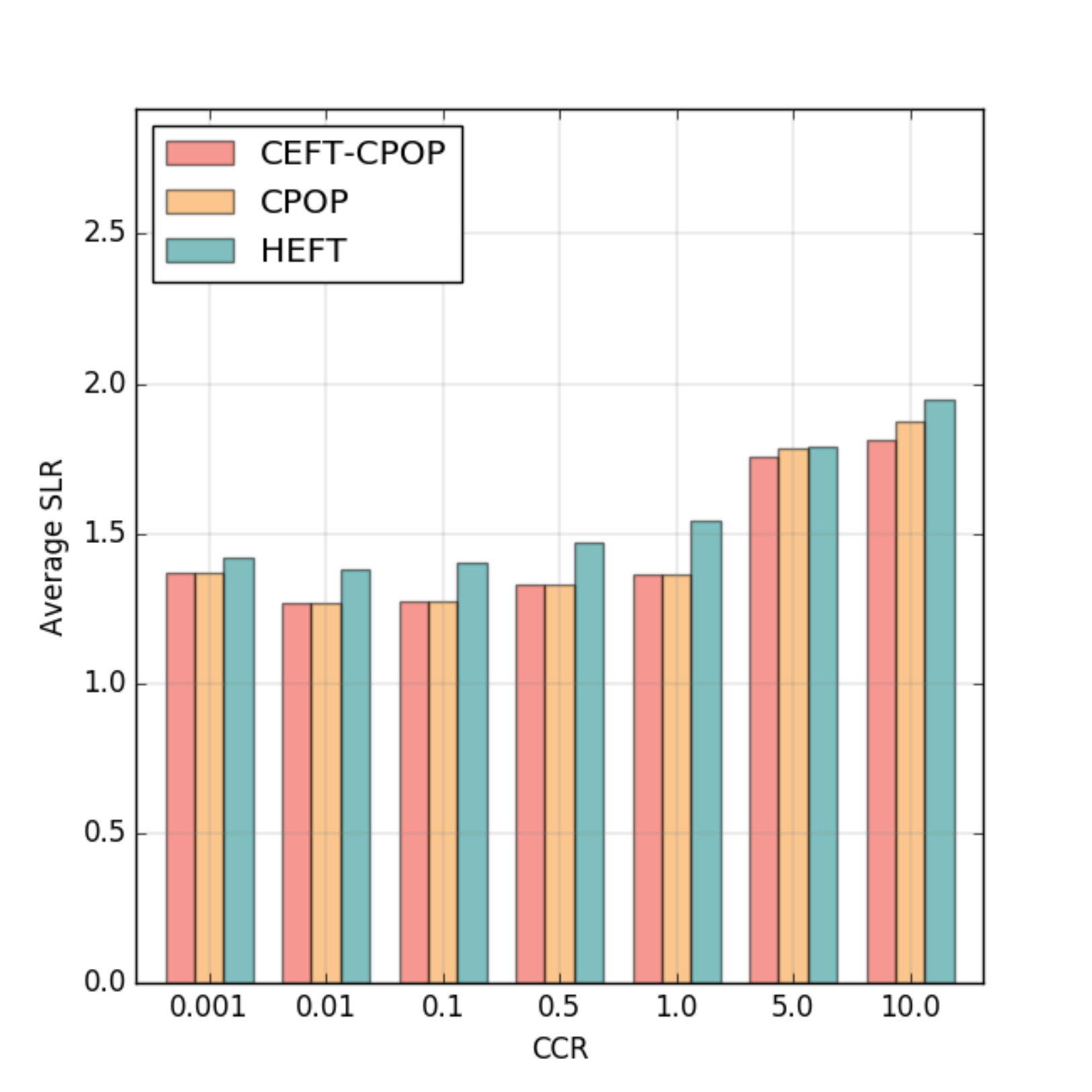}
		\label{fig:fft-clas-ccr-slr}
	}
	\centerhfill
	\csubfloat[GE-classic]{
		\includegraphics[width=0.22\linewidth]{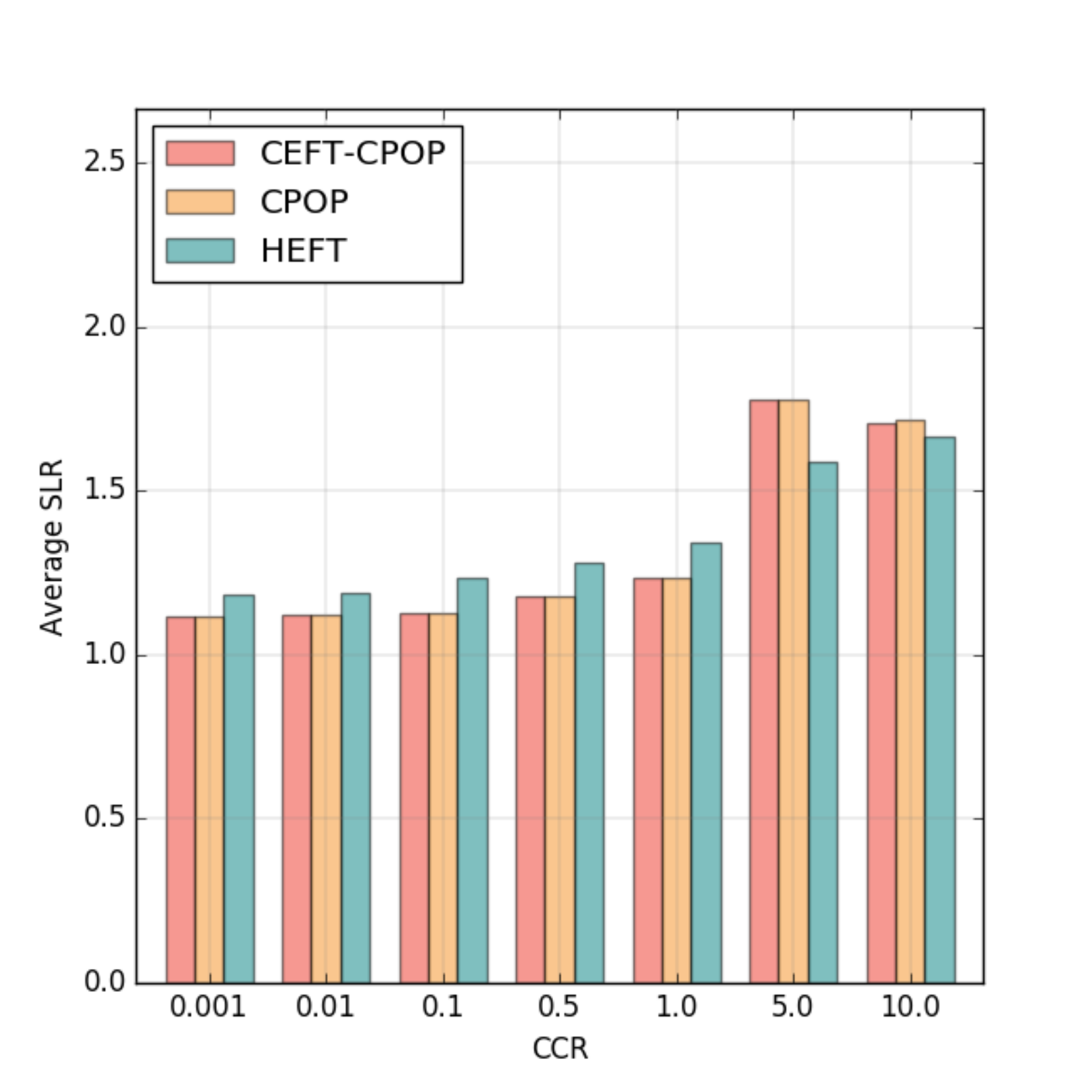}
		\label{fig:ge-clas-ccr-slr}
	}
	\centerhfill
	\csubfloat[MD-classic]{
		\includegraphics[width=0.22\linewidth]{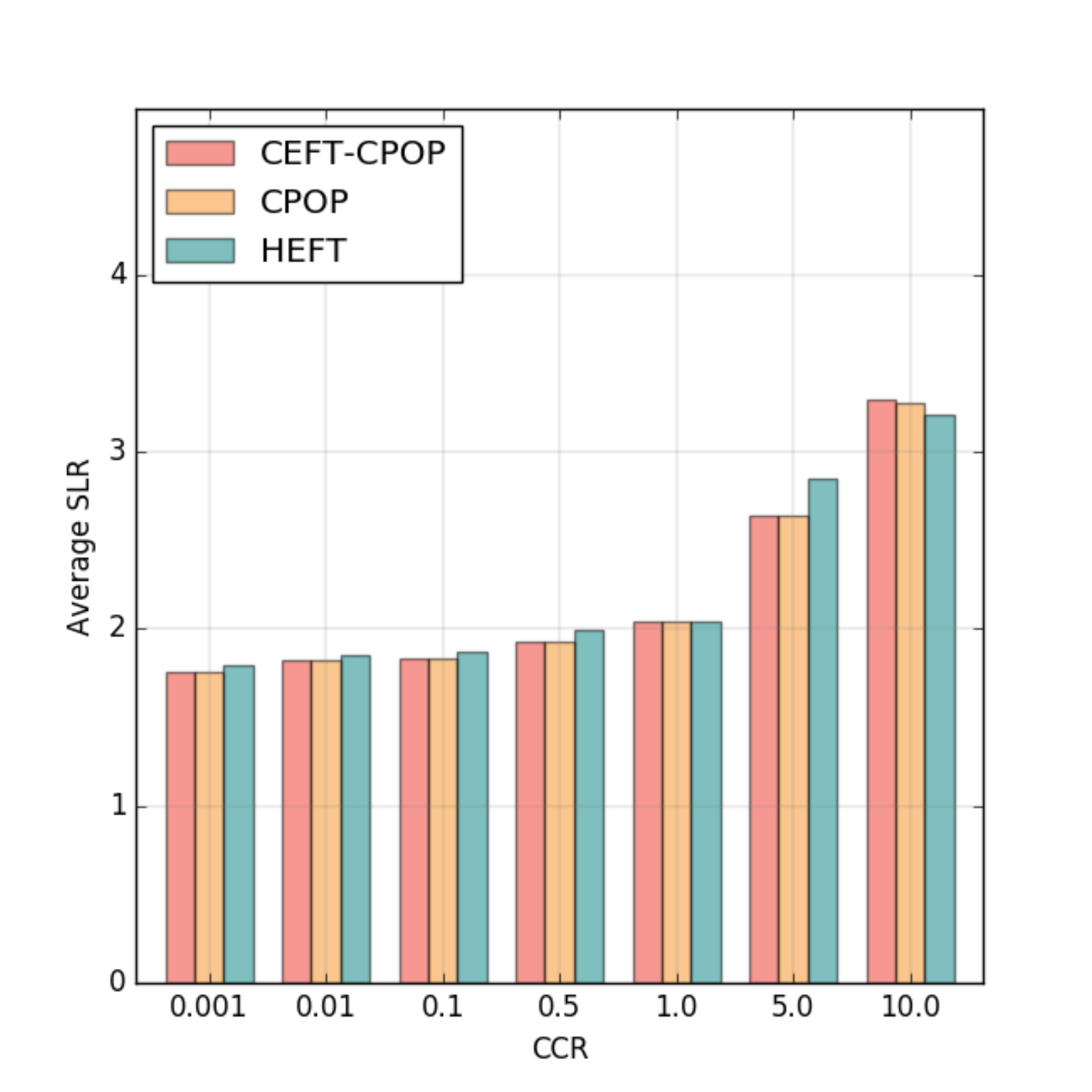}
		\label{fig:md-clas-ccr-slr}
	}
	\centerhfill
	\csubfloat[EW-classic]{
		\includegraphics[width=0.22\linewidth]{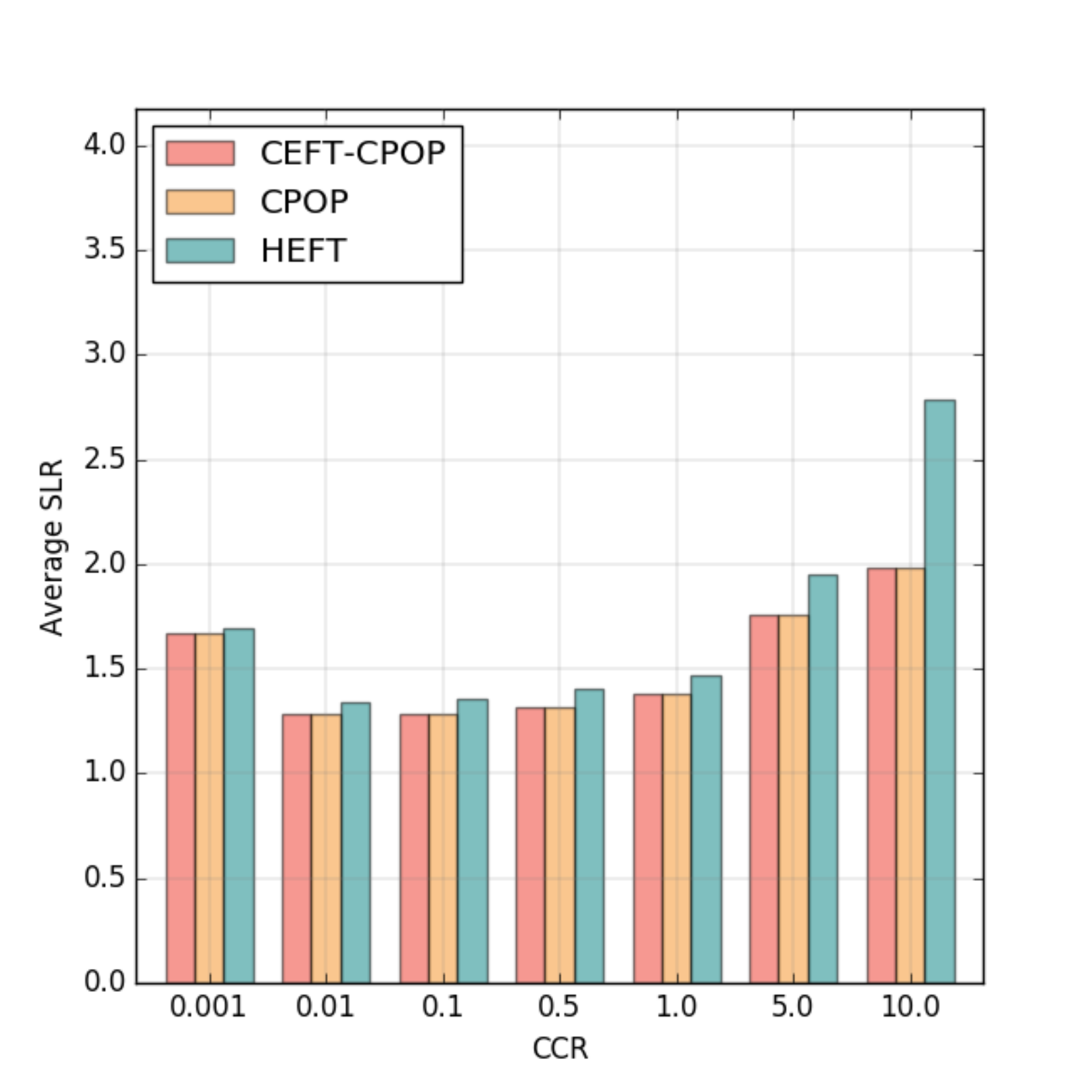}
		\label{fig:ew-clas-ccr-slr}
	}
	\hspace*{\fill}%
	\caption{Comparing SLR across the different real-world benchmarks (classic variants) in terms of CCR of the input graphs. Lower is better.}
	\label{fig:realworld-clas-ccr-slr}
\end{figure*}
\fi

\subsection{Real World Benchmarks}
\label{sec:real-worl-benc}

\ifdefined\longversion

\begin{figure*}[ht]
	\centering
	\hspace*{\fill}%
	\csubfloat[FFT-medium]{
		\includegraphics[width=0.22\linewidth]{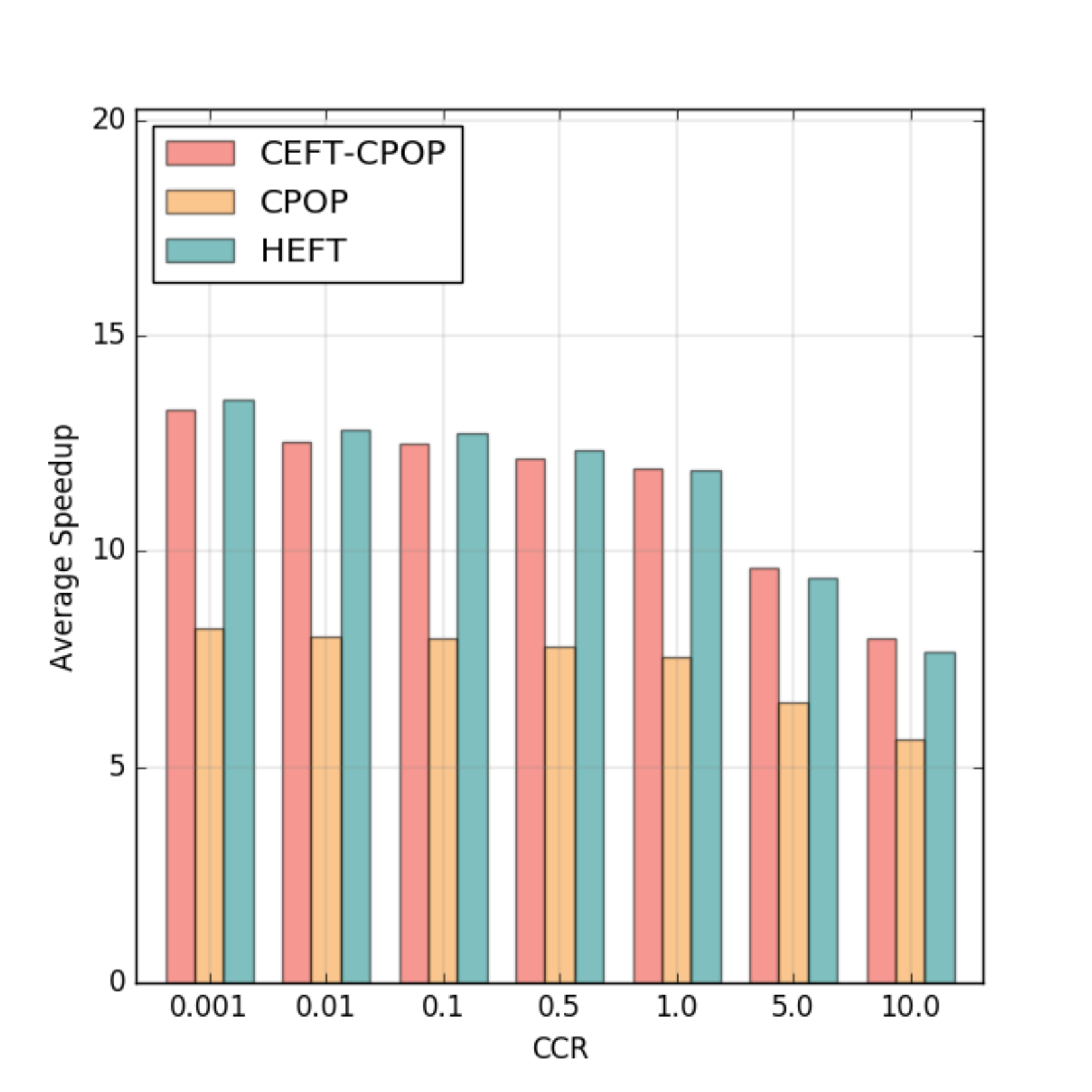}
		\label{fig:fft-med-ccr-speedup}
	}
	\centerhfill
	\csubfloat[GE-medium]{
		\includegraphics[width=0.22\linewidth]{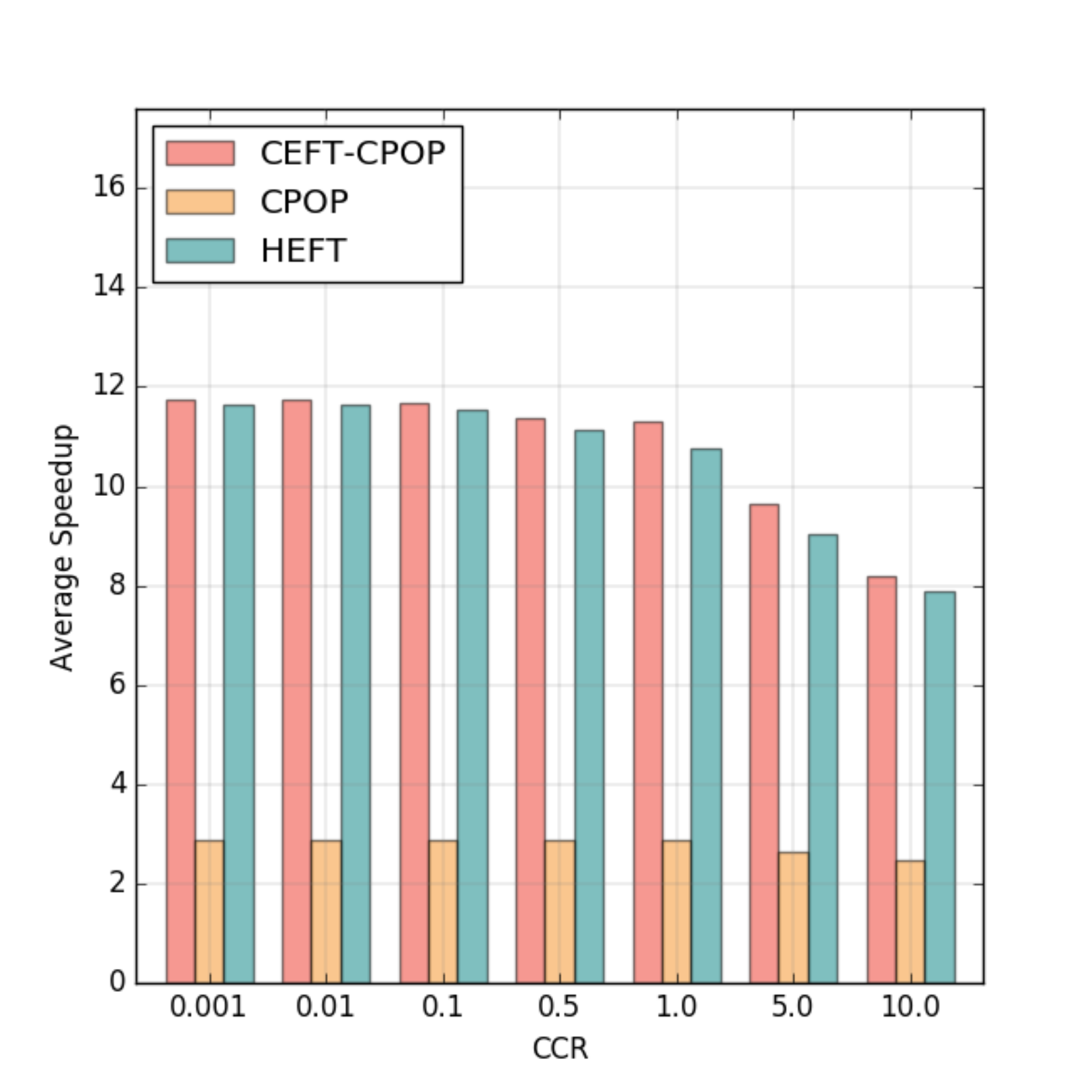}
		\label{fig:ge-med-ccr-speedup}
	}
	\centerhfill
	\csubfloat[MD-medium]{
		\includegraphics[width=0.22\linewidth]{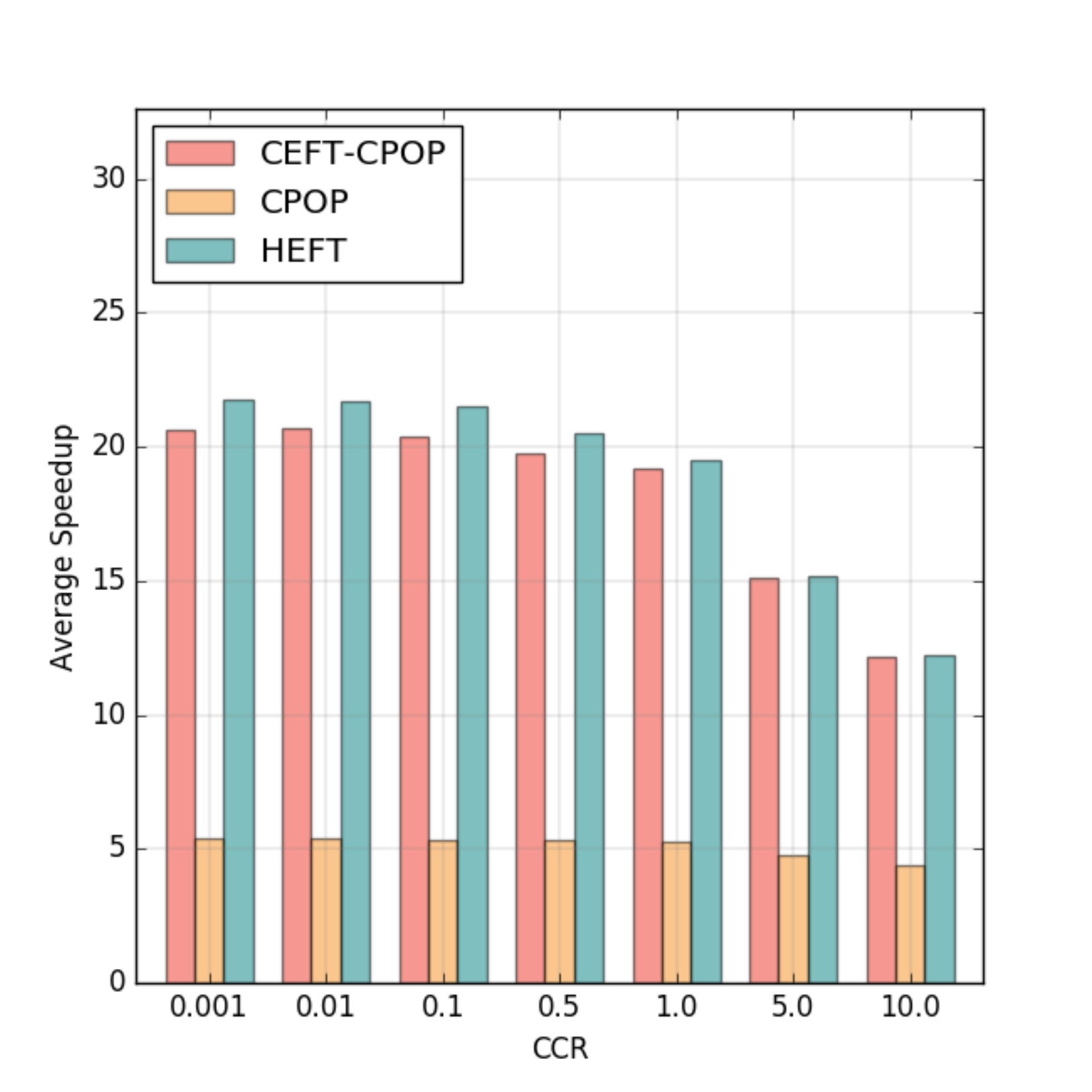}
		\label{fig:md-med-ccr-speedup}
	}
	\centerhfill
	\csubfloat[EW-medium]{
		\includegraphics[width=0.22\linewidth]{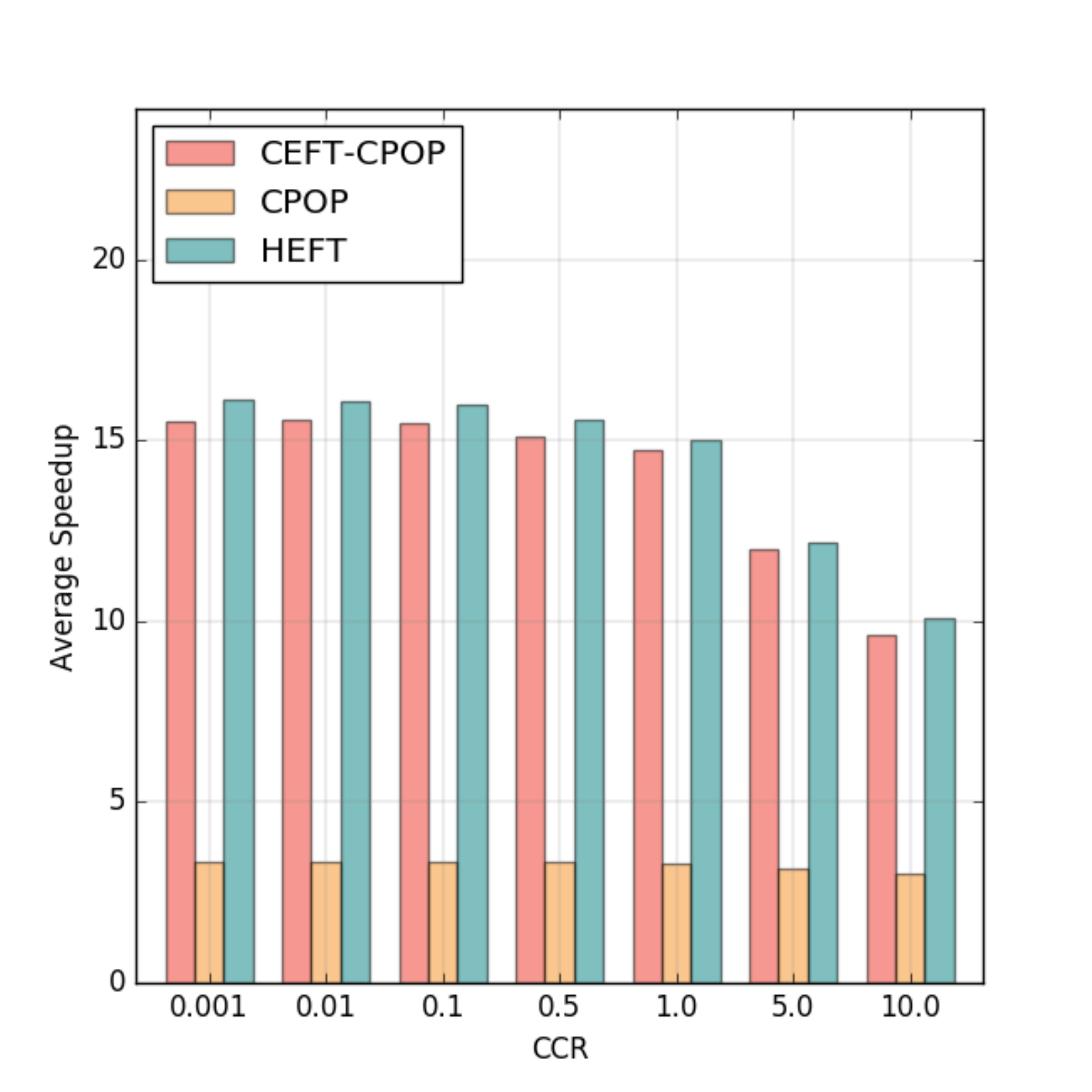}
		\label{fig:ew-med-ccr-speedup}
	}
	\hspace*{\fill}%
	\caption{Comparing speedup across the different real-world benchmarks (medium variants) in terms of CCR of the input graphs. Higher is better.}
	\label{fig:realworld-med-ccr-speedup}
\end{figure*}

\fi

\ifdefined\longversion Figures~\ref{fig:realworld-clas-ccr-slr},~\ref{fig:realworld-med-ccr-speedup} and~\ref{fig:realworld-med-ccr-slr} show the performance of the three algorithms in terms of the schedule length ratio (SLR) and speedup \else Figure~\ref{fig:realworld-med-ccr-slr} shows the performance of the three algorithms in terms of the schedule length ratio (SLR) \fi for the four real-world benchmarks from section~\ref{sec:paper-crit-path-real-worl-grap}: Fast Fourier Transform (FFT), Gaussian Elimination (GE), Molecular Dynamics code (MD) and the Epigenomics Workflow (EW). The values in the SLR graphs are the primary motivating factors for exploring the effectiveness of our algorithm using randomly generated graphs. Scheduling length ratio (SLR) is the ratio of the makespan to the length of the critical path (ignoring communication costs) when it is mapped onto the fastest processor. SLR is hence used as a metric to measure size of the application compared to the critical path. Lower values of SLR imply that the makespan is comparable to the length of the critical path and hence scheduling algorithms with lower values of SLR are preferred.

Intuitively, applications that exhibit higher SLR (applications whose optimal makespan is much larger compared to the length of the critical path) is useful for testing the effectiveness of scheduling algorithms as they require the scheduling algorithm to make the right decision on a larger percentage of the total number of tasks in the application. \ifdefined\longversion Figure~\ref{fig:realworld-clas-ccr-slr} shows that the average SLR values of the different algorithms on the real-world benchmarks is much lower than the SLR values of the algorithms on the modified versions (medium variants; generated in a similar fashion to RGG-medium using the structure of the real world graphs as discussed in sections~\ref{sec:paper-crit-path-real-worl-grap} and~\ref{sec:rand-gene-grap}). The graphs for the medium variants of the real-world graphs are shown in the paper as they are a token representative of the three randomly generated datasets (low, medium and high).\fi

From figure~\ref{fig:realworld-med-ccr-slr}, we can see that the performance (SLR) of all the algorithms on the real-world benchmarks suffer as the CCR increases. As explained earlier, the SLR is the ratio between the makespan and the length of the critical path. As communication costs increase well beyond computation costs ($CCR >>1.0$), the tasks from the critical paths are mapped onto the same processor in attempt to minimize the critical path length. The makespans however increase as a general consequence of increased communication costs. This explains why the average SLR values go up as communication costs increase. 

\ifdefined\longversion In the classic versions of the real-world benchmarks (figures~\ref{fig:realworld-clas-ccr-slr}) CEFT-CPOP produces the lowest average SLR values. Across all the four real-world benchmarks, CEFT produces critical paths of the same length as CPOP in $\sim 97.28\%$ of the cases. This is reflected in the SLR values in figure~\ref{fig:realworld-clas-ccr-slr}. Owing to the small number of tasks, HEFT outperforms CEFT by a small fraction in all the real-world benchmarks except in GE (this is further supported by the trend from figure~\ref{fig:rgg-clas-task-slr} where HEFT produces lower SLRs as the number of tasks increases). However, when the heterogeneity is generated using our method (as discussed in section~\ref{sec:rand-gene-grap}; medium variant of the real-world benchmarks), CEFT produces critical paths that are shorter than CPOP's in $\sim 73.8\%$ which lead to better makespans in $\sim 77.77\%$ of the cases.\fi

\subsection{HEFT Ranking Function with CEFT}
\label{sec:heft-rank-func-with-ceft}

\ifdefined\longversion

\begin{figure*}[ht]
	\centering
	\hspace*{\fill}%
	\csubfloat[RGG-classic]{
		\includegraphics[width=0.22\linewidth]{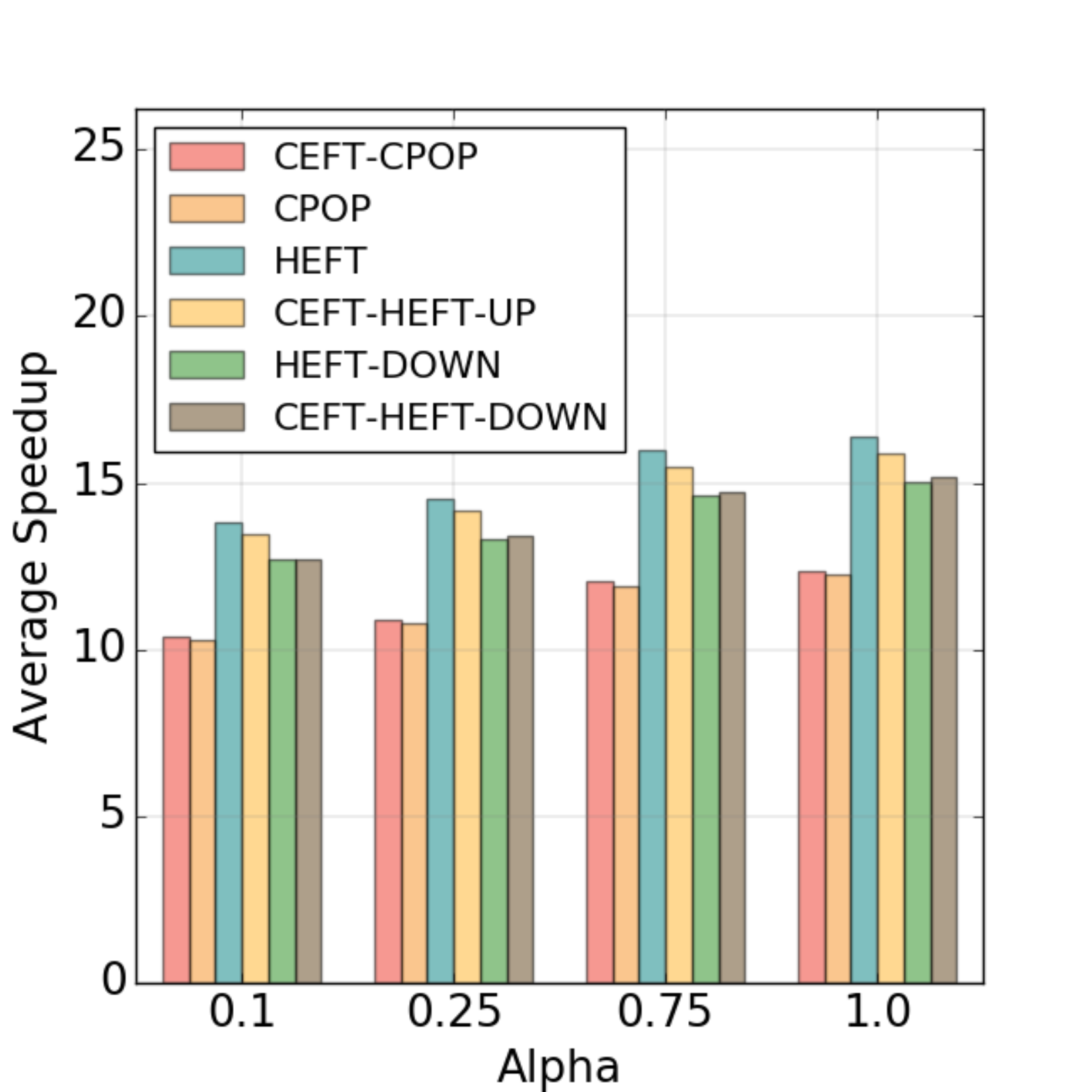}
		\label{fig:all-rgg-clas-alpha-speedup}
	}
	\centerhfill
	\csubfloat[RGG-low]{
		\includegraphics[width=0.22\linewidth]{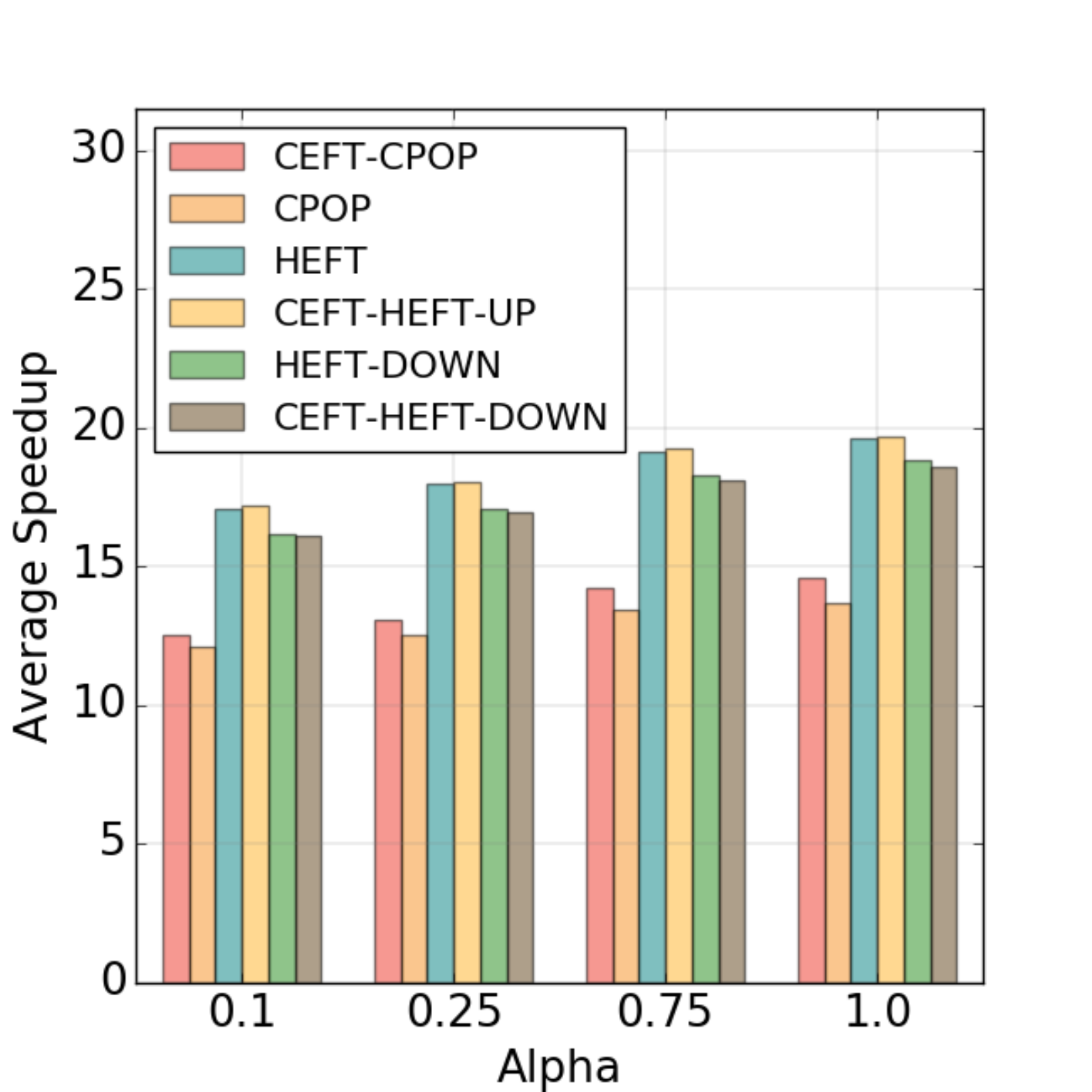}
		\label{fig:all-rgg-low-alpha-speedup}
	}
	\centerhfill
	\csubfloat[RGG-medium]{
		\includegraphics[width=0.22\linewidth]{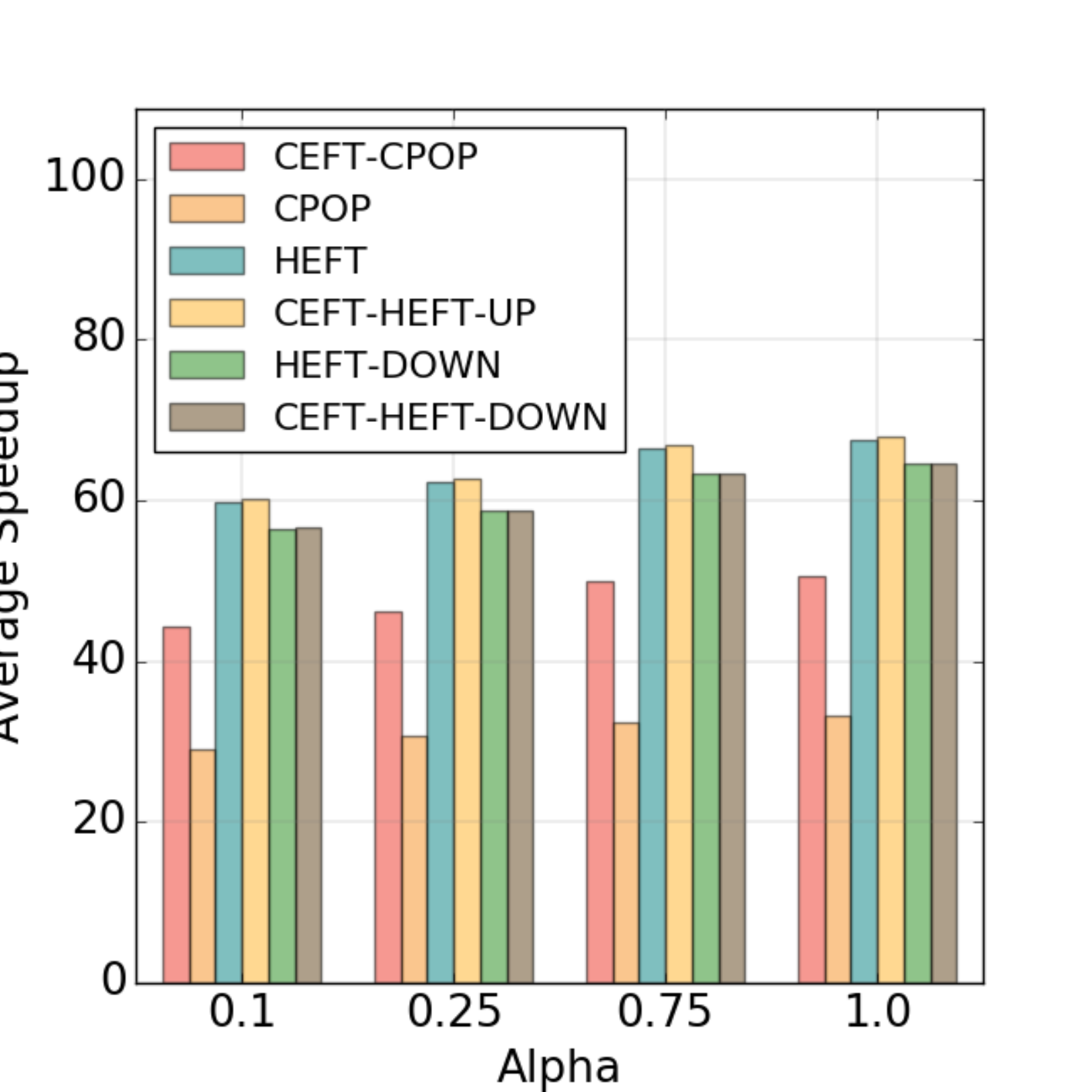}
		\label{fig:all-rgg-med-alpha-speedup}
	}
	\centerhfill
	\csubfloat[RGG-high]{
		\includegraphics[width=0.22\linewidth]{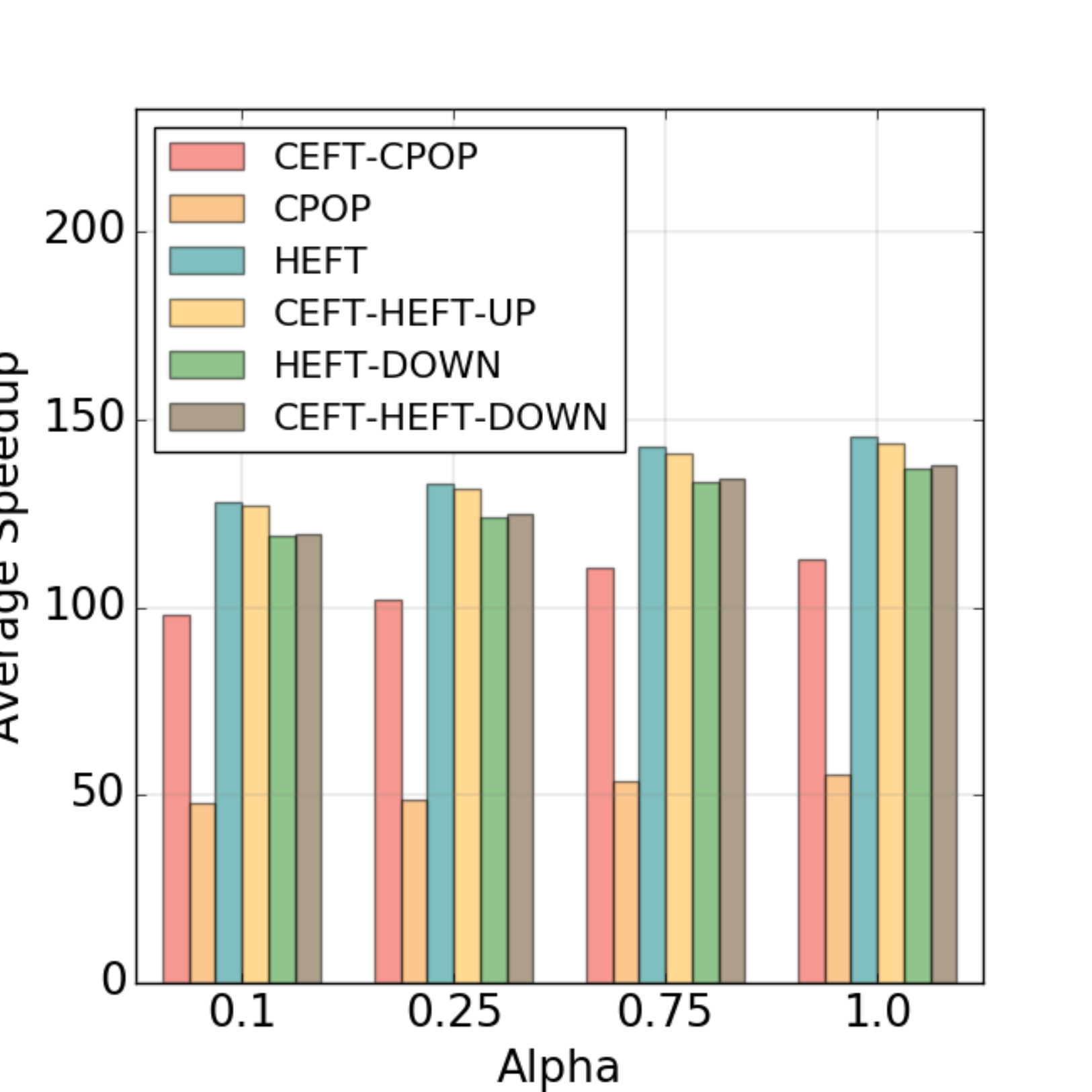}
		\label{fig:all-rgg-high-alpha-speedup}
	}
	\hspace*{\fill}%
	\caption{Comparing speedup across different workloads in terms of $\alpha$ of the input graphs. Higher is better.}
	\label{fig:all-alpha-speedup}
\end{figure*}

\begin{figure*}[ht]
	\centering
	\hspace*{\fill}%
	\csubfloat[RGG-classic]{
		\includegraphics[width=0.22\linewidth]{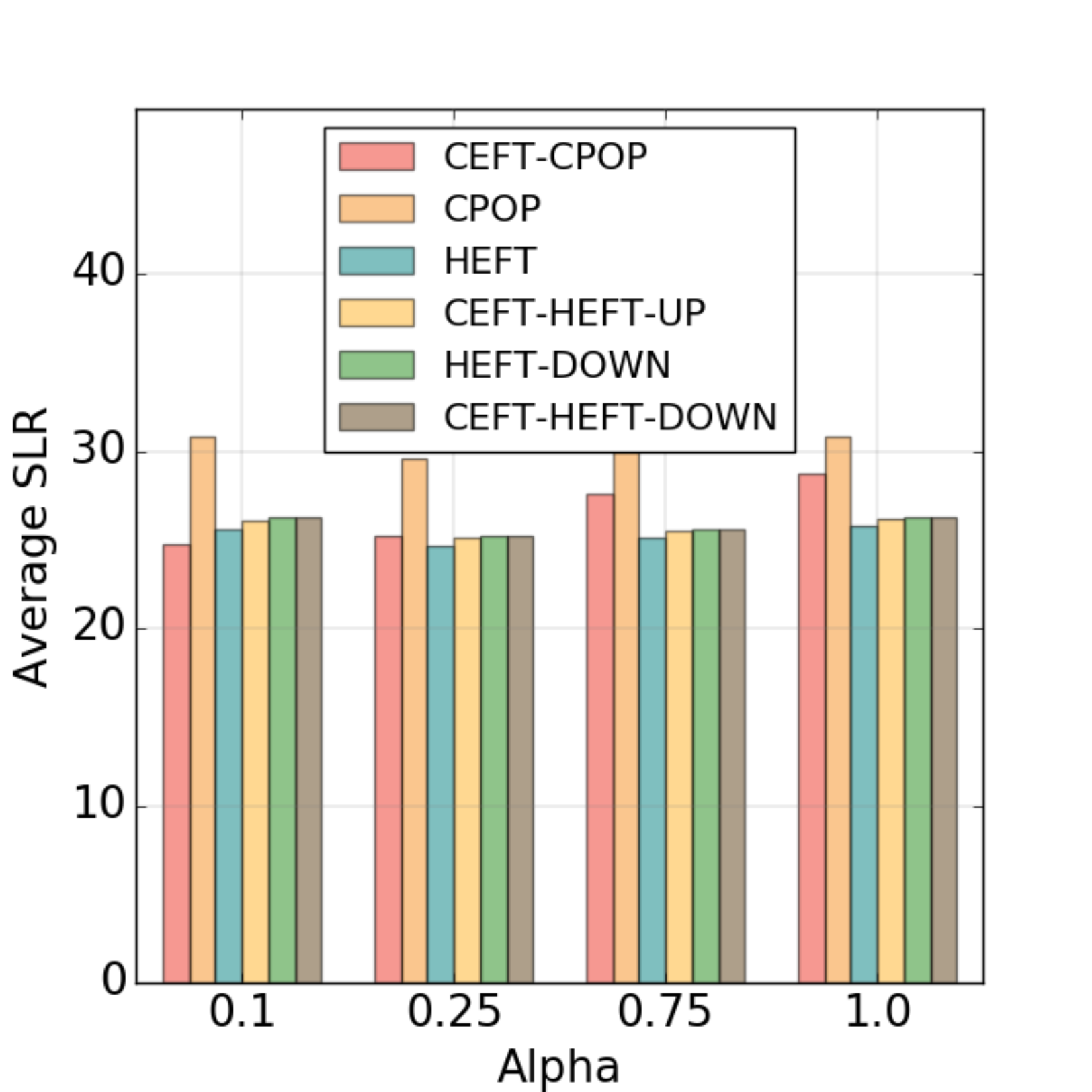}
		\label{fig:all-rgg-clas-alpha-slr}
	}
	\centerhfill
	\csubfloat[RGG-low]{
		\includegraphics[width=0.22\linewidth]{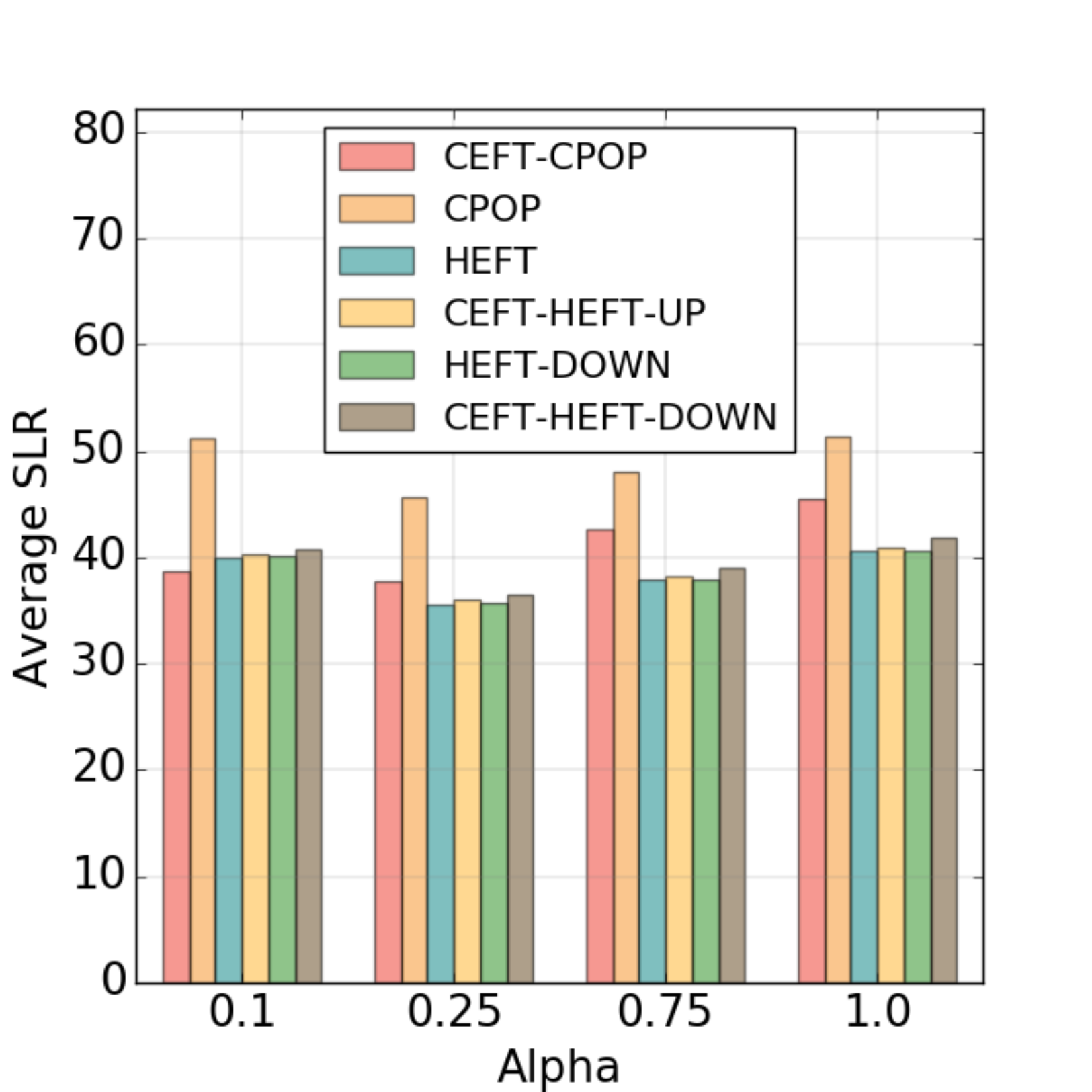}
		\label{fig:all-rgg-low-alpha-slr}
	}
	\centerhfill
	\csubfloat[RGG-medium]{
		\includegraphics[width=0.22\linewidth]{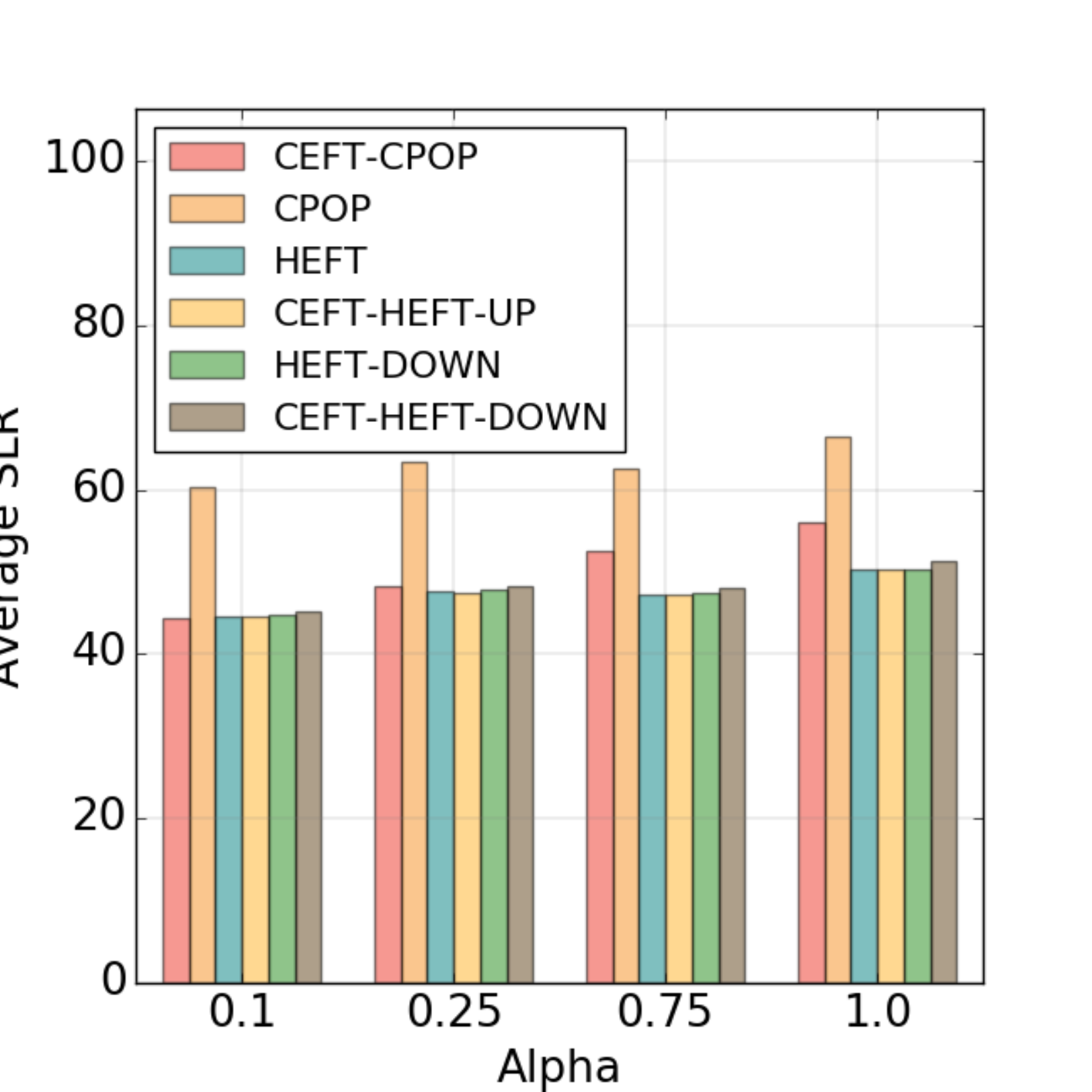}
		\label{fig:all-rgg-med-alpha-slr}
	}
	\centerhfill
	\csubfloat[RGG-high]{
		\includegraphics[width=0.22\linewidth]{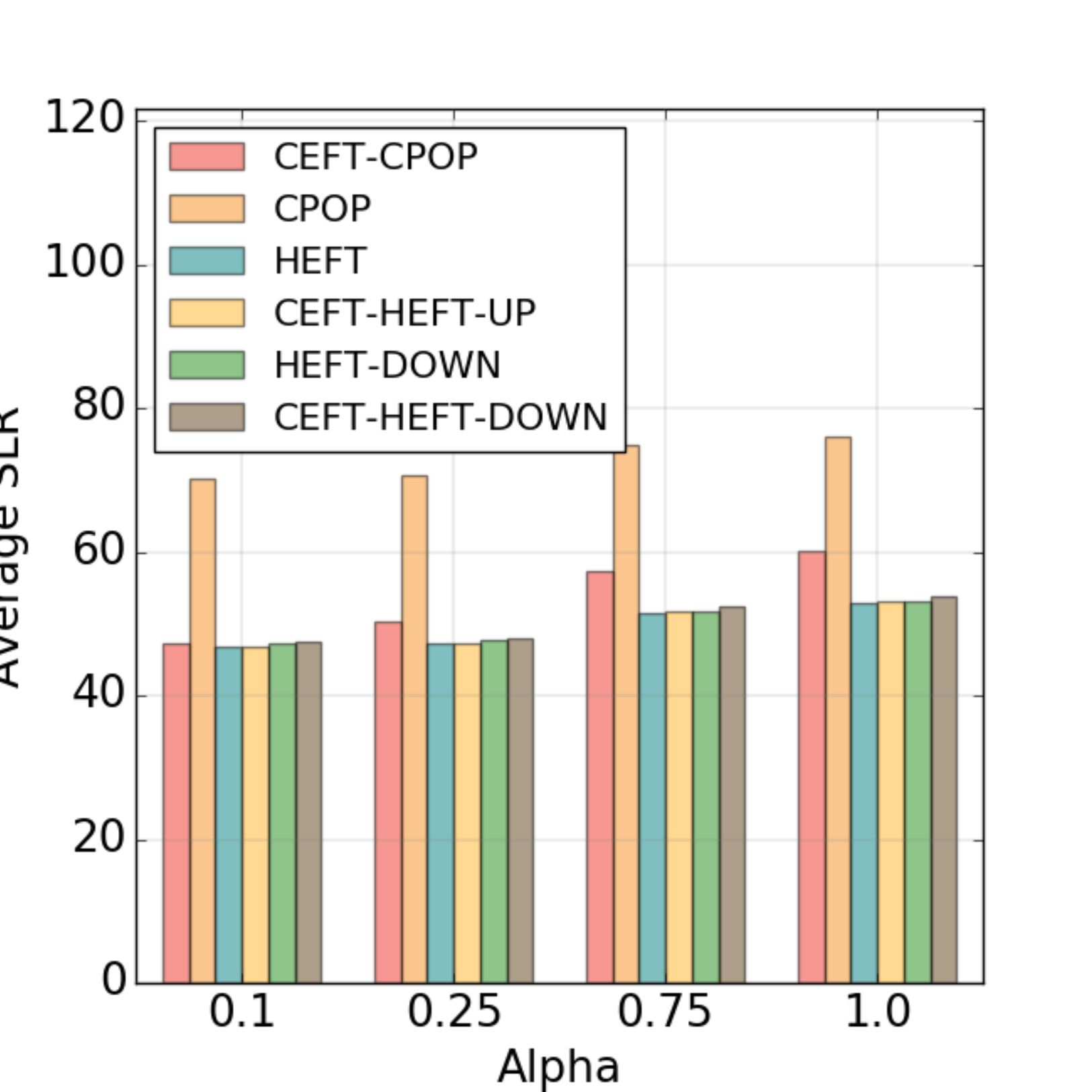}
		\label{fig:all-rgg-high-alpha-slr}
	}
	\hspace*{\fill}%
	\caption{Comparing slr across different workloads in terms of Alpha of the input graphs. Lower is better.}
	\label{fig:all-alpha-slr}
\end{figure*}

\fi


Topcuoglu et al.~\cite{topcuoglu2002performance} calculate two heuristics for assigning priorities to tasks: downward rank ($rank_d$) and upward rank ($rank_u$). The downward rank of a task $t_i$ is the length of the longest path from the entry task to $t_i$ in the DAG. Consequently, the upward rank ($rank_u$) of task $t_i$ is the length of the longest path from $t_i$ to an exit node.

Both these ranks are calculated using average execution times and average communication times. In order to calculate the priorities more accurately we propose two new ranking schemes called $rank_{ceft-down}$ and $rank_{ceft-up}$. For the former, we use the CEFT dynamic programming array that has been calculated by traversing the application graph in a topological order. For every task, we choose the processor that minimizes the CEFT value and use that as its downward rank. \ifdefined\longversion As explained in section~\ref{sec:paper-crit-path-crit-path-pape-dyna-prog-solu}, CEFT calculates the length of the critical path from the source task to a given task using accurate execution times, which serves as the primary motivation for modifying $rank_d$ in this manner.\fi In order to calculate the upward rank using CEFT ($rank_{ceft-up}$), we transpose the application graph (invert the edges, keeping the vertices same) and run the CEFT algorithm on this newly transposed graph. We then employ a similar strategy as before, assigning every task to the processor that minimizes its CEFT value and use that value as its upward rank.

In all the graphs presented in this section, the bars labelled HEFT refer to the default HEFT algorithm using the upward rank $rank_u$, while CEFT-HEFT-UP refers to the HEFT algorithm using $rank_{ceft-up}$. HEFT-DOWN refers to HEFT using the downward rank $rank_d$ while CEFT-HEFT-DOWN refers to HEFT using $rank_{ceft-down}$.

Figure~\ref{fig:all-alpha-speedup} compares the average \textit{speedup} obtained by these algorithms with the three algorithms from the previous section. It is evident from the graphs that these variants perform very similar to the HEFT variants. In the classic variant of the randomly generated benchmarks, HEFT produces an average speedup $\sim5\%$ greater than CEFT-HEFT-UP while HEFT-DOWN and CEFT-HEFT-DOWN produce very similar speedups with CEFT-HEFT-DOWN winning marginally for wider graphs. However, in the other variants of the workload we can see that the upward ranking function calculated with accurate computation and communication costs yields marginally better results than HEFT.\ifdefined\longversion Figure~\ref{fig:all-alpha-slr} shows the corresponding SLR values of all the algorithms.\fi

\section{Related Work}
\label{sec:paper-crit-path-rela-work}

In the past, the intractability of finding optimal solutions for the DAG scheduling problem has been well explored~\cite{kohler1975preliminary,michael1979computers,bruno1976computer}. As a result, efforts in the recent past have been focused on finding sub-optimal solutions in shorter runtimes, using heuristics. On the one hand, the heuristic solutions have been studied in terms of guided search space based techniques in~\cite{daoud2005gats,gao2008hybrid,sathappan2011modified,sanyal2005match,orsila2006parameterizing,pan2015improved} but are generally computationally intensive. On the other, list scheduling algorithms which are not as computationally expensive, produce results not far from the optimal~\cite{braun2001comparison}. In this section, we introduce key critical path based static list scheduling algorithms and methods of calculating the critical path. 

The idea of using Critical Paths in heuristics for scheduling DAGs has existed for a long time~\cite{kohler1975preliminary,hu1961parallel,lockyer1969introduction}. The conventional definition of the critical path as given in Definition~\ref{def:crit-path-defi}, is as follows : \textit{\textbf{Critical-Path} (CP) of a DAG is the longest path of from the entry node to the exit node in the application graph}. Existing algorithms to compute the critical path of a graph for heterogeneous machines make simplifying assumptions. As mentioned before, a simple strategy is to take execution times of a given task on various processors and average them \cite{kwok1996dynamic}. Another~\cite{topcuoglu2002performance}, is to assign all tasks on the critical path to a single processor, and to simply choose the processor that minimizes the critical path length. The latter approach also avoids having to consider communication costs because all tasks are assumed to be on the same processor. However, for some scenarios these algorithms perform poorly as we discuss in section~\ref{sec:paper-crit-path-expe-resu}. In general, the approach of calculating critical paths based on averages can give a result that is greater or lesser than the \textit{true} critical path. Also, adding a new processor can radically change the critical path which is not handled well by the average approach.

%
%

\ifdefined\longversion
In the past decade, there have been two \ifdefined\longversion\footnote{These algorithms are DAG scheduling algorithms that are based on the idea of critical path. Since the critical path is an integral part of their work, they define said critical path, hence making it relevant to the work presented in this paper. At this juncture, we would like to stress on the fact that our algorithm presented in the previous sections, is strictly a critical path finding algorithm which can be extended to form a DAG scheduling algorithm.}\fi main critical path based scheduling algorithms : the Dynamic Critical Path algorithm (DCP)~\cite{kwok1996dynamic} and the Critical Path On a Processor (CPOP)~\cite{topcuoglu2002performance}. Kwok et al. in 1996 developed the DCP algorithm which used the idea of critical-path to solve the problem of DAG scheduling. In their algorithm, the authors do not calculate the critical path of the application before scheduling. During the scheduling process, tasks from the graph can get dynamically added to the critical-path (CP). In order to distinguish the CP at an intermediate step in scheduling from the original CP, Kwok et al. term the CP at an intermediate step as the \textit{Dynamic Critical Path} (DCP). They then proceed to construct a theoretical basis by which they either remove nodes from the DCP or include nodes to it to monotonically reduce the schedule length. That is, at every consecutive step of the scheduling process, the intermediate schedule length or the DCPL remains the same or reduces. 

The CPOP algorithm borrows a lot of ideas from it's superior counterpart : HEFT~\cite{topcuoglu2002performance}. It differs from HEFT by redefining its ranking function. The rank of every task is calculated as the sum of its static upward rank and static downward rank. These two ranks signify the distance of the given task from the exit task and the source task respectively which are in turn used to calculate the critical path.

However, using this ranking function has its advantage. The rank score of the source task and the exit task are the same which is equal to the length of the critical-path. The critical path can then be easily found by traversing the graph depth-first and looking for tasks that have the same ranking score. Once the critical path is found, it is scheduled onto a single processor which minimizes $\sum_{t_i \in CP}^{}w(t_i, p_j), \forall p_j \in P$. This is the second biggest shortcoming of CPOP. By restricting that the tasks from the critical-path can be scheduled only onto one processor, the authors take away the ability to explore different assignments of the tasks in the CP to potentially obtain a lower/higher schedule length and hence a lower/higher makespan. Once the CP is scheduled, the processor selection phase proceeds as defined in HEFT. For each task $t_i$ that is not on the CP, processor $p_j$ is chosen that minimizes the earliest finish time of task $t_i$ on processor $p_j$.

\fi



\section{Conclusion}
\label{sec:conc}

In this paper, we have designed, implemented and tested a critical path finding algorithm (\textit{CEFT}) that finds the true critical path of an application for heterogeneous processors. The quality of the critical paths are shown to be better than those produced by the state of the art CPOP algorithm. We show that the critical path lengths produced by our algorithm is always at least as long as the ones produced by CPOP for the RGG-classic workload. Our experiments show that when the heterogeneity is better expressed in the workloads (RGG-high) our paths are shorter than CPOP's paths in 83.99\% of the experiments. 

We also extend our critical path finding algorithm into a DAG scheduling algorithm (\textit{CEFT-CPOP}) by replacing the path found by our algorithm (with its corresponding partial assignment) into the CPOP algorithm. We compare the efficacy of our algorithm mainly against CPOP through the use of makespan related comparison metrics like: schedule length ratio (SLR), speedup and slack. It is evident from the results that our algorithm outperforms CPOP even as a scheduling algorithm, in nearly all aspects. For the RGG-classic, RGG-low, RGG-medium and RGG-high, our algorithm (CEFT-CPOP) produces smaller makespans in 15.9\%, 75.94\%, 90.29\% and 89.69\% of the experiments respectively. We also consistently produce smaller SLR and slack values than CPOP. In some cases as explained in section~\ref{sec:paper-crit-path-resu-anal}, our algorithm outperforms HEFT in terms of SLR and makespans, but falls short of HEFT's capabilities to produce the tightest schedules (lowest slack values). We observe similarly consistent results from four real-world benchmarks: Fast Fourier Transform (FFT), Gaussian Elimination (GE), Molecular Dynamics code (MD) and the Epigenomics Workflow (EW).

One of the biggest impediment for CEFT-CPOP on the road to lower makespans, is that it has been extended to function as a scheduling algorithm using CPOP. Although, this helps us in finding relatively good schedules, we believe the extension provided by CPOP is still a limiting factor. We also believe that, by extending our algorithm into a full scheduling algorithm without task duplication, one can attain even better results in terms of obtaining smaller makespans.

\ifCLASSOPTIONcompsoc
  \section*{Acknowledgments}
\else
  \section*{Acknowledgment}
\fi


This work was supported by Science Foundation Ireland grant 12/IA/1381 and IRC Enterprise Partnership Scheme in collaboration with IBM Research, Ireland.

\ifCLASSOPTIONcaptionsoff
  \newpage
\fi

\bibliographystyle{IEEEtran}
\bibliography{main}
\begin{IEEEbiography}[{\includegraphics[width=1in,height=1.25in,clip,keepaspectratio]{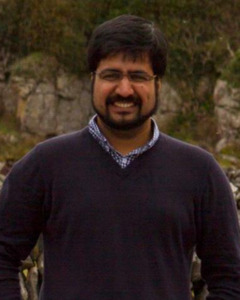}}]{Aravind Vasudevan}
is a Postdoctoral researcher in the Software Tools Group at Trinity College Dublin. He finished his PhD at Trinity College Dublin under the supervision of Dr. David Gregg.
\end{IEEEbiography}

\begin{IEEEbiography}[{\includegraphics[width=1in,height=1.25in,clip,keepaspectratio]{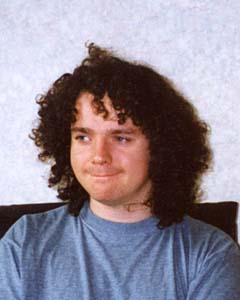}}]{David Gregg}
did an M.Sc in computer science in University College Dublin and PhD in Technical University of Vienna. He is currently a Professor in Computer Science at Trinity College Dublin and a Fellow of Trinity College Dublin.
\end{IEEEbiography}

\end{document}